\documentclass[aps,twocolumn,superscriptaddress,amsmath,amssymb, notitlepage,pra,raggedbottom,floatfix,numbers]{revtex4-1}
\usepackage[latin1]{inputenc}
\usepackage[shortlabels]{enumitem}
\usepackage{graphicx}
\usepackage{amsmath}
\usepackage{amsthm}
\usepackage{bm}
\usepackage{layout}
\usepackage{float}
\usepackage{amsfonts}
\usepackage{dsfont}
\usepackage{amssymb}
\usepackage[margin=0.75in]{geometry}
\usepackage{soul}
\usepackage{mathtools}
\usepackage[colorlinks=true,citecolor=blue]{hyperref}
\usepackage{color}
\usepackage[capitalise,compress]{cleveref}
\usepackage{csquotes}
\usepackage{svg}
\usepackage{pgfplots}
\pgfplotsset{compat=1.18}

\newtheorem{theorem}{Theorem}

\newtheorem{lemma}{Lemma}

\newtheorem{conjecture}{Conjecture}

\newcommand{\abs}[1]{\left | #1 \right|}
\renewcommand{\epsilon}{\varepsilon}

\def\ba{\begin{eqnarray}}
\def\ea{\end{eqnarray}}
\def\beq{\begin{equation}}
\def\eeq{\end{equation}}

\usepackage{braket}

\newcommand{\E}{\mathbb{E}}

\newcommand{\haf}{\mathrm{Haf}}

\usepackage{array}

\let\vec\mathbf

\usepackage{mathtools}
\DeclarePairedDelimiter\ceil{\lceil}{\rceil}

\DeclarePairedDelimiter\bkt{[}{]}				
\DeclarePairedDelimiter\paren{(}{)}

\begin{document}
\title{The Second Moment of Hafnians in Gaussian Boson Sampling}

\author{Adam Ehrenberg}
\affiliation{Joint Center for Quantum Information and Computer Science, NIST/University of Maryland College Park, Maryland 20742, USA}
\affiliation{Joint Quantum Institute, NIST/University of Maryland College Park, Maryland 20742, USA}

\author{Joseph T. Iosue}
\affiliation{Joint Center for Quantum Information and Computer Science, NIST/University of Maryland College Park, Maryland 20742, USA}
\affiliation{Joint Quantum Institute, NIST/University of Maryland College Park, Maryland 20742, USA}

\author{Abhinav Deshpande}
\affiliation{IBM Quantum, Almaden Research Center, San Jose, California 95120, USA}

\author{Dominik Hangleiter}
\affiliation{Joint Center for Quantum Information and Computer Science, NIST/University of Maryland College Park, Maryland 20742, USA}

\author{Alexey V. Gorshkov}
\affiliation{Joint Center for Quantum Information and Computer Science, NIST/University of Maryland College Park, Maryland 20742, USA}
\affiliation{Joint Quantum Institute, NIST/University of Maryland College Park, Maryland 20742, USA}
\date{\today}

\begin{abstract}
Gaussian Boson Sampling is a popular method for experimental demonstrations of quantum advantage, but many subtleties remain in fully understanding its theoretical underpinnings. 
An important component in the theoretical arguments for approximate average-case hardness of sampling is anticoncentration, which is a second-moment property of the output probabilities.
In Gaussian Boson Sampling these are given by hafnians of generalized circular orthogonal ensemble matrices. 
In a companion work [\href{https://arxiv.org/abs/2312.08433}{arXiv:2312.08433}], we develop a graph-theoretic method to study these moments and use it to identify a transition in anticoncentration.
In this work, we find a recursive expression for the second moment using these graph-theoretic techniques.
While we have not been able to solve this recursion by hand, we are able to solve it numerically exactly, which we do up to Fock sector $2n = 80$. 
We further derive new analytical results about the second moment. 
These results allow us to pinpoint the transition in anticoncentration and furthermore yield the expected linear cross-entropy benchmarking score for an ideal (error-free) device. 
\end{abstract}
\maketitle
\tableofcontents
\section{Introduction}\label{sec:introduction}
One of the major goals of quantum computer science is to find examples of certain tasks on which quantum devices can outperform classical computers. While the ultimate goal is to develop quantum computers that can run, say, Shor's algorithm \cite{shor_algorithms_1994}, the qubit numbers, gate fidelities, and error correction needed to accomplish such a task fault-tolerantly are well beyond the current state of the art. Therefore, there is interest in finding near-term examples of quantum advantage.

One area of focus that has strong theoretical evidence for an exponential speedup over the best possible classical algorithms comprises the so-called sampling problems. Aaronson and Arkhipov introduced one such promising framework called Boson Sampling \cite{aaronsonComputationalComplexityLinear2013}. The Boson Sampling task is to produce a sample (that is, a valid output Fock state) according to the outcome distribution generated by measuring indistinguishable photons that have been subjected to a random linear optical network of beam-splitters and phase shifters. In Boson Sampling, the input states consist of single photons on many input modes. However, because single-photon sources have imperfect efficiency,  
these states are difficult to produce experimentally, requiring an exponential amount of post-selection \cite{hangleiterComputationalAdvantageQuantum2023}. Therefore, generalizing this framework to other inputs that are more reliably produced has been an important topic of study.

Gaussian Boson Sampling represents one such popular generalization. There, the input states are quadratic, meaning they are generated from the vacuum by some combination of displacement and squeezing (assuming pure input states that have no thermal contribution) \cite{serafini2017quantum-continu}. Typically, the displacements are ignored because they do not contribute to entanglement between the modes. Hence, the input states are simply squeezed vacuum states, which are much easier to prepare in a lab than many parallel single-photon states \cite{hangleiterComputationalAdvantageQuantum2023}. Much theoretical work has been done to generalize the original statements from Ref.~\cite{aaronsonComputationalComplexityLinear2013} about the computational complexity of sampling in the Fock basis to this Gaussian setting~\cite{lundBosonSamplingGaussian2014,rahimi-keshariWhatCanQuantum2015,hamiltonGaussianBosonSampling2017,kruseDetailedStudyGaussian2019a,grierComplexityBipartiteGaussian2022,chabaudResourcesBosonicQuantum2023b,deshpandeQuantumComputationalAdvantage2022a}. In due course, many labs have performed experiments claiming to show quantum advantage using Gaussian Boson Sampling \cite{zhongQuantumComputationalAdvantage2020a, zhongPhaseProgrammableGaussianBoson2021a}.

Broadly speaking, the hardness of sampling schemes in general, and therefore of both Fock state and Gaussian Boson Sampling, is based on certain statistical properties of the output probability distributions. Fock state Boson Sampling and Gaussian Boson Sampling have output probabilities defined by permanents and hafnians, respectively, which are combinatorial functions mapping matrices over a field to an element of that field. If one treats the input matrix as a weighted adjacency matrix, then the permanent and the hafnian count the number of perfect matchings in the bipartite and generalized weighted graph, respectively, defined by this adjacency matrix \cite{barvinokCombinatoricsComplexityPartition2016}. 
These functions are, in general, difficult to compute. The permanent is \#\textsf{P}-hard to compute exactly \cite{valiant_complexity_1979}, and this hardness extends to the hafnian because one can encode the permanent of a matrix as the hafnian of a matrix that is twice as big. Even further, Ref.~\cite{aaronsonComputationalComplexityLinear2013} extended this exact hardness to a proof that it is \textsf{GapP}-hard to approximate the modulus squared of the permanent up to inverse polynomial multiplicative error (which similarly extends to the hafnian). However, showing that it is hard to compute or approximate specific output probabilities is not, in and of itself, enough to demonstrate hardness of actually producing a sample from the Fock or Gaussian Boson Sampling distributions; many theoretical tools are needed to show that a difficulty in computing probabilities further implies a difficulty in sampling.

One such crucial tool is called anticoncentration. Anticoncentration is a property of the output distribution that says, roughly, that the outputs are not too clustered on individual probabilities, hence making it more difficult to adequately mimic this distribution in a sampling procedure, and it is commonly used as evidence for approximate average-case hardness of sampling \cite{hangleiterComputationalAdvantageQuantum2023}. Anticoncentration is usually proven by analyzing the {moments} of the outcome probability distribution. 
In a companion piece to this work, Ref.~\cite{ehrenbergTransitionAnticoncentrationGaussian2023}, we study anticoncentration in the non-collisional limit (where the outcome states are very likely to have at most a single photon in each mode). We develop a graph-theoretic technique to find a closed form for the first moment and a few simple analytical results about the second moment; most saliently, we show that the second moment admits a polynomial expansion in the number of initially squeezed modes, and we derive the leading order in this expansion. These simple results are sufficient to show that there is actually a transition in whether or not anticoncentration holds based on how many of the initial modes are squeezed; when few are squeezed, there is a lack of anticoncentration, but, in the opposite limit, a weak version of anticoncentration holds. 

However, the second moment itself deserves a more thorough treatment beyond the few analytic results needed to prove this transition in anticoncentration. For example, linear cross-entropy benchmarking (LXEB) is a tool that has been used to characterize the performance of sampling experiments, most notably in the random circuit sampling experiment of Ref.~\cite{aruteQuantumSupremacyUsing2019c}. It can be shown that the LXEB score that an error-free sampler would achieve when averaged over all possible random networks is precisely given by the second moment of the output probabilities normalized by the square of the first moment. Therefore, a better understanding of the second moment is crucial to achieving a better understanding this popular benchmarking scheme. 

To that end, we develop a classically efficient recursion relation that allows us to exactly calculate the second moment up to any desired Fock sector $n$, which is the main technical contribution of this work. The recursion relation follows from the graph-theoretic approach we introduce in Ref.~\cite{ehrenbergTransitionAnticoncentrationGaussian2023}, which we generalize and expand upon here. This approach reduces the algebraic evaluation of the hafnian to simply counting the number of connected components of a certain class of graphs. We then carefully study how higher-order graphs reduce to lower-order ones under certain operations, and the effect that this has on the number of connected components, in order to recursively solve for the second moment. Not only does this allow us to make statements about the average LXEB score for an error-free sampler, but it also allows us to pin down more precisely \emph{where} the aforementioned transition in anticoncentration occurs. If $k$ is the number of initially squeezed modes, we provide strong evidence that this transition occurs at $k = \Theta(n^{2})$.

The rest of the paper proceeds as follows. In \cref{sec:setup}, we provide some background information, set up the system and problem of interest, and briefly summarize our main results. In \cref{sec:review}, we review our results from Ref.~\cite{ehrenbergTransitionAnticoncentrationGaussian2023}; specifically, in \cref{sec:first_moment}, we review results about the first moment, and in \cref{sec:second_moment}, we discuss how to calculate the second moment. This latter section sets up the discussion of the recursion in \cref{sec:recursion} (though most of the technical details are addressed in \cref{app:recursion_complexity,app:recursion}). \cref{sec:numerical_evaluation} discusses the actual exact numerical evaluation of the recursion. Complementing this, \cref{sec:scaling} discusses some preliminary analytical results and scaling properties of the second moment. Finally, in \cref{sec:anticoncentration}, we apply these results to give evidence for the exact location of the transition in anticoncentration we derive in Ref.~\cite{ehrenbergTransitionAnticoncentrationGaussian2023}. 

\section{The output distribution of Gaussian boson sampling}\label{sec:setup}

In this section, we provide some necessary background information on Gaussian Boson Sampling and set up our system of interest. We also motivate the study of the moments of the output probabilities. Finally, we provide a brief summary of our main results. 

\subsection{Gaussian boson sampling}
We consider a paradigmatic Gaussian Boson Sampling system on $m$ modes~\cite{hamiltonGaussianBosonSampling2017,kruseDetailedStudyGaussian2019a}. These modes pass through a random sequence of beamsplitters and phase shifters that effect a linear optical (i.e.~photon-number-conserving Gaussian) unitary $U \in \mathrm{U}(m)$ and are then measured in the Fock basis (this non-Gaussian operation is necessary for classical hardness of sampling \cite{chabaudResourcesBosonicQuantum2023b}). We consider the typical case where the initial state on the first $k$ modes consists of single-mode squeezed states of equal squeezing parameter $r$, and the remaining $m-k$ modes are initialized to the vacuum state. 

Reference~\cite{hamiltonGaussianBosonSampling2017} calculates the outcome probability of the Fock measurement of such a system. Given a unitary $U$, the probability of obtaining an outcome $\vec{n} = (n_{1},n_{2},\dots,n_{m})\in\mathbb{N}_{0}^{m}$ with total photon count $2n = \sum_{i=1}^{m}n_{i}$ is given by
\begin{equation}\label{eqn:gbs_probability}
    P_{U}(\vec{n}) = \frac{\tanh^{2n} r}{\cosh^{k} r}\abs{\haf(U^{\top}_{1_{k},\vec{n}}U_{1_{k},\vec{n}})}^{2}.
\end{equation}
$U_{1_{k},\vec{n}}$ is the $k \times 2n$ submatrix of $U$ corresponding to its first $k$ rows and its columns determined by the nonzero elements of $\vec{n}$ (appropriately repeated $n_{i}$ times). $\haf$ refers to the hafnian, which, for a $2n\times2n$ symmetric matrix $A$, is
\begin{equation}\label{eqn:Hafnian}
    \haf(A) = \frac{1}{n! 2^{n}}\sum_{\sigma \in S_{2n}}\prod_{j = 1}^{n}A_{\sigma(2j-1),\sigma(2j)},
\end{equation}
with $S_{2n}$ the permutation group on $2n$ elements. We specify that the dimensions of $A$ are even because the hafnian of an odd matrix vanishes; it also vanishes if the input matrix is not symmetric. In our setting, this aligns with the physical fact that single-mode squeezed vacuum states are supported only on even Fock states. The hafnian generalizes the permanent (whose computational complexity controls the hardness of Fock state Boson Sampling) because one can prove that \cite{hamiltonGaussianBosonSampling2017}
\begin{equation}
    \mathrm{Per}(A) = \haf\bkt*{
    \begin{pmatrix}
        0 & A \\
        A^{\top} & 0
    \end{pmatrix}
    }.
\end{equation}
Hence, computing the hafnian is at least as hard as computing the permanent. 

We work in the regime where the measured output states are, with high probability, photon-collision-free, which means that the output vector $\vec{n}$ has $n_{i} \in \{0,1\}$. That is, $U_{1_{k},\vec{n}}$ has no repeated columns. It suffices for $\E[2n] = k \sinh^{2}r = o(\sqrt{m})$ for photon-collision-freeness to hold with high probability. When $n = o(\sqrt{m})$, Ref.~\cite{deshpandeQuantumComputationalAdvantage2022a} provides strong numerical and theoretical evidence that the distribution of submatrices $U_{1_{k},\vec{n}}$ is well-captured by a generalization of the circular orthogonal ensemble (COE):

\begin{conjecture}[Hiding~\cite{deshpandeQuantumComputationalAdvantage2022a}]\label{con:gaussian}
For any $k$ such that $1 \leq k \leq m$ and $2n=o(\sqrt{m})$, the distribution of the symmetric product $U^{\top}_{1_{k},\vec{n}}U_{1_{k},\vec{n}}$ of submatrices of a Haar-random $U \in \mathrm{U}(m)$ closely approximates in total variation distance the distribution of the symmetric product $X^{\top}X$ of a complex Gaussian matrix $X~\sim ~\mathcal{N}(0,1/ m)_{c}^{k \times 2n}$ with mean $0$ and variance $1/m$.
\end{conjecture}
We note that, in Ref.~\cite{deshpandeQuantumComputationalAdvantage2022a}, this conjecture is only formulated for the case $n \leq k \leq m$. However, here we allow $k$ to reach $1$. The reasoning is that the evidence for Conjecture~\ref{con:gaussian} in the regime $k = n$ is based on a proof from Ref.~\cite{aaronsonComputationalComplexityLinear2013} showing that $n \times n$ submatrices of Haar-random unitaries are approximately Gaussian. Clearly the proof must still hold in the case $k < n$ (if $n \times n$ submatrices are approximately Gaussian, then so too are smaller submatrices), 
meaning we can safely extend the conjecture to all $k \leq m$.

Roughly speaking, the intuition behind the conjecture and the original proof of the $k = n$ regime in Ref.~\cite{aaronsonComputationalComplexityLinear2013} is that, if one looks at a small enough submatrix of a unitary, this submatrix no longer ``notices'' the unitary constraints. Multiplying this small submatrix by its transpose washes out the remaining correlations between elements of the unitary.  Hence, the product of the submatrices is approximately the same as a product of i.i.d.~Gaussian matrices. Observe also that working in the non-collisional regime, $n \in o(\sqrt{m})$, is crucial for this argument to hold; an output state with more than one photon in a given mode leads to a repeated column/row in the respective submatrix, which, of course, destroys the independence of these elements. In what follows, we work under the assumption that Conjecture~\ref{con:gaussian} holds. We are therefore interested in the statistical properties of $X^{\top}X$ when the elements of $X$ are i.i.d.~Gaussian.

\subsection{Moments of the Gaussian Boson Sampling distribution and their significance}
In order to understand the statistical properties of the outcome probabilities of Gaussian Boson Sampling, we must study not just the distribution over individual matrix elements of $X^{\top}X$, but how they interact with one another through the hafnian. Under Conjecture~\ref{con:gaussian} and \cref{eqn:gbs_probability}, the outcome probabilities of Gaussian Boson Sampling are given by (up to a prefactor that is mostly irrelevant for our purposes)
\begin{equation}\label{eqn:moment_def}
    M_t(k,n) \coloneqq \E_{X \sim \mathcal G^{k \times 2n}}[|\haf(X^\top X)|^{2t}],
\end{equation}
where we use $\mathcal{G}^{k\times 2n}$ as shorthand for $\mathcal{N}(0,1)_{c}^{k\times 2n}$ (we consider unit variance for computational simplicity; rescaling $X$ by $1/\sqrt m$ leads to another overall prefactor that can be dealt with independently). Specifically, we are most interested in the first and second moments, $t = 1$ and $t=2$, respectively. We motivate this interest in two ways: the study of anticoncentration and linear cross entropy benchmarking in Gaussian Boson Sampling.

We first recall the framework for anticoncentration established in Ref.~\cite{ehrenbergTransitionAnticoncentrationGaussian2023}. There, the key definition is $p_{2}$, the inverse average collision probability in the output, which, under the hiding conjecture (Conjecture~\ref{con:gaussian}), is approximately given by the ratio of the square of the first moment to the second moment:
\begin{align}\label{eqn:anticoncentration_conversion}
     p_2(\mathrm U(m)) = \frac{\E_{U \in \mathrm U(m)}[P_U(\vec n)]^2}{\E_{U \in \mathrm U(m)}[P_U(\vec n)^2]} \approx \frac{M_1(k,n)^2}{M_2(k,n)} \eqqcolon m_2(k,n).
\end{align} 
We refer to $m_{2}(k,n)$ as the inverse normalized second moment. 
Reference~\cite{ehrenbergTransitionAnticoncentrationGaussian2023} uses $p_{2}$ to define three different classes of anticoncentration:
\begin{itemize}
     \item[(A)] 
We say that $P_U, U \in \mathrm U(m),$ \emph{anticoncentrates} if $p_2 = \Omega(1)$; 

\item[(WA)] We say that $P_U$ \emph{anticoncentrates weakly} if $p_2 = \Omega(1/n^a)$ for some $a = O(1)$;

\item[(NA)] We say that $P_{U}$ \emph{does not anticoncentrate} if $p_2 = O(1/n^{a})$ for any constant $a > 0$. 
 \end{itemize} 
Reference~\cite{ehrenbergTransitionAnticoncentrationGaussian2023} (especially Section~S$5$ in the Supplementary Material) contextualizes these definitions in relation to the approximate average-case hardness necessary for formal hardness of Gaussian Boson Sampling.

We note also that, of course, it is important how precise this approximation in \cref{eqn:anticoncentration_conversion} really is. That is, exactly how close in total variation distance the exact and approximate distributions are is important to formalizing the complexity theoretic implications of our work. In particular, if the distribution $U_{1_{k},\vec{n}}^{\top}U_{1_{k},\vec{n}}$ is not close enough in total variation distance to the distribution $X^{\top}X$, then it is not possible to transfer statements about, say, anticoncentration between the two distributions. We address this subtlety in the Supplemental Material of the companion work Ref.~\cite{ehrenbergTransitionAnticoncentrationGaussian2023}, but, in short, we can formalize and sharpen Conjecture~\ref{con:gaussian} such that statements made about anticoncentration of the approximate distribution via $m_{2}$ imply anticoncentration of the exact distribution via $p_{2}$ as well. 

Beyond understanding anticoncentration, calculations of $M_{1}(k,n)$ and $M_{2}(k,n)$ also allow one to study linear cross-entropy benchmarking in Gaussian Boson Sampling. Recall that linear cross-entropy benchmarking is a method by which one can compare the outputs of a potentially noisy Gaussian Boson Sampling experiment with the output of a perfect, error-free experiment. Cross-entropy benchmarking was introduced in the context of random circuit sampling in Refs.~\cite{neill_blueprint_2018-1,boixo_characterizing_2018} and later linearized in Ref.~\cite{aruteQuantumSupremacyUsing2019c}. We review this linearized form now, translating from the random circuit sampling language to that of bosonic sampling. 

Let $\{\vec{n}\}$ be the possible output photon strings sampled in some Gaussian Boson Sampling experiment that are produced with respective experimental probabilities $\tilde{P}_{U}(\vec{n})$. Let $P_{U}(\vec{n})$ be the ideal probabilities for these outputs; that is, these are the probabilities for an output $\vec{n}$ given by \cref{eqn:gbs_probability}. The linear cross-entropy score $F_{\textrm{XEB}}$ for such an experiment is
\begin{equation}
     F_{\textrm{XEB}} = |\Omega_{2n}| \sum_{\vec{n}\in\Omega_{2n}}  P_{U}(\vec{n})\tilde{P}_{U}(\vec{n})-1,
\end{equation}
where $\Omega_{2n}$ is the non-collisional sample space with $2n$ output photons in $m$ modes. If the noisy outputs are correct, i.e.~the experiment is error-free, then $\tilde{P}(\vec{n}_{i}) = P(\vec{n})$. The ideal cross-entropy score, then, is
\begin{equation}
    F^{\textrm{ideal}}_{\textrm{XEB}} = |\Omega_{2n}|\sum_{\vec{n}\in\Omega_{2n}}P_{U}(\vec{n})^{2} - 1.
\end{equation}
The expected value of the ideal cross-entropy over all possible unitaries is, therefore, 
\begin{equation}
    \mathbb{E}_{U\in U(m)}[F^{\textrm{ideal}}_{\textrm{XEB}}] = |\Omega_{2n}|\sum_{\vec{n}\in\Omega_{2n}}\mathbb{E}_{U\in U(m)}[P_{U}(\vec{n})^{2}] - 1.
\end{equation}
Assuming that one operates in the hiding regime, then two facts are true: first, $|\Omega_{2n}| \sim M_{1}(k,n)$; second, $\mathbb{E}_{U\in U(m)}[P_{U}(\vec{n})^{2}]$ is independent of $\vec{n}$ (see Ref.~\cite{ehrenbergTransitionAnticoncentrationGaussian2023} for more details). Therefore, 
\begin{equation}
    \mathbb{E}_{U\in U(m)}[F^{\textrm{ideal}}_{\textrm{XEB}}] = \frac{M_{2}(k,n)}{M_{1}^{2}(k,n)} - 1 = m_{2}(k,n)^{-1}-1.
\end{equation}
Thus, anticoncentration and the expected ideal linear cross-entropy benchmarking score both depend on this inverse average collision probability. Therefore, a precise calculation of the second moment beyond asymptotics is valuable to a more fine-grained understanding of both anticoncentration and cross-entropy benchmarking.

\subsection{Summary of Results}
We now come to a brief summary of our main results. 

In Ref.~\cite{ehrenbergTransitionAnticoncentrationGaussian2023}, we develop a graph-theoretic formalism that allows us to derive various analytic properties of the first and second moments, $M_{1}(k,n)$ and $M_{2}(k,n)$. We use this formalism to find a closed form expression for $M_{1}(k,n)$ and to show that $M_{2}(k,n)$ admits a polynomial expansion in $k$; we also calculate the leading order of this expansion. This allows us to show the transition in anticoncentration. We review these results in more depth in \cref{sec:review}. 

In this work, we significantly expand upon this graph-theoretic formalism and derive an efficiently evaluable recursion relation that allows us to numerically exactly calculate all coefficients of the polynomial expansion of the second moment. We then apply this algorithm and calculate these expansions up to photon sector $2n=80$. In the photon-non-collisional regime, where $n \in o(\sqrt{m})$, this corresponds to approximately $6400$ modes, which is well beyond the current state-of-the-art experiments. Therefore, the technique that we develop in this work yields results that can help characterize the output distribution of any near-term Gaussian Boson Sampling experiment. The recursion is developed in \cref{sec:recursion}, with details about its efficiency and construction deferred to Appendices~\ref{app:recursion_complexity} and~\ref{app:recursion}, respectively. 

We then discuss some simple analytic results about the scaling of the second moment in \cref{sec:numerical_evaluation}. We follow this with substantial numerical investigation of the results of the recursion up to $2n = 80$ in \cref{sec:anticoncentration}. In particular, we are able to give strong evidence that the transition in anticoncentration occurs at $k = \Theta(n^{2})$. We accomplish this with numerical plots of $m_{2}(k,n)$, the quantity that controls anticoncentration, when $k$ scales polynomially with $n$. We also provide a brief analytic argument that this transition occurs somewhere between $k = \Omega(n)$ and $k = O(n^{2})$. 

This result, along with the fact that we operate in the conjectured hiding regime where $2n = o(\sqrt{m})$ and $k \leq m$, implies concrete advice for experimental demonstrations of quantum advantage via Gaussian Boson Sampling. Namely, one should squeeze all $m$ modes with squeezing parameter $\sinh^{2}r = o(m^{-1/2})$. 

\section{Graph-theoretical analysis of Gaussian Boson Sampling moments}\label{sec:review}

In this section, we lay out the graph-theoretic framework for analyzing the moments of Gaussian Boson Sampling output probabilities. This is a review of the same framework we develop in Ref.~\cite{ehrenbergTransitionAnticoncentrationGaussian2023}. We first briefly recall the derivation of the closed form of the first moment $M_{1}(k,n)$, and we follow this with a discussion of how an extension of this framework also allows us to analyze the second moment $M_{2}(k,n)$.

\subsection{First Moment}\label{sec:first_moment}
In this section, we discuss the first moment of the output probabilities, which is, up to some multiplicative factors, $\underset{X\sim \mathcal{G}^{k \times 2n}}{\E}\bkt*{\abs{\haf(X^{\top}X)}^{2}}$.
We calculate and analyze this moment in Ref.~\cite{ehrenbergTransitionAnticoncentrationGaussian2023}, but we review the key elements of that discussion because they are a useful point of reference for the calculation of the second moment.

Using the definition of the hafnian and properties of the expectation value of complex Gaussians, we reduce the first moment to a sum over Kronecker $\delta$s:
\begin{equation}\label{eqn:first_moment_delta}
 M_1(k,n) =  \frac{(2n)!}{(2^{n}n!)^{2}}\sum_{\tau \in S_{2n}}\sum_{\{o_{i}\}_{i=1}^{n}}^{k}\prod_{j=1}^{n}\delta_{o_{\ceil*{\frac{\tau(2j-1)}{2}}}o_{\ceil*{\frac{\tau(2j)}{2}}}}. 
\end{equation}
We ascribe a graph-theoretic interpretation to this equation; see \cref{fig:first_moment_graph} for an example. Each permutation $\tau$ instantiates a graph $G_{\tau}$ on $2n$ vertices labeled $O_{1}$ to $O_{2n}$ with edges defined by two perfect matchings: one fixed black set of edges, and one set of red edges determined by $\tau$. More specifically, each index $o_{j}$ in the sum splits into two vertices $O_{\ell}$ and $O_{\ell'}$ such that $\ceil{\tau(\ell)/2} = j = \ceil{\tau(\ell')/2}$ (that is, $o_{\lceil\tau(\ell)/2\rceil}$ maps to a vertex $O_{\ell}$). One perfect matching consists of black edges between $O_{2j-1}$ and $O_{2j}$ for all $j \in [n] := \{1, 2, \dots, n\}$; these edges enforce that $o_{\lceil\tau(2j-1)/2\rceil}$ and $o_{\lceil\tau(2j)/2\rceil}$ are linked by a Kronecker $\delta$. The other perfect matching has red edges between $O_{\ell}$ and $O_{\ell'}$ if $\ceil{\tau(\ell)/2}=\ceil{\tau(\ell')/2}$; these edges ensure that there is an edge between the $\ell, \ell'$ mapped to the same value under $\tau$ and the ceiling function, meaning the vertices arose from the same lower-case-$o$ index. 
\begin{figure}[ht!]
    \centering
    \includegraphics[width=\linewidth]{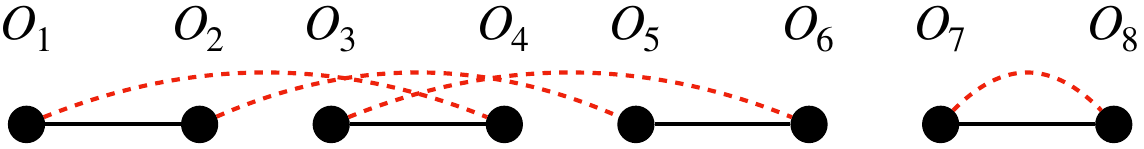}
    \caption{Graph $G \in \mathbb{G}_{n}^{1}$. One of $2^{n}n!$ permutations that induces this graph is $\tau = \bigl(\begin{smallmatrix}
    1 & 2 & 3 & 4 & 5 & 6 & 7 & 8 \\
    1 & 3 & 5 & 2 & 4 & 6 & 8 & 7
  \end{smallmatrix}\bigr).$ This graph has two connected components, therefore contributing $k^{2}$ to the first moment.}
    \label{fig:first_moment_graph}
\end{figure}

This definition of $G_{\tau}$ ensures that the number of  connected components of $G_{\tau}$, $C(G_{\tau})$, is equivalent to the number of unconstrained indices in the interior sum in \cref{eqn:first_moment_delta}, and, hence, the number of factors of $k$ that $\tau$ contributes overall. Therefore, 
\begin{equation}\label{eqn:first_moment_G1}
     M_1(k,n) = \frac{(2n)!}{\paren*{2^{n}n!}^{2}} \sum_{\tau\in S_{2n}}k^{C(G_{\tau})}.
\end{equation}
We simplify this expression using a degeneracy whereby $2^{n}n!$ different $\tau$ all induce the same final graph; the factor of $n!$ corresponds to choosing which tuple $(2j-1, 2j)$ corresponds to which index $\ceil{\tau(\ell)/2}=\ceil{\tau(\ell')/2}$, and the factor of $2^{n}$ comes from ordering within each tuple. Therefore, we study only these final sets of graphs, which we label $\mathbb{G}_{n}^{1}$ ($1$ refers to the first moment, and $n$ indexes the order). We study the connected components of graphs in $\mathbb{G}_{n}^{1}$ by writing down a recursion relation in $n$ and $k$ that, when solved, yields the first theorem of Ref.~\cite{ehrenbergTransitionAnticoncentrationGaussian2023}:
\begin{theorem}[Ref.~\cite{ehrenbergTransitionAnticoncentrationGaussian2023}]\label{thm:first_moment}
The sum over graphs in $\mathbb{G}^{1}_{n}$ satisfies 
\begin{equation}\label{eqn:g_first_moment}
   \sum_{G\in\mathbb{G}^{1}_{n}}k^{C(G)} = k(k+2)\dots (k+2n-2),
\end{equation}
and hence $M_1(k,n) = (2n-1)!!(k+2n-2)!!/(k-2)!!$.
\end{theorem}

To summarize: \cref{eqn:first_moment_delta} gives an expression for the first moment of the outcomes of Gaussian Boson Sampling probabilities in terms of sums of products of Kronecker $\delta$s. We then reinterpret this as counting the number of connected components of a certain type of graph with two perfect matchings. We solve this counting problem by developing and evaluating a recursion relation. We use the same overall technique to calculate the second moment, as we explain in the next section. 

\subsection{Second moment}
\label{sec:second_moment}
We now move on to analyzing the second moment of the output probabilities. Using similar techniques as described for the first moment, in Ref.~\cite{ehrenbergTransitionAnticoncentrationGaussian2023} we derive an expression for the second moment that is equivalent to \cref{eqn:first_moment_delta}:
\begin{widetext}
\begin{multline}\label{eqn:second_moment_delta}
      M_2(k,n) \coloneqq \underset{X\sim \mathcal{G}^{k \times 2n}}{\E}\bkt*{\abs{\haf(X^{\top}X)}^{4}} = \paren*{\frac{1}{2^{n}n!}}^{4} (2n)! \sum_{\tau,\alpha,\beta \in S_{2n}}\sum_{\{\ell_{i},o_{i},p_{i}\}_{i=1}^{n}=1}^{k}\Biggr[\prod_{j=1}^{n} \\
    \Biggr(\delta_{o_{\ceil*{\frac{\tau(2j-1)}{2}}}o_{\ceil*{\frac{\tau(2j)}{2}}}}\delta_{p_{\ceil*{\frac{\alpha(2j-1)}{2}}}q_{\ceil*{\frac{\beta(2j-1)}{2}}}}\delta_{p_{\ceil*{\frac{\alpha(2j)}{2}}}q_{\ceil*{\frac{\beta(2j)}{2}}}} + \delta_{o_{\ceil*{\frac{\tau(2j-1)}{2}}}q_{\ceil*{\frac{\beta(2j)}{2}}}}\delta_{p_{\ceil*{\frac{\alpha(2j-1)}{2}}}q_{\ceil*{\frac{\beta(2j-1)}{2}}}}\delta_{p_{\ceil*{\frac{\alpha(2j)}{2}}}o_{\ceil*{\frac{\tau(2j)}{2}}}}
    +\\
   \delta_{q_{\ceil*{\frac{\beta(2j-1)}{2}}}o_{\ceil*{\frac{\tau(2j)}{2}}}}\delta_{p_{\ceil*{\frac{\alpha(2j-1)}{2}}}o_{\ceil*{\frac{\tau(2j-1)}{2}}}}\delta_{p_{\ceil*{\frac{\alpha(2j)}{2}}}q_{\ceil*{\frac{\beta(2j)}{2}}}} + \delta_{q_{\ceil*{\frac{\beta(2j-1)}{2}}}q_{\ceil*{\frac{\beta(2j)}{2}}}}\delta_{p_{\ceil*{\frac{\alpha(2j-1)}{2}}}o_{\ceil*{\frac{\tau(2j-1)}{2}}}}\delta_{p_{\ceil*{\frac{\alpha(2j)}{2}}}o_{\ceil*{\frac{\tau(2j)}{2}}}}\Biggr)
    \Biggr].
\end{multline}.
\end{widetext}
The main differences between \cref{eqn:second_moment_delta} and \cref{eqn:first_moment_delta} are threefold: 
\begin{enumerate}
    \item We sum over three permutations (instead of a single one) labeled $\tau, \alpha, \beta$;
    \item There are now $3n$ indices to sum over,  $\{o_{i}, q_{i}, p_{i}\}_{i=1}^{n}$, instead of just the $n$ given by $\{o_{i}\}_{i=1}^{n}$;
    \item Each factor is a sum of four possible terms instead of just one.
\end{enumerate}
However, this expression still possesses a natural graph-theoretic interpretation, as we now review. See \cref{fig:second_moment_graph} for an example graph as a guide to the following discussion.  
\begin{figure}
    \centering
    \includegraphics[width=\linewidth]{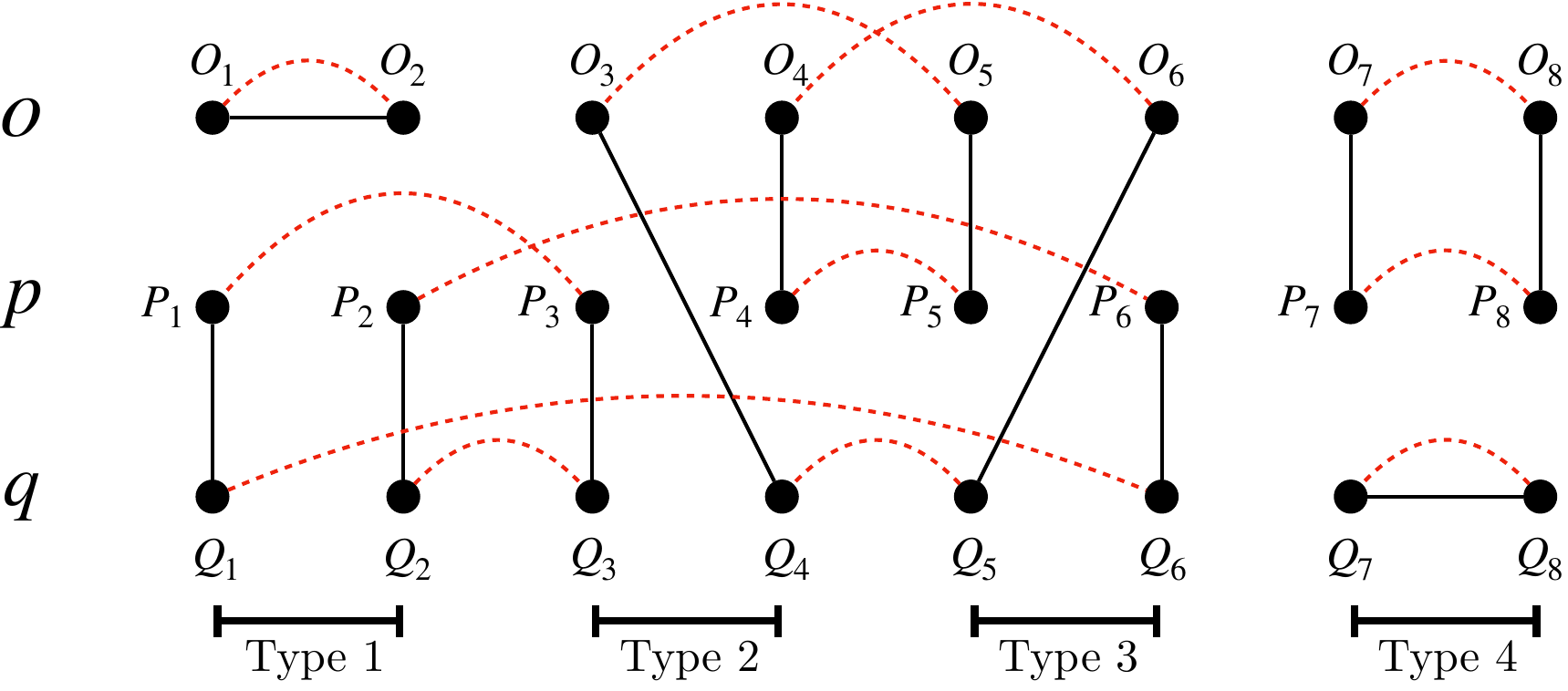}
    \caption{Example graph on $n = 4$ used in the calculation of the second moment. Each of the four possible sets of black edges are shown. An example of three permutations that would induce this graph is: $\tau = \bigl(\begin{smallmatrix}
    1 & 2 & 3 & 4 & 5 & 6 & 7 & 8 \\
    1 & 2 & 4 & 5 & 3 & 6 & 8 & 7
  \end{smallmatrix}\bigr)$, $\alpha = \bigl(\begin{smallmatrix}
    1 & 2 & 3 & 4 & 5 & 6 & 7 & 8 \\
    8 & 6 & 7 & 3 & 4 & 5 & 1 & 2
  \end{smallmatrix}\bigr)$, and $\beta = \bigl(\begin{smallmatrix}
    1 & 2 & 3 & 4 & 5 & 6 & 7 & 8 \\
    8 & 5 & 6 & 2 & 1 & 7 & 3 & 4
  \end{smallmatrix}\bigr)$. This graph has $5$ connected components, so it contributes $k^{5}$ to the second moment.}
    \label{fig:second_moment_graph}
\end{figure}

Each index in $\{o_{i}, q_{i}, p_{i}\}_{i=1}^{n}$ is again split into two graph vertices $\{O_{i}, Q_{i}, P_{i}\}_{i=1}^{2n}$ that are placed into $2n$ columns and three rows labeled $o$, $p$, and $q$, respectively. 
As for the first moment, we define two perfect matchings on these vertices given by black and red edges.
The black edges are between vertices whose labels are linked under the Kronecker $\delta$s, and the red edges connect graph vertices that came from the same original summation index.

More specifically, consider fixing a set of three permutations $\tau, \alpha, \beta$.
There is a red edge between $O_{\ell}$ and $O_{\ell'}$ if $\ceil{\tau(\ell)/2}=\ceil{\tau(\ell')/2}$.
An analogous statement holds for $P$ and $Q$ vertices, though one uses permutations $\alpha$ and $\beta$, respectively, instead of $\tau$. 
Note that this implies that red edges are always contained within a single row. 
Now, the black edges are slightly more complicated. There is only a single Kronecker $\delta$ term in each factor in the product \cref{eqn:first_moment_delta}, meaning there is only a single set of black edges for the graphs in $\mathbb{G}_{n}^{1}$. However, because the second moment as expressed in \cref{eqn:second_moment_delta} contains factors with four Kronecker $\delta$ terms, each value of $j \in [n]$ can lead to one of four different patterns of black edges on columns $2j-1$ and $2j$. 
We refer to these patterns of black edges on a single pair of adjacent columns as type-1, type-2, type-3, and type-4; see \cref{fig:second_moment_graph} for an example graph that has one of each type.
The Kronecker $\delta$ terms and their corresponding black edges, listed in order from type-1 to type-4, are given by 
\begin{widetext}
\begin{align}
   \delta_{o_{\ceil*{\frac{\tau(2j-1)}{2}}}o_{\ceil*{\frac{\tau(2j)}{2}}}}\delta_{p_{\ceil*{\frac{\alpha(2j-1)}{2}}}q_{\ceil*{\frac{\beta(2j-1)}{2}}}}\delta_{p_{\ceil*{\frac{\alpha(2j)}{2}}}q_{\ceil*{\frac{\beta(2j)}{2}}}} & \to \{(O_{2j-1}, O_{2j}), (P_{2j-1},Q_{2j-1}),(P_{2j},Q_{2j})\}, \\
    \delta_{o_{\ceil*{\frac{\tau(2j-1)}{2}}}q_{\ceil*{\frac{\beta(2j)}{2}}}}\delta_{p_{\ceil*{\frac{\alpha(2j-1)}{2}}}q_{\ceil*{\frac{\beta(2j-1)}{2}}}}\delta_{p_{\ceil*{\frac{\alpha(2j)}{2}}}o_{\ceil*{\frac{\tau(2j)}{2}}}} & \to \{(O_{2j-1}, Q_{2j}), (P_{2j-1},Q_{2j-1}),(O_{2j},P_{2j})\}, \\
     \delta_{q_{\ceil*{\frac{\beta(2j-1)}{2}}}o_{\ceil*{\frac{\tau(2j)}{2}}}}\delta_{p_{\ceil*{\frac{\alpha(2j-1)}{2}}}o_{\ceil*{\frac{\tau(2j-1)}{2}}}}\delta_{p_{\ceil*{\frac{\alpha(2j)}{2}}}q_{\ceil*{\frac{\beta(2j)}{2}}}} &\to  \{(O_{2j}, Q_{2j-1}), (P_{2j-1},O_{2j-1}),(P_{2j},Q_{2j})\}, \\
     \delta_{q_{\ceil*{\frac{\beta(2j-1)}{2}}}q_{\ceil*{\frac{\beta(2j)}{2}}}}\delta_{p_{\ceil*{\frac{\alpha(2j-1)}{2}}}o_{\ceil*{\frac{\tau(2j-1)}{2}}}}\delta_{p_{\ceil*{\frac{\alpha(2j)}{2}}}o_{\ceil*{\frac{\tau(2j)}{2}}}} &\to \{(O_{2j-1}, P_{2j-1}), (O_{2j},P_{2j}),(Q_{2j-1},Q_{2j})\}.
\end{align}
\end{widetext}

Because there are four patterns of black edges per pair of adjacent columns, and $n$ such pairs, there are $4^{n}$ possible arrangements of black edges on the entire graph. We label these arrangements by an integer $z \in [4^{n}]$, and we label a graph as $G_{\tau,\alpha,\beta}(z)$. 

Analogously to the first moment, we can rewrite the sum over products of Kronecker $\delta$s in \cref{eqn:second_moment_delta} as a sum over these graphs, where $G_{\tau,\alpha,\beta}(z)$ contributes a factor of $k$ raised to its number of connected components. Therefore, \cref{eqn:second_moment_delta} becomes
\begin{equation}\label{eqn:second_moment_intermediate}
    M_{2}(k,n) = \frac{(2n)!}{\paren*{2^{n}n!}^{4}}\sum_{\underset{z\in[4^{n}]}{\tau,\alpha,\beta \in S_{2n}}}k^{C(G_{\tau,\alpha,\beta}(z))}.
\end{equation}
There is again a degeneracy where many permutations all lead to the same set of red edges in a given row, and, hence, the same graph. Specifically, this degeneracy is again $2^{n}n!$, but for each copy of $S_{2n}$. We can therefore again ignore the permutations and look only at the underlying graphs. For any given $z$, we define $\mathbb{G}_{n}^{2}(z)$ to be the graphs on $6n$ vertices with two perfect matchings: the $z$th set of black edges and red edges that pair vertices in the same row. We then define $\mathbb{G}_{n}^{2} = \cup_{z\in[4^{n}]}\mathbb{G}_{n}^{2}(z)$. Thus, accounting for the described degeneracy and these definitions, we get
\begin{equation}\label{eqn:second_moment_graph}
    M_{2}(k,n) = (2n-1)!!\sum_{G\in\mathbb{G}_{n}^{2}}k^{C(G)}.
\end{equation}

This implies the following theorem. 
\begin{theorem}[Ref.~\cite{ehrenbergTransitionAnticoncentrationGaussian2023}]\label{thm:second_moment}
The second moment $M_2(k,n)$ is a degree-$2n$ polynomial in $k$ and can be written as $ M_2(k,n) = (2n-1)!! \sum_{i = 1}^{2n} c_{i}k^{i}$, where $c_{i}$ is the number of graphs $G\in \mathbb{G}^{2}_{n}$ that have $i$ connected components. 
\end{theorem}

Our goal, then, is to determine these coefficients $c_{i}$. 
It is possible to directly compute $c_{2n}$ and $c_{2n-1}$, that is, the number of graphs $G \in \mathbb{G}_{n}^{2}$ with $2n$ or $2n-1$ connected components, respectively (see Appendix~\ref{app:individual_coefficients}).
However, these calculations do not easily generalize to the other $c_{i}$.
Therefore, we take a different approach, which is to derive a recursion relation that is similar in spirit to the one we use to compute the first moment.

\section{Recursion for the second moment}\label{sec:recursion}

We now move on to the recursion relation that builds the $c_{i}$ for larger $n$ from those of smaller $n$. It is useful to refer to \cref{fig:recursion_fail} for the following discussion. We are interested in the connected components of the graphs in $\mathbb{G}_{n}^{2}$, and the number of connected components does not change if one takes a graph and then ``collapses'' vertices that are connected via an edge into a single larger vertex. The graphs that we have defined for the second moment are composed of $2n$ columns of $3$ vertices each. Therefore, if one performs this collapsing operation on all of the vertices in, say, the first two columns, this converts a graph with $2n$ columns into one with $2n-2$ columns. Let us refer to these first two columns as $\mathbb{C}_{1,2}$; that is, $\mathbb{C}_{1,2} = \{O_{1}, O_{2}, P_{1}, P_{2}, Q_{1}, Q_{2}\}$. Two facts follow from the approach we have just described: (1) there are only a finite number of ways that the two columns can connect into the rest of the graph, (2) if one ``integrates out'' $\mathbb{C}_{1,2}$ by collapsing all of the vertices, one can write the number of connected components of the original graph as the sum of the remaining connected components plus the number of connected components contained entirely within $\mathbb{C}_{1,2}$. 
This is a generalization of the approach used to prove \cref{thm:first_moment}. 

However, this recursion is substantially more complicated than the one we use to calculate the first moment, as illustrated in  \cref{fig:recursion_fail}. In particular, we must generalize the types of graphs that we consider in order to build a recursion that ``closes'' on itself, that is, to build a recursion that consistently produces valid graphs. Consider the graphs that we have described so far in the context of this ``integration'' procedure whereby sets of vertices are collapsed onto one another. As stated, this procedure can induce a graph with red edges that cross between rows, which is not allowed in our current formulation. In  \cref{fig:recursion_fail}, the first figure shows a graph in $\mathbb{G}_{4}^{2}$ where $\mathbb{C}_{1,2}$ is integrated out, as denoted by the hashing, and the second figure depicts the consequence of this integration. Consider the path $P_{3}\mbox{---} P_{1}\mbox{---}Q_{1}\mbox{---}Q_{6}$ that passes through the first column. Collapsing the vertices $P_{1}$ and $Q_{1}$ into $P_{3}$ and $Q_{3}$, respectively, does not change the number of connected components, but it induces an edge $P_{3}\mbox{---}Q_{6}$ that is heretofore unallowed because it crosses between rows $2$ and $3$. Therefore, the newly induced graph is not an element of $\mathbb{G}_{3}^{2}$, hence why we must generalize what kinds of graphs we consider. 
\begin{figure}
    \centering
    \includegraphics[width=\linewidth]{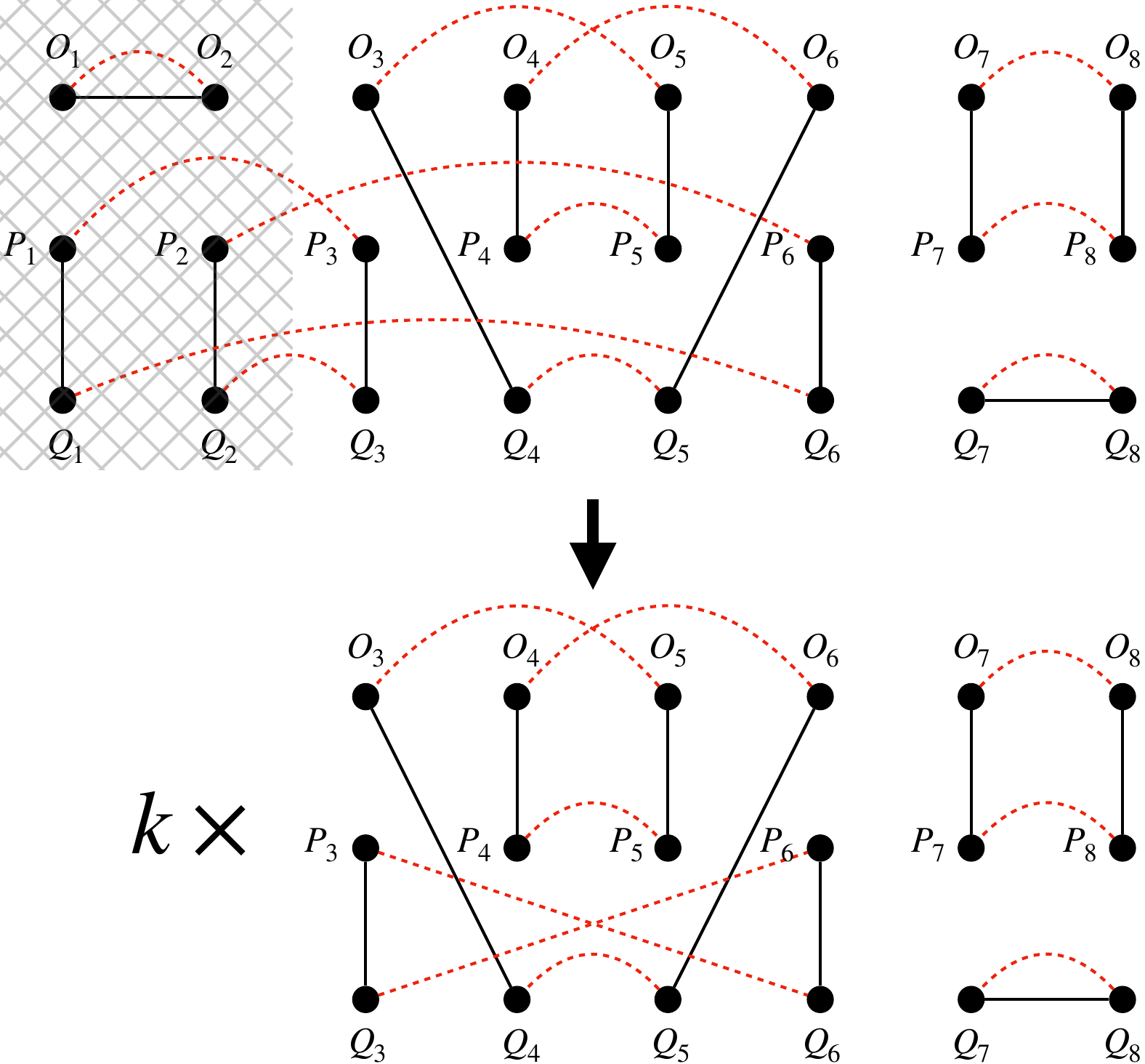}
    \caption{Example showing why the simple graphs with red edges that do not cross between rows are not sufficient to develop a recursion. Trying to ``integrate out'' or collapse the edges that connect the six vertices in the leftmost two columns (represented here with a crosshatch pattern over those vertices) induces a multiplicative factor of $k$ due to the connected component $O_{1}\mbox{---}O_{2}$ as well as red edges $P_{3}\mbox{---}Q_{6}$ and $P_{6}\mbox{---}Q_{3}$. Such red edges are not allowed for graphs in $\mathbb{G}_{n}^{2}$, so we must expand the set of graphs we consider in the recursion.}
    \label{fig:recursion_fail}
\end{figure}

To that end, we define a simple generalization of our graphs, where we allow \emph{all} possible perfect matchings of red edges across the $6n$ vertices. That is, we no longer restrict red edges to connect only vertices of the same letter (i.e., in the same row); we now allow the red edges to cross between two different rows. However, we still demand that each vertex still possess exactly one red edge. 

Let $a_{12}, a_{13}, a_{23}$ be the number of edges that span between the first and second, first and third, and second and third rows, respectively. We can then define a set of graphs $\mathbb{G}_{n}^{2}(a_{12},a_{13},a_{23},z)$ on $6n$ vertices, where the $z$ again indexes the $4^{n}$ possible sets of black edges. We can again write $\mathbb{G}_{n}^{2}(a_{12},a_{13},a_{23}) = \cup_{z\in[4^{n}]}\mathbb{G}_{n}^{2}(a_{12},a_{13},a_{23},z)$. Finally, then, we have
\begin{equation}\label{eqn:g_general} 
    g(n,a_{12},a_{13},a_{23}) := \sum_{\lambda \in \mathbb{G}_{n}^{2}(a_{12},a_{13},a_{23})}k^{C(\lambda)}
\end{equation}
The second moment we desire is then, of course, proportional to $g(n,0,0,0)$. 

A few constraints on $a_{12}, a_{13}, a_{23}$ are apparent immediately:
\begin{itemize}
    \item $a_{12} + a_{13}$, $a_{12}+a_{23}$, and $a_{13}+a_{23}$ (that is, the number of edges coming out of the first, second, and third row, respectively) must be even;
    \item  $a_{12} + a_{13}, a_{12}+a_{23}, a_{13}+a_{23}$ must all be less than or equal to $2n$ (there cannot be more than $2n$ edges coming out of a row with only $2n$ vertices given that there is exactly one red edge incident on every vertex).
\end{itemize}
We also observe that, while we do not explicitly keep track of these edges, we can also define $a_{11}, a_{22}, a_{33}$ as the number of ``proper'' edges that map between vertices in the first, second, and third rows, respectively. These edges have a simple relationship to the ones we do keep track of that can be derived by simply counting how many vertices in a given row are left after subtracting those that are used in edges that cross between rows:
\begin{align}
    a_{11} &= \frac{2n - a_{12} - a_{13}}{2}, \\
    a_{22} &= \frac{2n - a_{12} - a_{23}}{2}, \\
    a_{33} &= \frac{2n - a_{13} - a_{23}}{2}. 
\end{align}
Because we have the constraints that $a_{12} + a_{13}, a_{12}+a_{23}, a_{13}+a_{23}$ must all be even, $a_{11}, a_{22}, a_{33}$ are all integral. Also, the fact that $a_{12} + a_{13}, a_{12}+a_{23}, a_{13}+a_{23}$ must all be less than or equal to $2n$ ensures that $a_{11}, a_{22}, a_{33}$ are all non-negative as well.

It is also useful to write down the total number of graphs of each type. There are $6n$ total vertices, and $2n$ in each row. Given a vector $\vec{a} = (a_{12},a_{13},a_{23})$, we need to choose $a_{12}$ vertices in row 1 and row 2 to link to one another, $a_{13}$ in rows 1 and 3 (with no overlap between the vertices chosen in the first row corresponding to $a_{12}$ vs.~$a_{13}$), and $a_{23}$ in rows 2 and 3 (again, no overlap with previously chosen vertices is allowed). Once these vertices are chosen, it also remains to choose how to connect them. Finally, one must pair off the remaining vertices in each row, then multiply by $4^{n}$ to account for the black edges. The result is
\begin{widetext}
\begin{multline}\label{eqn:graph_counts}
    |\mathbb{G}_{n}^{2}(a_{12},a_{13},a_{23})| =
    \binom{2n}{a_{12}}\binom{2n-a_{12}}{a_{13}}\binom{2n}{a_{12}}\binom{2n-a_{12}}{a_{23}}\binom{2n}{a_{13}}\binom{2n-a_{13}}{a_{23}}a_{12}!a_{13}!a_{23}!\\
    \times (2n-a_{12}-a_{13}-1)!!(2n-a_{12}-a_{23}-1)!!(2n-a_{13}-a_{23}-1)!! 4^{n}.
\end{multline}
\end{widetext}
This result is useful because, if one sets $k = 1$ in \cref{eqn:g_general}, then every graph is put on equal footing; that is, any number of connected components contributes equally to the sum. Therefore, given a polynomial expansion in $k$ for any $g(n,a_{12},a_{13},a_{23})$ (note that \cref{thm:second_moment} still holds for the generalized graphs, except the highest order term need not be $2n$ anymore---generically it can reach $3n$), \cref{eqn:graph_counts} gives the sum of the coefficients on the monomials. 

\begin{figure*}
    \centering
    \includegraphics[width=\linewidth]{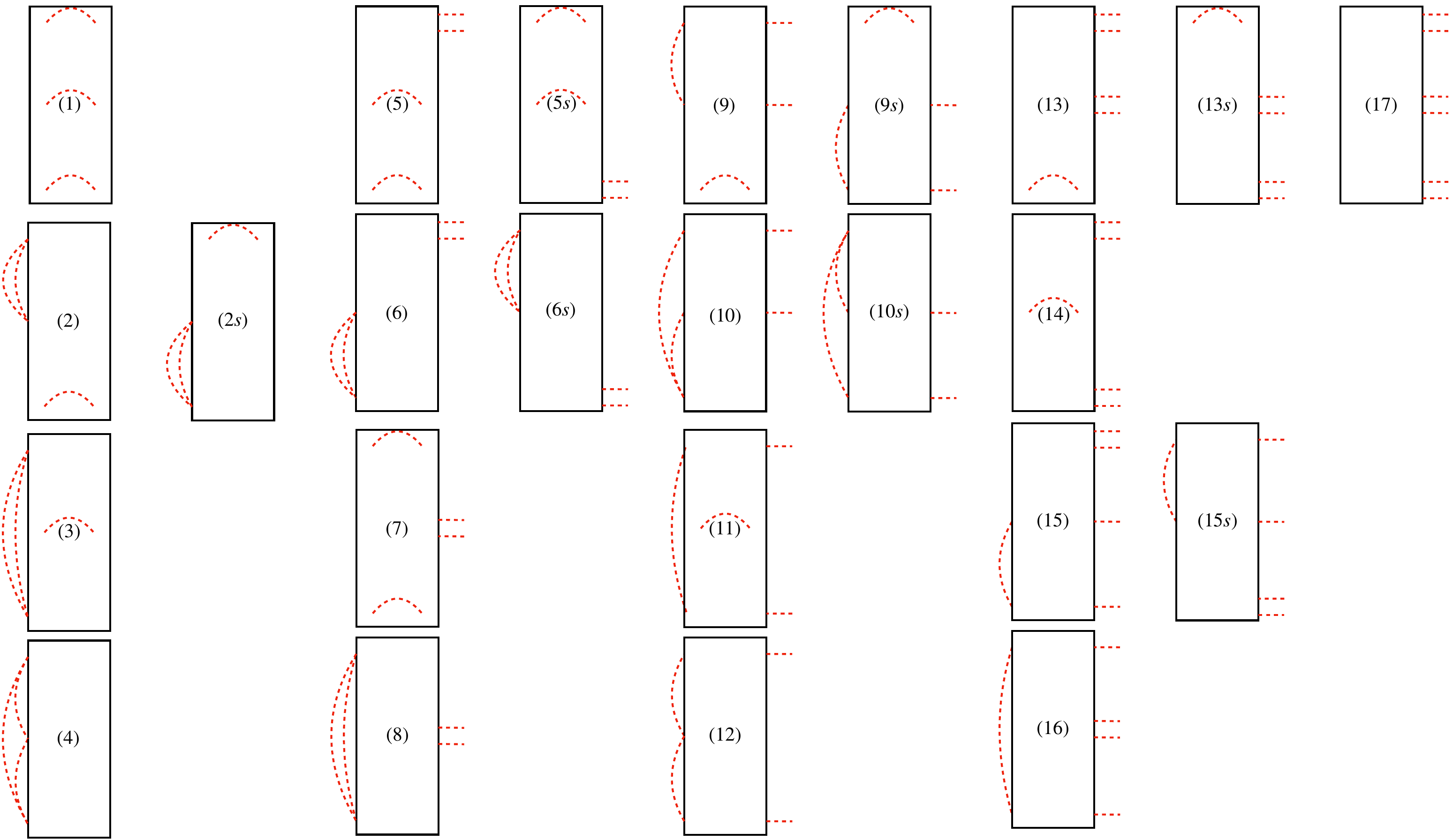}
    \caption{List of $17$ cases (up to symmetry) for how the first two columns in a graph of order $n$ can connect into the rest of the graph.}
    \label{fig:cases}
\end{figure*}

We now describe the recursion using the following equation
\begin{multline}\label{eqn:recursion}
    g(n,a_{12},a_{13},a_{23}) 
    =\\ \sum_{b_{12},b_{13},b_{23}}c(a_{12},a_{13},a_{23},b_{12},b_{13},b_{23})g(n-1,b_{12},b_{13},b_{23}).
\end{multline}
The goal is to determine the coefficients $c$ for all valid sets $a_{12},a_{13},a_{23}$ allowed by $n$. In order to do this, one must effectively determine the various ways integrating out two columns changes the possible red edge configurations.

Specifically, there are 17 ways (24 if one disambiguates symmetric cases) in which $\mathbb{C}_{1,2}$ can attach into a graph of order $n-1$ (that is, one with $2n-2$ columns). We classify these by the number of red edges that ``protrude'' from $\mathbb{C}_{1,2}$. We illustrate these cases in \cref{fig:cases} and now describe how to interpret these images.

Red edges that attach within a row inside the block are fixed, as there is only one possible edge that can connect two vertices in the same row. We depict the red edges that connect $\mathbb{C}_{1,2}$ to the rest of the graph as protruding from the same row on the right side of the block, as shown in \cref{fig:cases}. 
We depict red edges that go between different rows in $\mathbb{C}_{1,2}$ on the left of the box, again shown in \cref{fig:cases}. 
We do not draw the four possible sets of black edges within the block, but understanding their effect is crucial to the actual mechanics of the recursion. 

We must determine how each of these cases leads to a relationship between $\vec{a}$ and $\vec{b}$, as well as the coefficient $c$ in Eq.\ (\ref{eqn:recursion}), which is related to the number of possible graphs of order $n$ that, when integrated out, lead to \emph{the same} graph at order $n-1$. The coefficient out front is also affected by how many internal loops the given case has, as that of course leads to extra connected components that yield factors of $k$.  There are overall three different contributions to $c(\vec{a},\vec{b}) \coloneqq c(a_{12},a_{13},a_{23},b_{12},b_{13},b_{23})$:
\begin{itemize}
    \item Loop: This corresponds to the number of connected components in $\mathbb{C}_{1,2}$. This is the easiest contribution to determine;
    \item Vectorial: This corresponds to the relationship between $\vec{a}$ and $\vec{b}$ in \cref{eqn:recursion}, and, while somewhat simple in spirit, it often requires significant casework. In short, when integrating out $\mathbb{C}_{1,2}$, one loses contributions from internal edges that are lost by collapsing the vertices, but one gains edges of the types that are induced between the remaining vertices;
    \item Combinatorial: This corresponds to the combinatorial factors that are associated with how many ways a given case leads to the same graph at lower order. This depends both on the number of protruding edges and how the red and black edges interact via the vertices in $\mathbb{C}_{1,2}$.
\end{itemize}

With these ideas set forth, the evaluation of the recursion proceeds as follows. We first evaluate the base cases when $n = 1$. We then determine the loop, vectorial, and combinatorial contributions to each of the 17 cases depicted in \cref{fig:cases}, thus determining how that case contributes to the overall recursion. Finally, we evaluate the recursion numerically exactly, which is classically efficient (see \cref{sec:numerical_evaluation} and Appendix~\ref{app:recursion_complexity} for details). Note that, while it is, in principle, possible to write down analytically the contribution of each of the 17 cases depicted in \cref{fig:cases}, the terms are sufficiently numerous and complicated that we could not actually solve the recursion analytically; for more details, see \cref{app:recursion}, where the loop, vectorial, and combinatorial contributions are worked out for the cases.

\section{Analysis of the second moment}
In this section, we analyze the results derived from the numerically exact evaluation of the recursion described in the previous section. Specifically, we first discuss the code behind the recursion and provide some checks to gain confidence that code is accurate. We then derive some analytic results upper and lower bounding the second moment, which we then compare to the numerically exact data to understand how well they capture the scaling of the second moment. 

\subsection{Numerical Evaluation of the Recursion}\label{sec:numerical_evaluation}
Once the theoretical principles behind the recursion in \cref{eqn:recursion} are developed, we simply account for the contributions from each case and evaluate the recursion numerically exactly. We accomplish this using the Julia programming language \cite{Julia-2017} and find $g(n,0,0,0)$ from $ n = 1$ to $n = 40$ (which, recall, means up to photon sector $80$).

We now briefly describe our implementation of the exact numerical recursion; the code is available on GitHub \cite{github-lxeb}.
As a consequence of \cref{eqn:graph_counts}, the polynomial coefficients in $g(n, a_{12}, a_{13}, a_{23})$ grow at most factorially, so the number of bits needed to store the integers grows polynomially. Therefore, to ensure exact accuracy of all of the integer calculations, we use Julia's \texttt{BigInt} type, which allows us to achieve arbitrary-precision arithmetic \cite{Julia-2017}. Next, in order to avoid performing slow symbolic arithmetic operations, we represent polynomials in $k$ as \texttt{BigInt} arrays, where the $i^{\rm th}$ element of the array corresponds to the coefficient in front of the $k^i$ term in the polynomial. Multiplication and addition of polynomials in $k$ is then done at the array level.
We begin with $n=1$ and store the base case values of $g(1, a_{12}, a_{13}, a_{23})$ given in \cref{app:base-case}.
To compute the value of $g(n, a_{12}, a_{13}, a_{23})$, we iterate through the 17 cases described in \cref{app:recursion} and compute the various combinatorial factors and values of $b_{12}, b_{13}, b_{23}$ that show up in the sum in \cref{eqn:recursion}. We then recursively compute the values of $g(n-1, b_{12}, b_{13}, b_{23})$. The algorithm utilizes memoization every time any value of $g(n, a_{12}, a_{13}, a_{23})$ is computed so that the recursion rarely needs to go particularly deep.
In the end, in order to compute up to $g(40, 0, 0, 0)$, we compute $g(n, a_{12}, a_{13}, a_{23})$ for around 50\,000 combinations of arguments, resulting in almost 200 megabytes of (uncompressed) data. 

As mentioned, the evaluation of the recursion is classically efficient. In short, the number of allowed $\vec{a}$ (i.e.~those that satisfy the necessary bounds and parity constraints) is polynomially bounded, the size of the coefficients cannot be more than factorially large (meaning they can be stored with polynomial space), and the array-based multiplication and addition is classically tractable. More details are presented in \cref{app:recursion_complexity}.

We can check the computed values of $g(n,a_{12},a_{13},a_{23})$ derived via the recursion in a few ways. First, we note again that for any value of $g(n,a_{12},a_{13},a_{23})$, setting $k = 1$ (i.e., summing the coefficients in front of each monomial) yields the total number of graphs of this type, which is given in \cref{eqn:graph_counts}. Furthermore, \cref{lemma:k}(ii) (to be introduced below) gives the coefficient in front of the leading order term in $g(n,0,0,0)$. Our numerically exact computation of these numbers using the recursion matches these predicted values.

Second, for various $n$ and $k$, we numerically sample $10^{5}$ random $X \in \mathcal{G}^{k\times 2n}$, compute $\abs{\haf [X^{\top}X]}^{4}$ using the code provided by Ref.~\cite{gupt2019the-walrus:-a-l}, and average the results. This gives a numerical approximation to $(2n-1)!!g(n,0,0,0)$. We perform this calculation for $n, k \in \{1,2,\dots,9\}$. The result is shown in \cref{fig:testing_theory}, and we see good agreement between the approximate numerical calculations (data points and error bars) and the theoretical values predicted by the recursion (solid lines).

\begin{figure}
    \centering
    \newcommand{\plotdata}[2]{
  \addplot[color=#2, thick, opacity=0.5] table [x=k, y=analytic#1]{numerical_verification.dat};
  \addplot [color=#2, only marks, mark size=1.25pt, thick] 
  plot [error bars/.cd, y dir=both, y explicit]
  table [x=k, y=numeric#1, 
  y error=numericerror#1
  ] {numerical_verification.dat};
}

\begin{tikzpicture}
\begin{axis} [
    xlabel={$k$},
    ylabel={$\mathbb E_X \left\lvert\operatorname{Haf}(X^T X)\right\rvert^4$},
    ymode=log, 
    xtick=data,
    %
    legend cell align={left},
    legend style={draw=none,fill=none},
    legend style={font=\scriptsize},
    legend style={at={(0, 1)},anchor=north west}
]

    \plotdata{1}{black}
    \plotdata{2}{red}
    \plotdata{3}{green!80!black}
    \plotdata{4}{purple}
    \plotdata{5}{brown}
    \plotdata{6}{yellow!80!black}
    \plotdata{7}{pink!80!black}
    \plotdata{8}{orange}
    \plotdata{9}{blue}

    \addlegendentry{Analytic}
    \addlegendentry{Numeric}

    \node at (9.5,1e-2) [anchor=south east,rectangle] {\scriptsize $n=1$};
    \node at (9.5,1e44) [anchor=north east,rectangle] {\scriptsize $n=9$};

    \draw[->,>=stealth](9.25, 1e3)--(9.25, 1e39);

\end{axis} 
\end{tikzpicture}
    \caption{Numerical test of recursion. The $x$-axis represents $k$, and the $y$ axis represents $\underset{X\sim \mathcal{G}^{k\times 2n}}{\E}[\abs{\haf(X^{\top}X)}^{4}]$. Solid lines, from $n = 1$ through $n = 9$ are the theoretical predictions derived from the recursion relation (see \cite{github-lxeb} for the code). Dots and bars represent the expected value and standard error, respectively, estimated by sampling $10^{5}$ random Gaussian matrices and computing the second moment using the code provided by Ref.~\cite{gupt2019the-walrus:-a-l}. Note that, for many points, the size of the error bar is smaller than its associated dot. Further, there is an asymmetry in the error bars due to the log nature of the plot. We see excellent alignment between theory and numerics for $n = 1$ through $n=5$. For larger $n$, the agreement is still good, but we seem to undersample the true value in many cases. We suspect that this is because the distribution of the second moment has a long tail, meaning we do not suspect that the given error bars are indicative of the true difference between the sampled and numerically exact data.  
    We believe that were we able to either take sufficiently more samples we would see stronger agreement between the sampled and true means, but this option is too computationally demanding given the size of the matrices involved and the exponential complexity of classically computing the hafnian \cite{bjorklund_faster_2019-1}.}
    \label{fig:testing_theory}
\end{figure}
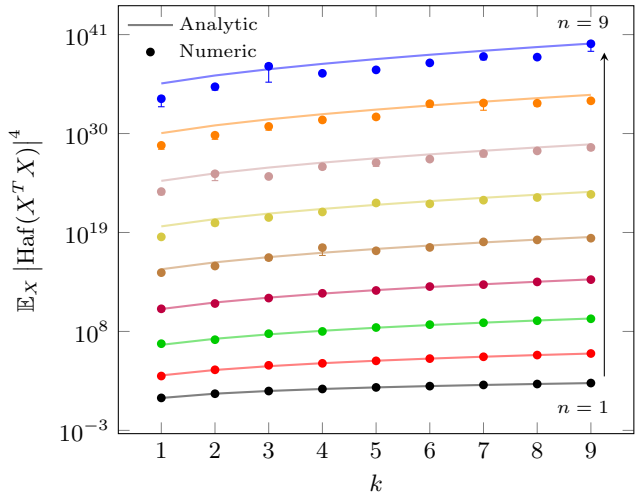

\subsection{Scaling of the Second Moment}\label{sec:scaling}
While we have not found a closed form for the solution to the recursion, we are able to derive a few simple analytic results about the values of the coefficients of the polynomial expansion as well as the overall scaling of the second moment. The former are covered in Ref.~\cite{ehrenbergTransitionAnticoncentrationGaussian2023}, as they are crucial to demonstrating the transition in anticoncentration that is the central result of that work. The latter are new to this work. 

We recall Lemma~1 from Ref.~\cite{ehrenbergTransitionAnticoncentrationGaussian2023}. \begin{lemma}[Ref.~\cite{ehrenbergTransitionAnticoncentrationGaussian2023}]\label{lemma:k}
We have that 
\begin{enumerate}
    \item[i.] $M_2(1,n) = ((2n-1)!!)^4 4^n$;
    \item[ii.]  $ c_{2n} = (2n)!!$.
\end{enumerate}
\end{lemma}
The proof of part (i) consists of a direct calculation using \cref{eqn:second_moment_delta}; it also follows from the graph-theoretic framework by simply counting the number of possible graphs of type $\vec{a} = (0,0,0)$ (see \cref{eqn:graph_counts}). The proof of part (ii) follows from a reduction of the problem of counting connected components to a special case of the first moment using $k = 2$. We also reprove this result in a slightly different way in \cref{app:individual_coefficients}.

As a corollary of \cref{lemma:k}, we can derive upper and lower bounds on the second moment:
\begin{lemma}\label{lem:bounds} \cref{lemma:k} implies
    \begin{align}
        M_{2}(k,n)  &\leq (2n-1)!!^{4}4^{n}k^{2n}, \label{eqn:ub} \\
        M_{2}(k,n) &\geq (2n)!k^{2n}, \label{eqn:lb1}\\
        M_{2}(k,n) &\geq (2n-1)!!^{4}4^{n}\label{eqn:lb2}.
    \end{align}
\end{lemma}
\begin{proof}
    We first prove the upper bound. The leading term in $g(n,0,0,0)$ is of the form $k^{2n}$, and the total number of graphs with no red edges crossing between rows is $(2n-1)!!^{3}4^{n}$. Thus, the upper bound comes from saying that all graphs have $2n$ connected components. 

    We next prove the lower bounds. The first lower bound comes from considering only the leading order term in the polynomial expansion, which is given in \cref{lemma:k}(ii). Because each term in the expansion is non-negative, this is a valid lower bound. The second lower bound comes from observing that $g(n,0,0,0)$ is monotonically increasing with $k$, as there are no negative coefficients in the polynomial expansion. Therefore, we can also take a lower bound which is simply the value at $k = 1$, which we know counts the total number of possible graphs and follows from \cref{lemma:k}(i).
\end{proof}
Stirling's approximation tells us when each lower bound is most useful: 
\begin{align}
    (2n)!k^{2n} &\sim (nk)^{2n}\paren*{\frac{4}{e^{2}}}^{n}, \label{eqn:lb1_stirling} \\
    (2n-1)!!^{4}4^{n} &\sim n^{4n}\paren*{\frac{64}{e^{4}}}^{n}\label{eqn:lb2_stirling}.
\end{align}
For $k \in o(n)$, \cref{eqn:lb2_stirling} is larger, and when $k \in \omega(n)$, \cref{eqn:lb1_stirling} is instead larger. When $k\in \Theta(n)$, then both lower bounds have a leading dependence of $n^{4n}$, so which is better depends on the constant of proportionality.

Armed with our analytical results and the exact numerical data from the recursion, we can now investigate how the second moment scales with $k$ and $n$.
In \cref{fig:scaling}(a), we plot the logarithm of the upper and lower bounds, as well as the numerically exactly computed values for $(2n-1)!!g(n,0,0,0)$, for our largest available $n$, which is $n = 40$. We set $k = n^{a}$ with $a \in [0,4]$. We see that, except for when $k = n^{0}$ and the upper bound is exactly correct (as is the lower bound based on the number of graphs), the lower bound is a much better approximation. In fact, as expected, the lower bound based on the leading order appears to become a very good approximation as $k$ gets larger. 
\begin{figure*}
    \centering
    \includegraphics[width=\linewidth]{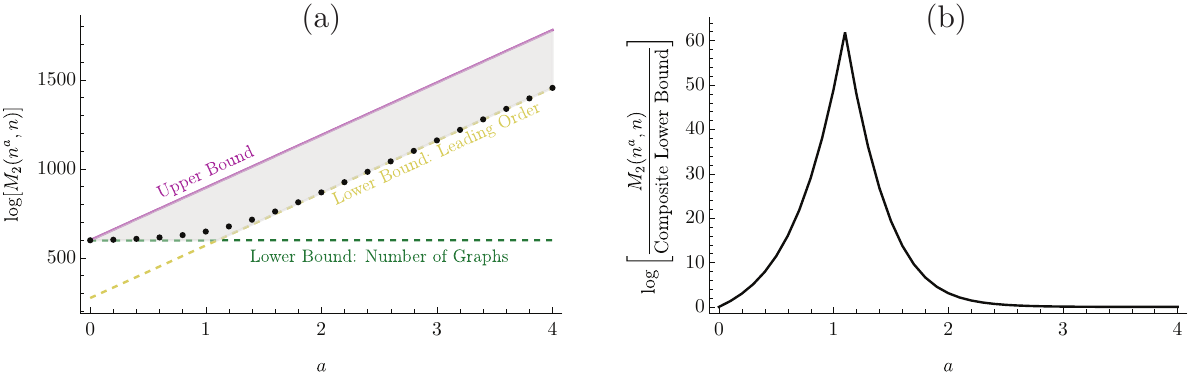}
    \caption{Plots showing scaling of the second moment compared to upper and lower bounds. For both plots, physically, $k$ should be an integer, but we here simply use the polynomial expansion of the second moment as a function of arbitrary real $k$. (a) Scaling of logarithm of the second moment and its upper and lower bounds for $n = 40$ and $k = n^{a}$ with $a \in [0,4]$. The green horizontal dashed line and the yellow slanted dashed line represent the lower bounds based on the number of graphs [Eq.\ (\ref{eqn:lb2})] and the leading order term [Eq.\ (\ref{eqn:lb1})], respectively. The maroon solid line represents the upper bound Eq.\ (\ref{eqn:ub}). The bound region is, therefore, highlighted in gray. Numerically exact data is given for $n = 40$ by the black dots \cite{github-lxeb}. Notice that the black dots representing the exact data stay within the gray region and, for most values of $a$, closely track the lower bound. (b) Difference between the logarithms of the exact data and the combined lower bound. The peak of the curve shows where the lower switches from the number of graphs to the leading order term. We see that, around $a = 2$, the lower bound based on the leading order becomes a good approximation. }
    \label{fig:scaling}
\end{figure*}

We should also point out that the logarithmic scaling of the y-axis of \cref{fig:scaling}(a) means that small differences between the exact values and the corresponding lower bound actually represent large multiplicative differences between the true values. For this reason, in \cref{fig:scaling}(b), we also plot the difference between the logarithms of the exact data and the composite lower bound defined by $\max\{$\cref{eqn:lb1},\,\cref{eqn:lb2}$\}$. This helps show how the exact data trends toward \cref{eqn:lb1} as $k$ grows. 

Relatedly, we can actually show analytically that \cref{eqn:lb1} cannot fully capture the scaling of the second moment when $k = O(n^{2})$. In \cref{app:individual_coefficients}, we discuss how to compute individual coefficients in the polynomial expansion of the second moment. There, we give a new proof that $c_{2n} = (2n)!!$, and we also prove for the first time that $c_{2n-1} = (2n)!!(3n-2)n$. Together, these two results mean
\begin{equation}
    \frac{c_{2n}k^{2n}}{c_{2n-1}k^{2n-1}} = \frac{k}{(3n-2)n} \sim \frac{k}{n^{2}}.
\end{equation}
Therefore, in order for the leading term $c_{2n}k^{2n}$ to asymptotically dominate $c_{2n-1}k^{2n-1}$, we require $k = \omega(n^{2})$. \emph{A fortiori}, for the leading term to dominate all other terms, and, therefore, for the leading-order lower bound to be a good approximation for the second moment, $k$ must be $\omega(n^{2})$. 

In summary, then, the lower bounds in \cref{eqn:lb1,eqn:lb2} typically track the true value of the second moment much better than the upper bound in \cref{eqn:ub}. When $k = \omega(n^{2})$, the first lower bound, \cref{eqn:lb1}, which is based on the leading order term, appears to be a very good approximation to second moment.

\section{Locating the Transition in Anticoncentration}\label{sec:anticoncentration}

We now move on to some of the concrete consequences of our work. The main result of Ref.~\cite{ehrenbergTransitionAnticoncentrationGaussian2023} is identifying a transition in anticoncentration in Gaussian Boson Sampling as a function of $k$, the number of initially squeezed modes. This result follows entirely from analytic results. Specifically, in Ref.~\cite{ehrenbergTransitionAnticoncentrationGaussian2023}, we show through direct computation that, when $k=1$, the probabilities do not anticoncentrate, and we use the leading order term to show that these probabilities weakly anticoncentrate in the limit that $k \to \infty$. Hence, we show the existence of a transition, but we do not isolate its exact location. We do conjecture that it occurs at $a = 2$, where $k$ scales with $n$ as $k=\Theta(n^{a})$, based on an allusion to Scattershot Boson Sampling \cite{lund_boson_2014}, which is another generalization of Fock state Boson Sampling; there the initial state is composed of two-mode squeezed states where one half of each state is measured and postselected on measurements with at most one photon. In short, one can roughly draw a connection between the presence of hiding in Scattershot Boson Sampling and the number of initially squeezed modes (this is detailed more thoroughly in Section S6 of the Supplementary Material in Ref.~\cite{ehrenbergTransitionAnticoncentrationGaussian2023}). 

The main contribution of this work is to show convincingly that the location of the transition is indeed at $k = \Theta(n^{2})$. We accomplish this through numerical arguments based on the exact data generated through the recursion for the second moment and a few more analytic results. We formalize this with the following conjecture:
\begin{conjecture}[Anticoncentration in Gaussian Boson Sampling]\label{con:anticoncentration}
    Let $2n = o(\sqrt{m})$ such that one operates in the (conjectured) hiding regime. Then Gaussian Boson Sampling does not anticoncentrate for $k = O(n^{2})$, but it weakly anticoncentrates with inverse normalized second moment, $m_{2}(k,n) := M_{1}^{2}(k,n)/M_{2}(k,n)$, scaling as $1/\sqrt{\pi n}$ for $k = \omega(n^{2})$.
\end{conjecture}

Our evidence for Conjecture~\ref{con:anticoncentration} is twofold and based on results regarding the anticoncentration of the approximate distribution (see the Supplemental Material of the companion piece Ref.~\cite{ehrenbergTransitionAnticoncentrationGaussian2023} for details on how to convert these statements to those about anticoncentration of the exact distribution):
\begin{enumerate}
    \item We provide a sequence of numerical plots of $\log[(m_{2}(k,n)\sqrt{\pi n})^{-1}]$ and its symmetric difference with respect to $n$ for various polynomial scalings of $k$ with $n$. The numerical plots of the function itself show an exponential scaling when $k = O(n^{2})$, but that the function becomes approximately constant when $k = \omega(n^{2})$. Similarly, the plots of the symmetric difference are positive in the $k=O(n^{2})$ regime, but approximately vanish when $k = \omega(n^{2})$.
    \item We show that, assuming the lower bound for $M_{2}(k,n)$ is a good approximation, weak anticoncentration holds for $k = \omega(n^{2})$. We also show that there is a lack of anticoncentration when $k = o(n)$. 
\end{enumerate}

We begin with the numerical evidence. In \cref{fig:transition_plot}, we set $k = n^{a}$ and plot $\log[(m_{2}(k,n)\sqrt{\pi n})^{-1}]$ for various values of $a$. We choose this quantity because, in the asymptotic limit of large $k$, $(m_{2}(k,n)\sqrt{\pi n})^{-1} \sim 1$, but when $k = 1$, it is exponentially big \cite{ehrenbergTransitionAnticoncentrationGaussian2023}. Therefore, we hope to use \cref{fig:transition_plot} to understand how this quantity interpolates between the exponential and polynomial behavior of $m_{2}(k,n)^{-1}$. In \cref{fig:transition_plot}(a), we plot $\log[(m_{2}(k,n)\sqrt{\pi n})^{-1}]$ for $a = 0.5$ to $a = 4.0$ with spacing $0.5$. We see that for $a \leq 2$, this quantity seems to linearly increase with $n$, meaning that $m_{2}(k,n)^{-1}$ is exponentially large in $n$. However, for $a > 2$, it trends to a small constant. Because $m_{2}(k,n) \sim 1/\sqrt{\pi n}$ is derived in the limit of asymptotically large $k$ using the leading order lower bound for the second moment in \cref{eqn:lb1}, this suggests that the use of this lower bound is a good approximation to the second moment when $a > 2$; this aligns well with \cref{fig:scaling}. Thus, we see that, when $a > 2$, the normalized second moment trends to its asymptotic-in-$k$ value of $\sqrt{\pi n}$. In \cref{fig:transition_plot}(b), we zoom in on the suspected transition point and plot the same quantity when $a \in \{1.95, 1.99, 2.00, 2.01, 2.05, 2.10, 2.15, 2.20\}$. We see similar behavior in this plot; namely, at approximately $a = 2$, the curves transition from growing in $n$ to decreasing toward $0$. To clarify this point even further, we also plot the symmetric difference of the above quantity as a function of $n$ (excluding the minimum and maximum values of $n$). Here, the symmetric difference of a function $f(n)$, which we refer to as $\Delta_{n}f(n)$, is defined as $(f(n+1)-f(n-1))/2$. \cref{fig:transition_plot}(c) and \cref{fig:transition_plot}(d) use the same values of $a$ as \cref{fig:transition_plot}(a) and \cref{fig:transition_plot}(b), respectively. We see that, up to some finite size effects, when $a > 2$ this symmetric difference trends to $0$, but it remains positive for $a \leq 2$. 
\begin{figure*}[ht!]
    \centering
    \includegraphics[width=\linewidth]{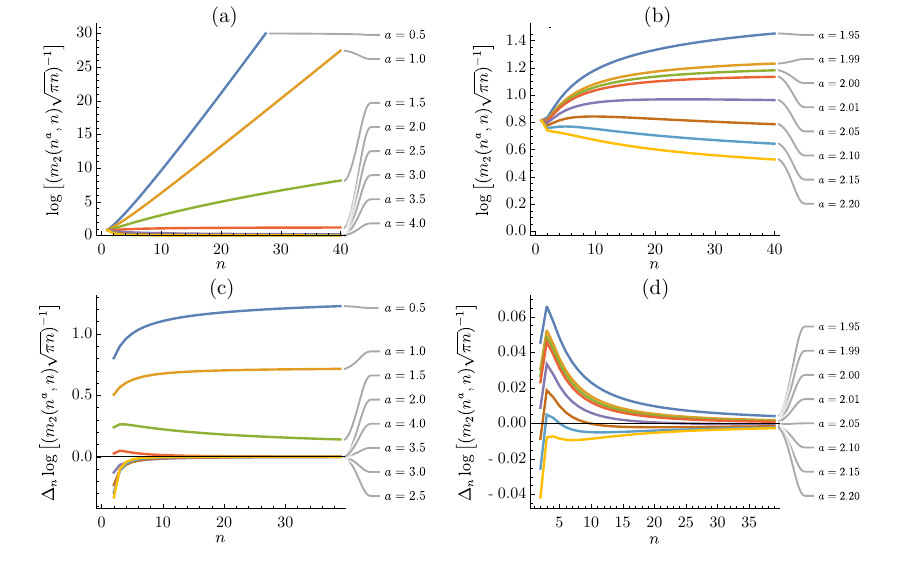}
    \caption{Plots of $\log[(m_{2}(k,n)\sqrt{\pi n})^{-1}]$ and its symmetric difference, notated as $\Delta_{n}$, as a function of $n$ for $k = n^{a}$. Recall that $m_{2}(k,n):=M_{1}(k,n)^{2}/M_{2}(k,n)$ and, for asymptotically large $k$, $m_{2}(k,n)\sim 1/\sqrt{\pi n}$ \cite{ehrenbergTransitionAnticoncentrationGaussian2023}. (a) $a \in [0.5, 4.0]$, equally spaced by $0.5$. (b) $a \in \{1.95, 1.99, 2.00, 2.01, 2.05, 2.10, 2.15, 2.20\}$ to show the regime around $a = 2$ more clearly. (c) The symmetric difference of $\log[(m_{2}(k,n)\sqrt{\pi n})^{-1}]$ with respect to $n$, again with $a \in [0.5, 4.0]$. (d) Zooming in on the symmetric difference when $a$ is around $2$, with the same values as plot (b). Note that each of the curves in plots (a) and (b) are composed of numerically exact data at $40$ points ($n \in \{1, \dots, 40\}$) that are smoothed for visualization. The same holds for plots (c) and (d), except there are only $38$ points ($n =1$ and $n=40$ are excluded because we compute the symmetric difference). Finally, while $k$ physically must be an integer, we do not enforce that for these plots; we instead just using the polynomial expansion of the moments to extend $k$ to arbitrary real numbers.}
    \label{fig:transition_plot}
\end{figure*}

We next plot in \cref{fig:derivative_transition_plot} the symmetric difference $\Delta_{n}\log[(m_{2}(k,n)\sqrt{\pi n})^{-1}]$ with respect to $n$ at $n = 39$ (the largest $n$ for which we can compute the symmetric difference) as a function of $a$. We see the symmetric difference vanish near $a = 2$, as would be expected if the transition occurs at $k = \Theta(n^{2})$. The inset of  \cref{fig:derivative_transition_plot} clarifies this by plotting the logarithm of this symmetric difference such that its vanishing instead becomes a divergence.

\begin{figure}[ht!]
    \centering
    \includegraphics[width=\linewidth]{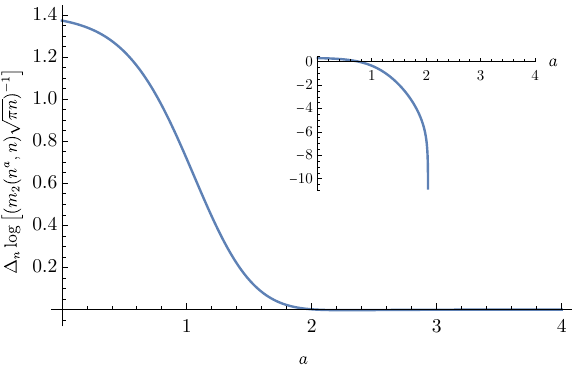}
    \caption{Symmetric difference $\Delta_{n}\log[(m_{2}(k,n)\sqrt{\pi n})^{-1}]$ evaluated at $n = 39$. Here, $k = n^{a}$, and $a$ represents the $x$-axis. Again, physically, $k$ must be an integer, but for this plot we are simply using the polynomial expansions of the moments where $k$ can be an arbitrary real number. This symmetric difference vanishes very close to $a = 2$, suggesting that, when $k = \Omega(n^{2})$, the quantity $m_{2}(k,n)\sqrt{\pi n}$ is a constant, meaning the normalized second moment appears to scale as $\sqrt{\pi n}$. The inset simply plots the $\log$ of the $y$-axis in the main plot (still with $a$ along the $x$-axis) in order to visualize more clearly the transition. The divergence occurs somewhere around $a=2.03$, but we suspect this difference is due solely to finite-size effects. Beyond this divergence, the symmetric difference is negative, meaning the logarithm is complex and, hence, unplotted.}
    \label{fig:derivative_transition_plot}
\end{figure}

For our second, more analytic argument, we show that if the lower bound is a good approximation to the second moment, then weak anticoncentration holds for $k = \omega(n^{2})$ and there is a lack of anticoncentration when $k = o(n)$.

First, consider the case $a < 1$. Note that $k = n^{a}$ is negligible to $n$ (asymptotically in $n$). Therefore, up to subleading order,
\begin{equation}
    \frac{(k+2n-2)!!}{(k-2)!!}\sim (2n)!!. 
\end{equation}
Using \cref{eqn:lb2}, which is a valid lower bound, we get
\begin{align}
    \frac{M_{2}(k,n)}{M_{1}(k,n)^{2}} &\gtrsim \frac{(2n-1)!!^{4}4^{n}}{((2n-1)!!(2n)!!)^{2}}\\
    &= 4^{n} \frac{(2n-1)!!^{2}}{(2n)!!^{2}} \\
    &\sim \frac{4^{n}}{\pi n},
\end{align}
which is exponentially big, demonstrating a lack of anticoncentration (accounting for the subleading contribution of $k$ does not change the conclusion). Here, we have used Stirling's approximation and
\begin{equation}
     \frac{(2n)!!}{(2n-1)!!} \sim \frac{\sqrt{2\pi n}(2n/e)^{n}}{\sqrt{2}(2n/e)^{n}} = \sqrt{\pi n}.
\end{equation}

We now examine the case where $k = n^{a}$ with $a > 2$. We use that, according to \cref{fig:scaling}, the lower bound $M_{2}(k,n) \geq (2n)!k^{2n}$ is actually an extremely good approximation to the second moment. Here, $k$ now dominates $n$, so
\begin{equation}
    \frac{(k+2n-2)!!}{(k-2)!!}\sim \sqrt{k^{2n}} = n^{an}. 
\end{equation}
Correspondingly, the normalized second moment scales as
\begin{align}
    \frac{M_{2}(k,n)}{M_{1}(k,n)^{2}}  &\sim \frac{(2n-1)!!(2n)!!k^{2n}}{(2n-1)!!^{2}k^{2n}}\\
    &= \frac{(2n)!!}{(2n-1)!!}\\
    &\sim \sqrt{\pi n}.
\end{align}
Therefore, when $k = \omega(n^{2})$, weak anticoncentration holds (again, the inclusion of any subleading terms does not change the conclusion). Note that this argument is similar to the argument used to demonstrate the existence of the transition in the first place, but it uses the fact that the second moment is already well approximated by the leading order lower bound at $k = \omega(n^{2})$ instead of just in the asymptotic limit of large $k$. Unfortunately, our current results are insufficient to more formally handle the regime $a \in [1,2]$ regime.

To recap, we have shown the following results. First, we have provided numerics in \cref{fig:transition_plot,fig:derivative_transition_plot} that suggest that $\sqrt{\pi n}$ is a good approximation to the normalized second moment when $k = \omega(n^{2})$. This is the value of the normalized second moment that is calculated when one uses the lower bound in \cref{eqn:lb1} that is based on the leading order term. Similarly, these plots numerically indicate that when $k = O(n^{2})$, the normalized second moment grows exponentially in $n$, meaning there is a lack of anticoncentration. Next, we have shown that, if the leading order is a good approximation to the second moment, which, according to \cref{fig:scaling} occurs when $k = \omega(n^{2})$, then the normalized second moment scales as $\sqrt{\pi n}$, meaning weak anticoncentration holds in that regime. We have also shown that for $k = O(n)$, there is a lack of anticoncentration. All together, the totality of the evidence presented here strongly suggests the veracity of Conjecture~\ref{con:anticoncentration} and that the transition between lack of anticoncentration and weak anticoncentration in the approximate output distribution occurs at $k = \Theta(n^{2})$.

\section{Conclusion}
\label{sec:conclusion}
In this work, we have studied the output distribution of the prototypical setup for Gaussian Boson Sampling in the hiding regime. Our main theoretical contribution is the development of a recursion relation that allows one to compute numerically exactly in polynomial time the second moment of these output probabilities for any photon Fock sector. We additionally detail separate ways to calculate individual coefficients of the polynomial expansion of the second moment. Together, these results provide strong evidence for our conjecture that the transition in anticoncentration, whose existence is proven in Ref.~\cite{ehrenbergTransitionAnticoncentrationGaussian2023}, occurs at $k = \Theta(n^{2})$. 

Ideally we would have been able to derive a closed-form expression for the polynomial description of the second moment akin to \cref{thm:first_moment}, as this might have allowed us to formally prove this conjecture, but we leave this important question to future work. It would also be nice to develop a better, more intuitive understanding for why this transition occurs. It appears to be related to the transition between collisional and collision-free outputs in Scattershot Boson Sampling, but the connection is not perfect, and further investigation seems worthwhile. 

Related to all of these points, the precise nature of the crossover at $k = \Theta(n^{2})$ is an interesting realm of future study. Specifically, we conjecture that weak anticoncentration holds for $k = \omega(n^{2})$ and there is a lack of anticoncentration when $k = O(n^{2})$, which of course places the transition at $k = \Theta(n^{2})$. But precisely how the normalized moment behaves as we tune $a$ through $a=2$ deserves special attention.

Our results open the door for answering other questions of interest. In particular, our results may make it possible to evaluate how well certain classical algorithms may sample from the output distribution or evaluate spoofing cross-entropy benchmarking in Gaussian Boson Sampling. Further exploration here is worthwhile. We also note that we have studied Gaussian Boson Sampling with no noise and number-resolving detectors. It would be interesting to see whether our techniques can be expanded to imperfect settings, such as when photons are partially distinguishable \cite{shi_effect_2022}, or when the measurement detectors only distinguish between the presence or absence of photons \cite{quesada_gaussian_2018}. 

Finally, the graph-theoretic approach that we have developed in this manuscript is surprisingly flexible, and it deserves continued treatment. In \cref{app:alternative}, we present another way to use the graphs in $\mathbb{G}_{2}^{n}$ in order to develop a recursion that can solve for the second moment. In short, this other approach observes that there are really only five types of black edges in our graphs: ones that stay in row 1, ones that stay in row 3, and ones that go between rows 1 and 2, rows 1 and 3, and rows 2 and 3. Because we are interested only in the number of connected components, and because we sum over \emph{all} perfect matchings defined by red edges in each row, we are free to drag the black edges around and order them in new, convenient ways. Therefore, looking at these graphs from the perspective of the total number of each type of black edge allows us to conceive of a different kind of recursion for the second moment. While we only sketch the idea behind this alternative recursion, we believe that it may be a promising new way of looking at the problem. In particular, this new approach allows us to find an, admittedly, somewhat complicated, expression for $c_{1}$ (that reproduces our expression for $c_{1}$ found via the original recursion up to $n = 40$). However, this new approach should not be viewed as a strict alternative to what we have derived in this manuscript, but a complementary approach that might yield new insights. We leave exploring it to future work. 

\vspace{1em}
\acknowledgments

We thank Changhun Oh, Bill Fefferman, Marcel Hinsche, Max Alekseyev, and Benjamin Banavige for helpful discussions. We thank Jacob Bringewatt for providing feedback on Appendix~A. A.E., J.T.I., and A.V.G.~were supported in part by the DoE ASCR Accelerated Research in Quantum Computing program (award No.~DE-SC0020312), DARPA SAVaNT ADVENT, AFOSR MURI, DoE ASCR Quantum Testbed Pathfinder program (awards No.~DE-SC0019040 and No.~DE-SC0024220), NSF QLCI (award No.~OMA-2120757), NSF STAQ program, and AFOSR. Support is also acknowledged from the U.S.~Department of Energy, Office of Science, National Quantum Information Science Research Centers, Quantum Systems Accelerator. 
JTI thanks the Joint Quantum Institute at the University of Maryland for support through a JQI fellowship. 
D.H. acknowledges funding from the US Department of Defense through a QuICS Hartree fellowship.
Specific product citations are for
the purpose of clarification only, and are not an endorsement by the authors or NIST.

\bibliography{LXEB.bib,doms_bib.bib}

\onecolumngrid
\cleardoublepage

\setcounter{section}{0}
\setcounter{figure}{0}
\renewcommand{\thefigure}{\thesection \arabic{figure}}
\counterwithin{figure}{section}
\setcounter{table}{0}
\renewcommand{\thetable}{\thesection \arabic{table}}
\counterwithin{table}{section}
\setcounter{theorem}{0}
\renewcommand{\thetheorem}{\thesection \arabic{theorem}}
\setcounter{lemma}{0}
\renewcommand{\thelemma}{\thesection \arabic{lemma}}
\setcounter{equation}{0}
\renewcommand{\theequation}{\thesection \arabic{equation}}

\clearpage
\begin{appendix}

In the appendices, we provide details and derivations that supplement the discussion in the main text.

\begin{itemize}
    \item \cref{app:recursion_complexity}: We discuss the classical complexity of evaluating the recursion and show that it is efficient (i.e., the time and space required scale polynomially) in the Fock sector $n$;
    \item \cref{app:recursion}: We provide the graph-theoretic details for how to derive the recursion;
    \item \cref{app:individual_coefficients}: We discuss how to compute individual coefficients of the polynomial expansion of the second moment. Specifically, we give one method to calculate the leading and first subleading terms in the polynomial expansion of the second moment;
    \item \cref{app:alternative}: We discuss an alternative method for developing a recursion to the solve for the second moment. We also apply this alternative picture to find an expression for the constant term in the polynomial expansion of the second moment. 
\end{itemize}

\section{Classical Complexity of Evaluating Recursion}\label{app:recursion_complexity}
In this appendix, we argue that the numerical evaluation of the recursion and, hence, the second moment, is classically efficient (that is, the runtime and space used are at most polynomial) in $n$, which corresponds to the Fock sector of interest in the output samples. 

We recall the setup of the recursion as we describe it in the main text. Specifically, we define
\begin{equation}
    g(n,a_{12},a_{13},a_{23}) := \sum_{\lambda \in \mathbb{G}_{n}^{2}(a_{12},a_{13},a_{23})}k^{C(\lambda)}. 
\end{equation}
$\mathbb{G}_{n}^{2}(a_{12},a_{13},a_{23})$ is the set of second-moment graphs of order $n$ with $a_{ij}$ red edges that cross between rows $i$ and $j$. $C(\lambda)$ is the number of connected components of $\lambda$. The second moment is given by $(2n-1)!!g(n,0,0,0)$. We then write down the recursion using these $g(n,a_{12},a_{13},a_{23})$ as
\begin{equation}
    g(n,a_{12},a_{13},a_{23}) = \sum_{b_{12},b_{13},b_{23}}c(a_{12},a_{13},a_{23},b_{12},b_{13},b_{23})g(n-1,b_{12},b_{13},b_{23}).
\end{equation}

We list the following constraints on $\vec{a}$, which is shorthand for $(a_{12}, a_{13}, a_{23})$. First, $a_{12} + a_{13}$, $a_{12}+a_{23}$, and $a_{13}+a_{23}$ (the edges that exit the first, second, and third rows respectively) must be even. Second, $a_{12} + a_{13}, a_{12}+a_{23}, a_{13}+a_{23}$ must all be less than or equal to $2n$, as there cannot be more than $2n$ edges coming out of a row with only $2n$ vertices given that there is exactly one red edge incident on every vertex. Finally, we also add here that, clearly, $a_{12}, a_{13}, a_{23}$ are non-negative. These constraints imply a finite number of valid vectors $\vec{a} = (a_{12},a_{13},a_{23})$ for a given order $n$, and any vector satisfying these constraints corresponds to a valid set of graphs and, therefore, a term $g(n,\vec{a})$ in the recursion. We provide an example of all possible $\vec{a}$ when $n = 4$ in \cref{tab:vectors}. 
\begin{table}[ht!]
    \def\arraystretch{1.2}
    \centering
    \begin{tabular}{c|c}
        $m$ & $\vec{a}$ \\
        \hline
        $0$ & $(0,0,0)$ \\
        $1$ & $\emptyset$ \\
        $2$ & $(2,0,0)$ \\
        $3$ & $(1,1,1)$ \\
        $4$ & $(2,2,0), (4,0,0)$ \\
        $5$ & $(3,1,1)$ \\
        $6$ & $(6,0,0), (4,2,0), (2,2,2)$ \\
        $7$ & $(5,1,1), (3,3,1)$ \\
        $8$ & $(8,0,0), (6,2,0), (4,4,0), (4,2,2)$ \\
        $9$ & $(7,1,1), (5,3,1), (3,3,3)$ \\
        $10$ & $(6,2,2), (4,2,2)$ \\
        $11$ & $(5,3,3)$ \\
        $12$ & $(4,4,4)$
    \end{tabular}
    \caption{All possible $\vec{a}$, up to permutations of the vector elements, for $2n=8$. Each entry satisfies the constraints that $a_{12} + a_{13}$, $a_{12}+a_{23},$ and $a_{13}+a_{23}$ are even and less than or equal to $2n$, $a_{12}, a_{13}$, and $a_{23}$ are non-negative, and $a_{12}+a_{13}+a_{23} = m$.}
    \label{tab:vectors}
\end{table}

Clearly, as $n$ grows, the number of possible $\vec{a}$ for which one must evaluate $g(n,\vec{a})$ also grows. However, we can bound this growth as being polynomial in $n$ using some arguments about partitions. Recall that a partition of a positive integer $m$ of size $s$ is a set (i.e., order does not matter) of $s$ positive integers whose sum is $m$. A weak partition of $m$ of size $s$ relaxes the positivity constraint of the set such that it contains $s$ non-negative elements ($m$ is still positive). 

Let $ m := a_{12} + a_{13} + a_{23}$. Then $m \leq 3n$, which follows from the fact that 
\begin{equation}
    2a_{12} + 2a_{13} + 2a_{23} = (a_{12}+a_{13}) + (a_{12}+a_{23}) + (a_{13}+a_{23}) \leq 6n.
\end{equation}
The conditions listed above on $\vec{a}$ imply that each $\vec{a}$ is a weak partition of size $3$ of $m \leq 3n$ that satisfies two further constraints: all $3$ elements of the set must have the same parity as $m$, and no element can be larger than $2n$. 

Now, the number of partitions of $m$ of size at most 3 is $\lfloor (m+3)^{2}/12 \rceil$ \cite{hardyFamousProblemsTheory2011} (note that $\lfloor M \rceil$ refers to the closest integer to $M$). Therefore, the number of partitions of $m$ of size exactly $3$, or $p_{3}(m)$, is bounded by this value, which implies that $\sum_{m=0}^{3n}p_{3}(m) = O(n^{3})$. In turn, the number of $\vec{a}$, up to permutations of the elements of $\vec{a}$, is bounded by $O(n^{3})$ (because they form an even more restricted class of weak permutations). We can overcount for these permutations with a simple constant multiplicative factor of $3!$ (this overcounts because, when numbers are repeated in the partition, there are fewer distinct permutations). Thus, we have a polynomial bound on the number of terms in our recursion at any Fock sector $n$ (note that we could tighten this bound a bit by accounting more precisely for the parity constraint on the elements $\vec{a}$, but, because we are interested only in classical efficiency, this polynomial bound that arises from considering only size-$3$ partitions is sufficient). 

To be sure that the recursion is efficiently computable, however, the actual values of the terms in the recursion must not grow too quickly. In particular, recall that each term $g(n,\vec{a})$ has a polynomial expansion in $k$ of order at most $3n$ (this is the largest number of connected components possible when each one must have at least $2$ vertices). The sum of the coefficients of $g(n,\vec{a})$ is the same as the number of graphs in $\mathbb{G}_{n}^{2}(a_{12},a_{13},a_{23})$, which we derived to be 
\begin{multline}
    |\mathbb{G}_{n}^{2}(a_{12},a_{13},a_{23})| =
    \binom{2n}{a_{12}}\binom{2n-a_{12}}{a_{13}}\binom{2n}{a_{12}}\binom{2n-a_{12}}{a_{23}}\binom{2n}{a_{13}}\binom{2n-a_{13}}{a_{23}}a_{12}!a_{13}!a_{23}!\\
    \times (2n-a_{12}-a_{13}-1)!!(2n-a_{12}-a_{23}-1)!!(2n-a_{13}-a_{23}-1)!! 4^{n}.
\end{multline}
This is, at most, factorially big in $n$, which means that the number of bits needed to store these numbers, and, hence, $g(n,\vec{a})$ is polynomial in $n$.

Therefore, we have a polynomial bound on the number of terms in the recursion, as well as on the space needed to represent each of these terms. Finally, because the actual recursion consists only of polynomial numbers of multiplication and addition, which can each be accomplished in time polynomial in the size of the inputs, the actual computation is efficient.

\section{Building the Recursion}\label{app:recursion}
We now describe precisely how to derive and evaluate the recursion relation \cref{eqn:recursion}, which we copy again here for convenience:
\begin{equation}\label{eqn:app_recursion}
    g(n,a_{12},a_{13},a_{23}) = \sum_{b_{12},b_{13},b_{23}}c(a_{12},a_{13},a_{23},b_{12},b_{13},b_{23})g(n-1,b_{12},b_{13},b_{23}).
\end{equation}
We note that we implement the full recursion \cite{github-lxeb} in both the Julia programming language \cite{Julia-2017} and Mathematica \cite{Mathematica}.
Recall that $g(n,a_{12},a_{13},a_{23})$ is a polynomial in $k$ where the coefficient in front of $k^{i}$ is the number of graphs of type $\vec{a} = (a_{12},a_{13},a_{23})$ that have $i$ connected components. Again, a graph of type $\vec{a}$ has $a_{ij}$ edges that go between rows $i,j$. 

\begin{figure}
    \centering
    \includegraphics[width=\linewidth]{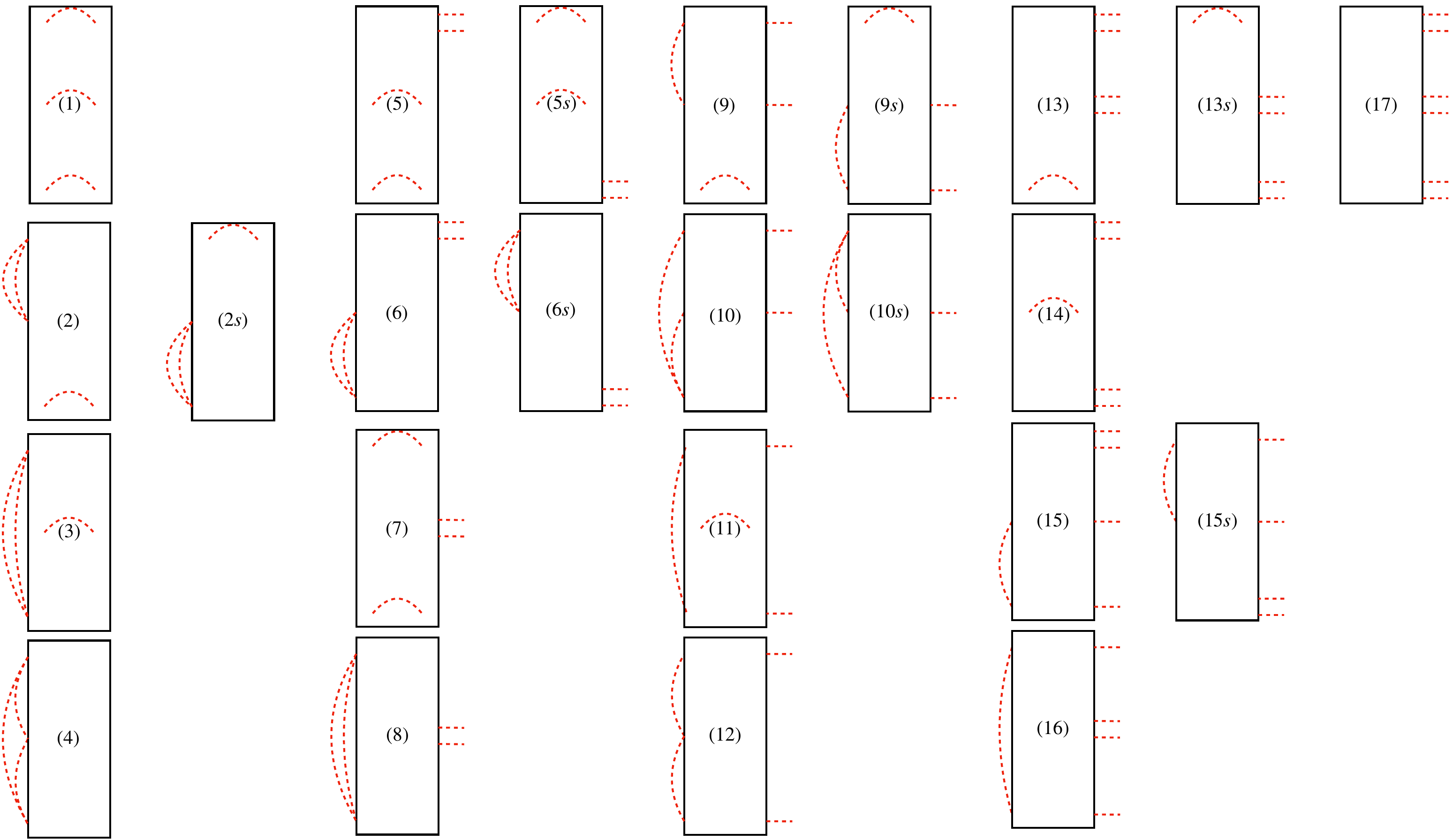}
    \caption{Copy of \cref{fig:cases}. List of $17$ cases (up to symmetry) for how the first two columns in a graph of order $n$ can connect into the rest of the graph. }
    \label{fig:app_cases}
\end{figure}

We first describe the base case, i.e.~ $g(1,a_{12},a_{13},a_{23})$ for all valid vectors $\vec{a} = (a_{12},a_{13},a_{23})$. We then describe how to handle each of the possible $17$ cases that contribute to the recursion that are depicted in \cref{fig:cases}, which is copied again here for convenience. 

The way that we handle each case is as follows. We consider all graphs of order $n$ such that the leftmost two columns, which, recall, we refer to as $\mathbb{C}_{1,2}$, have red edges that correspond to that case. We then ``integrate out'' these edges to determine how to write the contribution of that case at order $n$ in terms of the terms at order $n-1$. When we say integrate out, we mean that we collapse any path that goes through $\mathbb{C}_{1,2}$ into a new edge that remains entirely in the graph of order $n-1$ by collapsing together vertices connected by these paths. In doing this, we must account for three main contributions: (1) how many loops are contained solely within $\mathbb{C}_{1,2}$---each of these loops, of course, leads to a factor of $k$ multiplied by the contribution at order $n-1$; (2) what edges are erased when integrating out the case, as well as what edges are created after collapsing the paths into new edges---this tells us what $\vec{b}$ at lower order contribute to $\vec{a}$ at a higher order; (3) a combinatorial factor accounting for the fact that integrating out $\mathbb{C}_{1,2}$ in multiple graphs at order $n$ could lead to the same graph at order $n-1$, meaning we may need to multiply the contributions at order $n-1$ by something to get the correct final answer. The former loop calculation is usually quite simple, but the latter vectorial and combinatorial calculations require more significant casework. 
 
In the abstract, this is quite complicated, but we explain it more thoroughly through detailed examples as we proceed. We group our analysis of these cases into four categories corresponding to the number of edges, i.e.~$0$, $2$, $4$, or $6$, that protrude from the cases: $(1)$--$(4)$, $(5)$--$(12)$, $(13)$--$(16)$, and $(17)$, respectively. However, as mentioned, we begin with the base cases, to which we turn now. 

\subsection{Base Cases for Recursion}
\label{app:base-case}
Here we calculate the base cases for the recursion; that is, we determine all valid $\vec{a}$ when $n = 1$, construct all graphs with each $\vec{a}$, and count their connected components. Recall that the vector $\vec{a}$ must satisfy non-negativity, pairwise sums being even, and pairwise sums being at most $2n$; should any one of these conditions not be met, then $g(n,\vec{a}) = g(n,a_{12},a_{13},a_{23}) = 0$. For $n = 1$, there are $5$ possible options for $\vec{a}$: $(0,0,0)$, $(2,0,0)$, $(0,2,0)$, $(0,0,2)$, $(1,1,1)$. It remains then to construct the graphs and count their connected components. This is tedious, but the diagrams are shown in \cref{fig:base_1-3,fig:base_4}, and the final results are
\begin{align}
    g(1,0,0,0) &= 2k^{2}+2k, \\
    g(1,2,0,0) &= k^{3} + 3k^{2} + 4k, \\ 
    g(1,0,0,2) &= k^{3} + 3k^{2} + 4k, \\
    g(1,0,2,0) &= 2k^{2} + 6k, \\
    g(1,1,1,1) &= 2k^{3} + 14k^{2} + 16k.
\end{align}
This completes the base cases, and we now move on to the recursion.

\begin{figure}
    \centering
    \includegraphics[width = \linewidth]{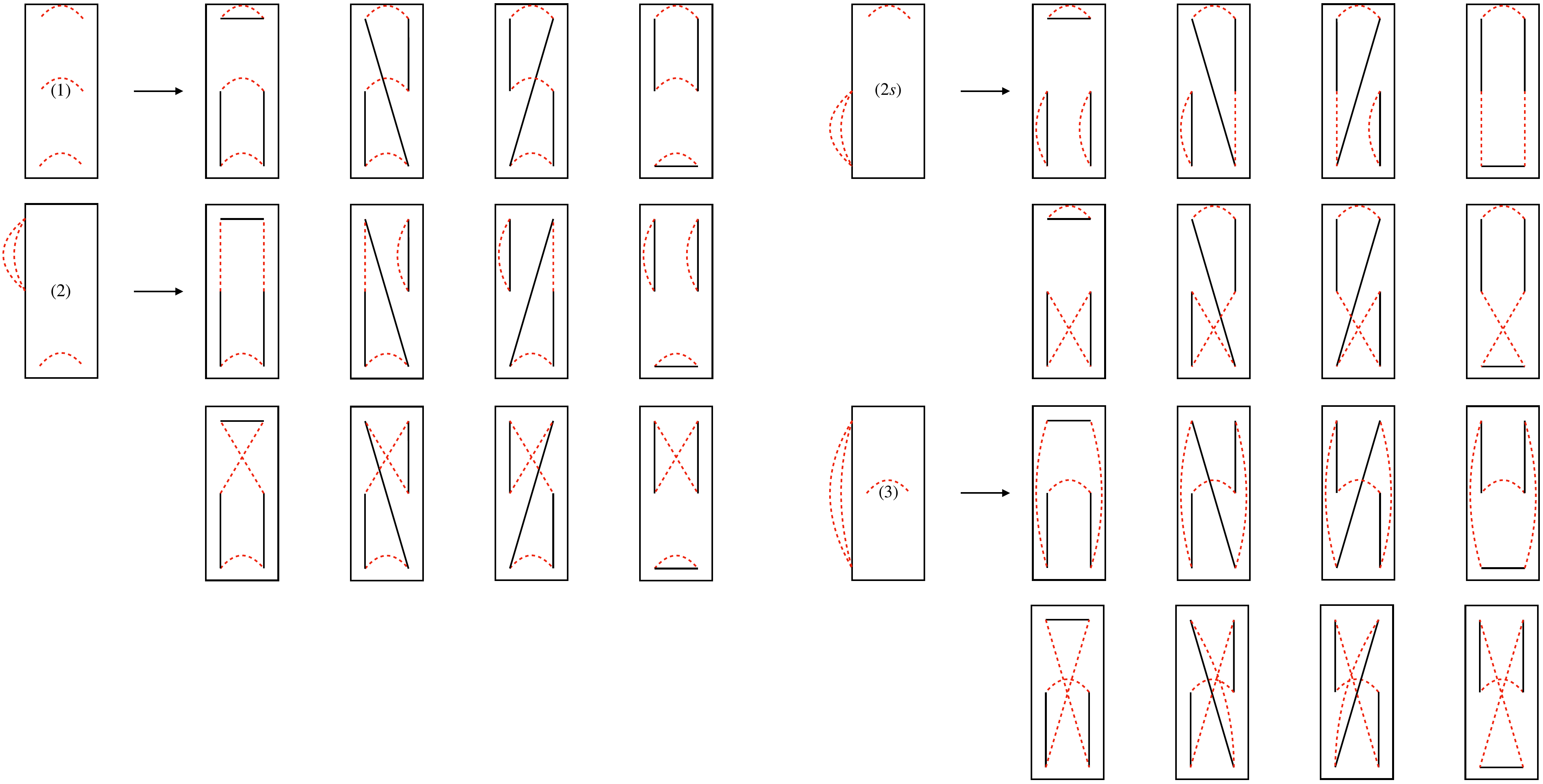}
    \caption{Base cases corresponding to $(1)$, $(2)$, $(2s)$, and $(3)$. Counting the connected components of the graphs in each case yields contributions of $2k^{2}+2k$, $k^{3}+3k^{2}+4k$, $k^{3}+3k^{2}+4k$, and $2k^{2}+6k$, respectively.}
    \label{fig:base_1-3}
\end{figure}

\begin{figure}
    \centering
    \includegraphics[width=\linewidth]{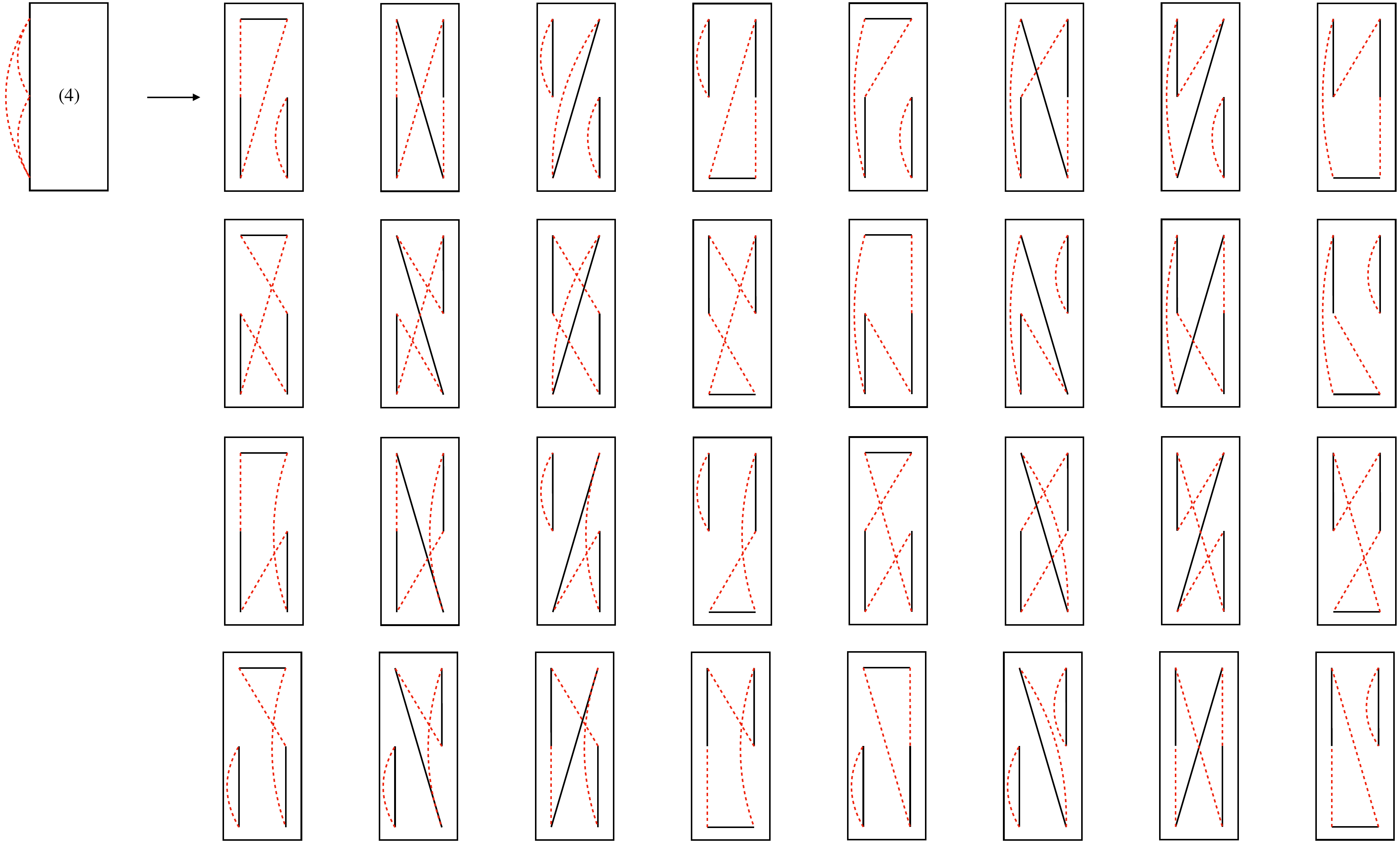}
    \caption{Base case corresponding to $(4)$. Counting the connected components of the graphs in each case yields $2k^{3}+14k^{2}+16k$.}
    \label{fig:base_4}
\end{figure}

\subsection{Cases \texorpdfstring{$(1)$--$(4)$}{(1)-(4)}}
We now handle cases $(1)$--$(4)$. There are no protruding edges, meaning many of the contributions are easy to derive because these cases are ``independent'' of from the lower order graph consisting of the final $n-1$ pairs columns. Therefore, when we integrate out $\mathbb{C}_{1,2}$, none of the paths affect the graph at lower order, meaning it is much simpler to calculate their contribution.

In fact, it is simple to see that the evaluation of the loops mimics exactly the calculation of the base cases:
\begin{align}
    \text{Loop } (1) & \to 2k^{2}+2k,\\
    \text{Loop }  (2) & \to k^{3}+3k^{2}+4k, \\
    \text{Loop }  (2s) & \to k^{3}+3k^{2}+4k, \\
    \text{Loop }  (3) & \to 2k^{2} + 6k, \\
    \text{Loop }  (4) & \to 2k^{3}+14k^{2}+16k. \\
\end{align}

Next, examining the diagrams for each case, one can derive simple relationships between $\vec{a}$ and $\vec{b}$ that yield a nontrivial contribution in \cref{eqn:app_recursion}:
\begin{align}
    \text{Vector } (1) & \to (b_{12}, b_{13}, b_{23}) = (a_{12}, a_{13}, a_{23}), \\
    \text{Vector } (2) & \to (b_{12}, b_{13}, b_{23}) = (a_{12}-2, a_{13}, a_{23}), \\
    \text{Vector } (2s) & \to (b_{12}, b_{13}, b_{23}) = (a_{12}, a_{13}, a_{23}-2), \\
    \text{Vector } (3) & \to (b_{12}, b_{13}, b_{23}) = (a_{12}, a_{13}-2, a_{23}), \\
    \text{Vector } (4) & \to (b_{12}, b_{13}, b_{23}) = (a_{12}-1, a_{13}-1, a_{23}-1). \\
\end{align}
These can be understood by looking at the diagram for each case and observing what kind of edges are eliminated when collapsing all of the paths that pass through the vertices in $\mathbb{C}_{1,2}$.

Finally, there are no combinatorial contributions because there are no protruding edges that have to be connected to the existing graph. That is, any graph that comes from integrating out one of these cases arises uniquely. 

Therefore, we can easily combine everything to get the contributions to the recursion from each of these cases: 
\begin{align}
    g(n,a_{12},a_{13},a_{23})_{\text{case}(1)} &= (2k^{2}+2k)g(n-1,a_{12},a_{13},a_{23}), \\
    g(n,a_{12},a_{13},a_{23})_{\text{case}(2)} &= (k^{3}+3k^{2}+4k)g(n-1,a_{12}-2,a_{13},a_{23}),\\
    g(n,a_{12},a_{13},a_{23})_{\text{case}(2s)} &= (k^{3}+3k^{2}+4k)g(n-1,a_{12},a_{13},a_{23}-2),\\
    g(n,a_{12},a_{13},a_{23})_{\text{case}(3)} &= (2k^{2}+6k)g(n-1,a_{12},a_{13}-2,a_{23}),\\
    g(n,a_{12},a_{13},a_{23})_{\text{case}(4)} &= (2k^{3}+14k^{2}+16k)g(n-1,a_{12}-1,a_{13}-1,a_{23}-1).
\end{align}
Note that we have introduced a notation $g(n,a_{12},a_{13},a_{23})_{\text{case}(i)}$, which simply refers to the contribution to $g(n,a_{12},a_{13},a_{23})$ from graphs where the vertices in $\mathbb{C}_{1,2}$ and their corresponding red edges fall into case $(i)$. That is, $g(n, \bm a) = \sum_{i \in \text{cases}} g(n, \bm a)_{\text{case}(i)}$.

\subsection{Cases \texorpdfstring{$(5)$--$(12)$}{(5)-(12)}}
We now tackle cases $(5)$--$(12)$, which have two edges that protrude and attach to the rest of the graph. Because of these two protruding edges, we have to carefully derive all three of the loop, vectorial, and combinatorial contributions. 

We start with the vectorial contributions, as understanding them allows us to more easily explain and derive the loop and combinatorial contributions. We start by carefully walking through case $(5)$, which contains two edges protruding from the first row. We take an existing graph of order $n$ where $\mathbb{C}_{1,2}$ and the respective red edges match case $(5)$. We then count how the numbers of edges of each type change after collapsing all of the paths that pass through the vertices in $\mathbb{C}_{1,2}$ into edges that lie within the other $2(n-1)$ columns.  

Now, it is crucial to observe the following extremely important fact for \emph{all cases} $(5)$--$(12)$: the two protruding edges are always part of the same path that goes through $\mathbb{C}_{1,2}$, regardless of which of the four types of black edges are present between the vertices in $\mathbb{C}_{1,2}$. Therefore, when $\mathbb{C}_{1,2}$ is integrated out in graphs that match these cases, the edge that is created in the lower order graph is simply given by the two rows upon which those protruding edges are incident. That is, if the protruding edges connected to rows $i$ and $j$, then, after integrating, an edge of type $ij$ is created. 

Now, there are, of course, $6$ types of edges that can be created by collapsing a path: $11$, $22$, $33$, $12$, $13$, and $23$. However, it is somewhat convenient to actually describe 9 possible edges, $11$, $22$, $33$, $12$, $13$, $23$, $21$, $31$, and $32$. The last three are equivalent to $12$, $13$, and $23$ edges, respectively, but we order the edges in this way to account for the two possible ways that the protruding edges can connect into the graph (that is, \emph{which} edge connects to row $i$ or $j$, for example). Note that this separation is extraneous for certain cases, i.e.~those with two edges protruding from the \emph{same} row, but it is useful when considering cases with edges protruding from different rows. 

To determine the vector contribution for a graph of order $n$ with $a_{12}$, $a_{13}$, and $a_{23}$ edges, we consider what edges $b_{12}$, $b_{13}$, and $b_{23}$ on the graph of order $n-1$ remain after integrating out $\mathbb{C}_{1,2}$. Case $(5)$ has two protruding edges coming from the first row, and then additional red edges of type $22$ and $33$. These $22$ and $33$ edges do not change the $12$, $13$, or $23$ edge counts. Therefore, the only changes come from the collapse of the path associated with the two protruding edges from row $1$. 

Let us say that these two protruding edges are originally incident on rows $2$ and $3$. In this example, this means that when integrating out $\mathbb{C}_{1,2}$, we lose one edge of type $12$ and one of type $13$, but we \emph{create} one of type $23$. Therefore, we must have that $b_{12} = a_{12} - 1$, $b_{13} = a_{13}-1$, and $b_{23} = a_{23}+1$. Or, if we define $\Delta_{ij} := b_{ij}-a_{ij}$, then $(\Delta_{12}, \Delta_{13}, \Delta_{23}) = (-1,-1,+1)$. We then consider all possible vertices that these two protruding edges could have been connected to in the remainder of the graph, and that defines all possible $g(n-1, b_{12}, b_{13}, b_{23})$ that can contribute to $g(n,a_{12},a_{13},a_{23})_{\text{case}(5)}$.

Now, we must also consider some combinatorial factors $\mathcal{C}$. The combinatorial factors are really just a shorthand for determining how many times a contribution $g(n-1,b_{12},b_{13},b_{23})$ shows up when integrating out a given case, here case $(5)$, from all the relevant graphs of order $n$. This is because different graphs at order $n$, when appropriately collapsed, lead to the same graph at order $n-1$. The combinatorial factor, then, is just a way of encoding this information. 

Say that we are again considering an example where the original protruding edges attach to vertices in rows $2$ and $3$. Then an edge of type $23$ is created. But if we look from the perspective of the lower order graph, \emph{any} of the $23$ edges could have been the one that was generated---that is, for some graph of order $n$ with case $(5)$ integrated out, a different $23$ edge that is present is the one generated. Therefore, when we sum up all the contribution from integrating out case $(5)$ over all relevant graphs of order $n$, we get a factor of $b_{23}$. Note also that, as we derived above, $b_{23} = a_{23}+1$. Also note that, were we looking at protruding edges attached to the same row, we would get an additional factor of $2$ due to the ambiguity of which edge attaches to which endpoint.

Finally, we consider the loop contribution. The calculation for case $(5)$ is a relatively straightforward diagrammatic proof, which is detailed in \cref{fig:case_5_loop}. In short, we draw all possible diagrams consistent with case $(5)$ and count up the loops that are induced. There are only four cases, as the red edges are essentially fixed and there are four possible sets of black edges. The result is a factor $2k+2$. That is, there are two sets of black edges that lead to an internal loop, leading to an extra factor of $k$, and there are two sets of black edges where the protruding edges snake through all vertices in $\mathbb{C}_{1,2}$ such that collapsing them just leads to a graph of order $n-1$ without any extra loop factors.
\begin{figure}[ht!]
    \centering
    \includegraphics[width=\linewidth]{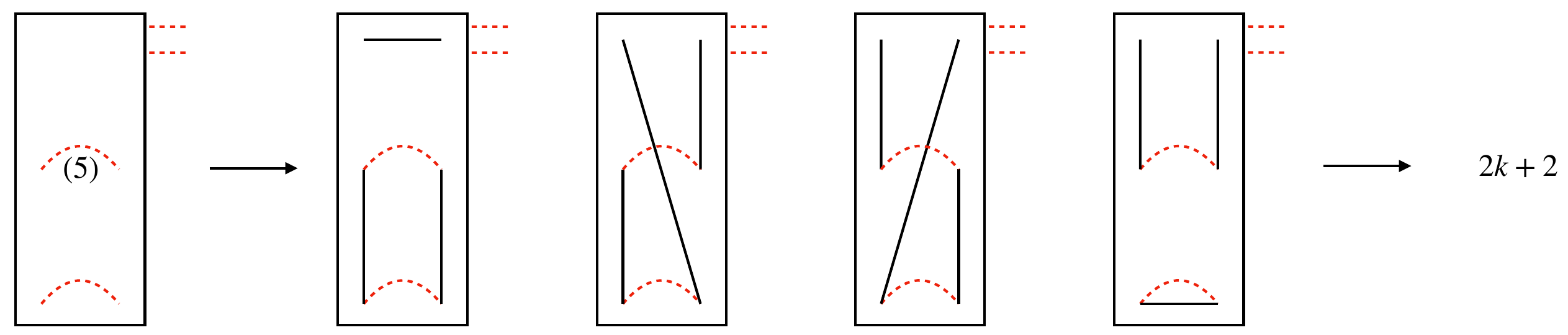}
    \caption{Loop contribution for case $(5)$.}
    \label{fig:case_5_loop}
\end{figure}

So, putting all of the information together, we have that a full contribution from case $(5)$ is
\begin{equation}\label{eqn:case_5_recursion}
\begin{split}
    g(n,a_{12},a_{13},a_{23})_{\textrm{case}(5)} &= (2k+2)[(2b_{11}+2b_{12}+2b_{13})g(n-1,a_{12},a_{13},a_{23})\\
    &\quad + 2b_{22}g(n-1,a_{12}-2,a_{13},a_{23}) \\
    &\quad + 2b_{23}g(n-1,a_{12}-1,a_{13}-1,a_{23}+1) \\
    &\quad +2b_{33}g(n-1,a_{12},a_{13}-2,a_{23})]
    \\
    &=(2k+2)[ (2(n-1)+a_{12}+a_{13})g(n-1,a_{12},a_{13},a_{23}) \\
    &\quad + (2(n-1)-(a_{12}-2)-a_{23})g(n-1,a_{12}-2,a_{13},a_{23}) \\
    &\quad + 2(a_{23}+1)g(n-1,a_{12}-1,a_{13}-1,a_{23}+1) \\
    &\quad + (2(n-1)-(a_{13}-2)-a_{23})g(n-1,a_{12},a_{13}-2,a_{23})].
\end{split}
\end{equation}
This includes the loop, combinatorial, and vectorial factors. We also note that, should any of the combinatorial factors actually be negative, they should be set to 0, as that indicates that the graph that is constructed at lower order when integrating out the given case does not really exist (this is also handled by the vector input to $g$ being negative---that is, one of the edge counts $b_{12},b_{13},b_{23}$ is negative). One can get the contribution from case $(5s)$ by simply mapping $1\leftrightarrow3$. 

We list the combinatorial and vectorial contributions for cases $(5)$--$(8)$ in \cref{tab:cases_5-8} and cases $(9)$--$(12)$ in \cref{tab:cases_9-12} (the main difference in the latter cases is that there is no longer a symmetry between red edges attaching to vertices $ij$ and $ji$ because, by convention, we attach the top protruding edge to the vertex in row $i$ and the bottom protruding edge to the vertex in row $j$, which gives us different types of new edges, generically). The first column of these tables gives what kind of edge is created at order $n-1$. The second column tells us the combinatorial factor. The next four multicolumns give the vector information for each of the cases. Note that we do not give the symmetric cases, as they can be obtained by simply mapping $1\leftrightarrow3$.  

\begin{table}[ht!]
\centering
\def\arraystretch{1.2}
\begin{tabular}{cc|ccc|ccc|ccc|ccc|}
            &           & \multicolumn{3}{c|}{\includegraphics[height=0.2\textwidth]{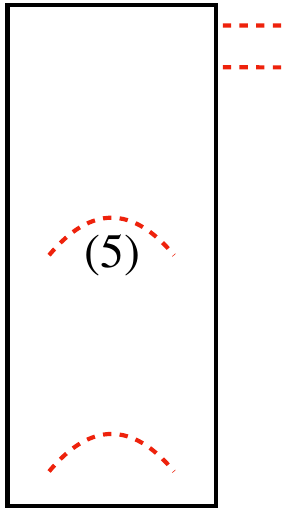}} & \multicolumn{3}{c|}{\includegraphics[height=0.2\textwidth]{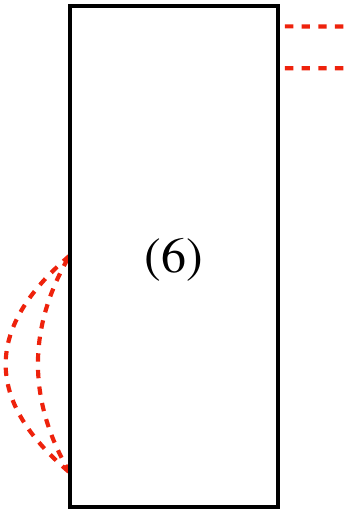}} & \multicolumn{3}{c|}{\includegraphics[height=0.2\textwidth]{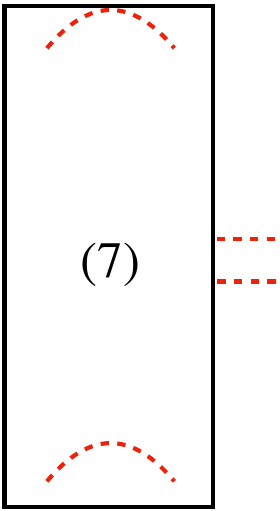}} & \multicolumn{3}{c|}{\includegraphics[height=0.2\textwidth]{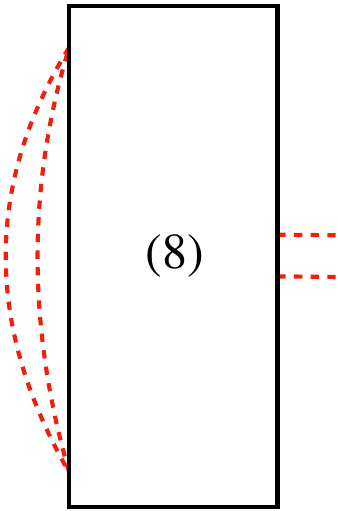}} \\ \cline{3-14} 
Protruding Endpoints & $\mathcal{C}$       & $\Delta_{12}$  & $\Delta_{13}$ & $\Delta_{23}$ & $\Delta_{12}$  & $\Delta_{13}$ & $\Delta_{23}$ & $\Delta_{12}$  & $\Delta_{13}$ & $\Delta_{23}$ & $\Delta_{12}$  & $\Delta_{13}$ & $\Delta_{23}$ \\ \cline{3-14} 
11          & $2b_{11}$ & 0         & 0        & 0        & 0         & 0        & -2       & -2        & 0        & 0        & -2        & -2       & 0        \\
12          & $b_{12}$  & 0         & 0        & 0        & 0         & 0        & -2       & 0         & 0        & 0        & 0         & -2       & 0        \\
13          & $b_{13}$  & 0         & 0        & 0        & 0         & 0        & -2       & -1        & +1       & -1       & -1        & -1       & -1       \\
21          & $b_{12}$  & 0         & 0        & 0        & 0         & 0        & -2       & 0         & 0        & 0        & 0         & -2       & 0        \\
22          & $2b_{22}$ & -2        & 0        & 0        & -2        & 0        & -2       & 0         & 0        & 0        & 0         & -2       & 0        \\
23          & $b_{23}$  & -1        & -1       & +1       & -1        & -1       & -1       & 0         & 0        & 0        & 0         & -2       & 0        \\
31          & $b_{13}$  & 0         & 0        & 0        & 0         & 0        & -2       & -1        & +1       & -1       & -1        & -1       & -1       \\
32          & $b_{23}$  & -1        & -1       & +1       & -1        & -1       & -1       & 0         & 0        & 0        & 0         & -2       & 0        \\
33          & $2b_{33}$ & 0         & -2       & 0        & 0         & -2       & -2       & 0         & 0        & -2        & 0         & -2       & -2      
\end{tabular}
\caption{Information for vectorial and combinatorial contributions to cases $(5)$--$(8)$. Observe that there is a symmetry when the endpoints of the protruding edges are $ij$ and $ji$. Also observe that, when the endpoints are the same, i.e.~$ii$, there is an extra factor of $2$ in the combinatorial term because of the ambiguity between how the protruding edges originally attach.}
\label{tab:cases_5-8} 
\end{table}
\begin{table}[ht!]
\def\arraystretch{1.05}
\centering
\begin{tabular}{cc|ccc|ccc|ccc|ccc|}
            &           & \multicolumn{3}{c|}{\includegraphics[height=0.2\textwidth]{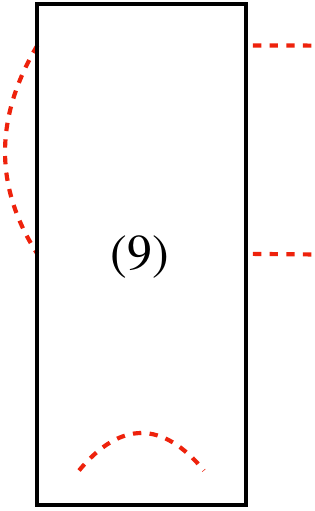}} & \multicolumn{3}{c|}{\includegraphics[height=0.2\textwidth]{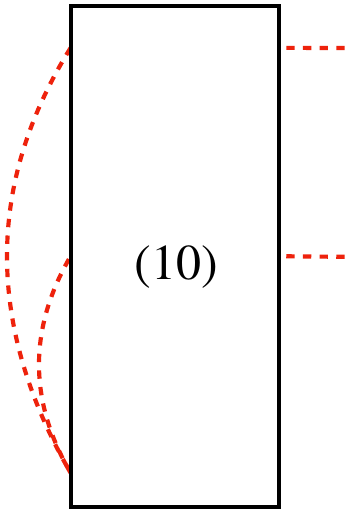}} & \multicolumn{3}{c|}{\includegraphics[height=0.2\textwidth]{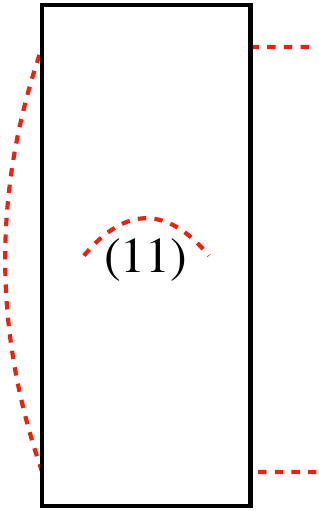}} & \multicolumn{3}{c|}{\includegraphics[height=0.2\textwidth]{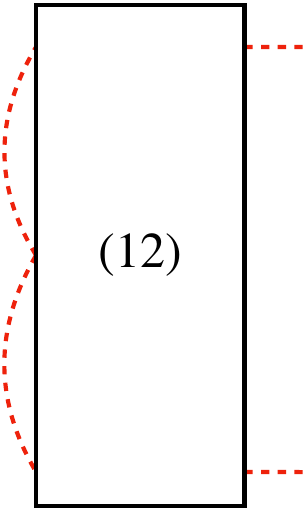}} \\ \cline{3-14} 
Protruding Endpoints & $\mathcal{C}$       & $\Delta_{12}$  & $\Delta_{13}$ & $\Delta_{23}$ & $\Delta_{12}$  & $\Delta_{13}$ & $\Delta_{23}$ & $\Delta_{12}$  & $\Delta_{13}$ & $\Delta_{23}$ & $\Delta_{12}$  & $\Delta_{13}$ & $\Delta_{23}$ \\ \cline{3-14} 
11          & $2b_{11}$ & -2        & 0        & 0        & -1        & -1       & -1       & 0         & -2       & 0        & -1        & -1       & -1       \\
12          & $b_{12}$  & 0         & 0        & 0        & +1        & -1       & -1       & +1        & -1       & -1       & 0         & 0        & -2        \\
13          & $b_{13}$  & -1        & +1       & -1       & 0         & 0        & -2       & 0         & 0        & 0        & -1        & +1       & -1        \\
21          & $b_{12}$  & -2        & 0        & 0        & -1        & -1       & -1       & 0         & -2       & 0        & -1        & -1       & -1       \\
22          & $2b_{22}$ & -2        & 0        & 0        & -1        & -1       & -1       & -1        & -1       & -1       & -2        & 0        & -2        \\
23          & $b_{23}$  & -2        & -0       & 0       & -1        & -1       & -1       & -1        & -1       & +1       & -2        & 0        & 0        \\
31          & $b_{13}$  & -2        & 0        & 0        & -1        & -1       & -1       & 0         & -2       & 0        & -1        & -1       & -1       \\
32          & $b_{23}$  & -1        & -1       & +1       & 0         & -2       & 0        & 0         & -2       & 0        & -1        & -1       & -1        \\
33          & $2b_{33}$ & -1        & -1       & -1       & 0         & -2       & -2       & 0         & -2       & 0        & -1        & -1       & -1      
\end{tabular}
\caption{Information for vectorial and combinatorial contributions to cases $(9)$--$(12)$. Observe that there is no longer a symmetry between $ij$ and $ji$, but the $ii$ cases still have an extra factor of $2$ in the combinatorial term because the ambiguity between how the protruding edges originally attach still exists.}
\label{tab:cases_9-12} 
\end{table}

We also provide the loop contributions for cases $(5)$--$(12)$ in \cref{tab:loops_(5)-(12)}. These are derived in an analogous way to the diagrammatic approach in \cref{fig:case_5_loop}, but there are many more graphs to consider. 
\begin{table}[ht!]
    \def\arraystretch{1.05}
    \centering
    \begin{tabular}{c|c}
         Case & Loop contribution \\
         \hline
         $(5)$ & $2k + 2$  \\
         $(5s)$ & $2k + 2$  \\
         $(6)$ & $k^{2}+3k+4$ \\
         $(6s)$ & $k^{2}+3k+4$ \\
         $(7)$ & $2k + 2$  \\
         $(8)$ & $2k + 6$ \\
         $(9)$ & $2k^{2} + 6k + 8$  \\
         $(9s)$ & $2k^{2} + 6k + 8$  \\
         $(10)$ & $2k^{2}+14k+16$ \\
         $(10s)$ & $2k^{2}+14k+16$ \\
         $(11)$ & $4k + 12$  \\
         $(12)$ & $2k^{2}+14k+16$ \\
    \end{tabular}
    \caption{Loop contributions for each of the cases $(5)$--$(12)$. Notice that symmetric versions of cases have the same loop contribution; only their vectorial and combinatorial contributions are different.}
    \label{tab:loops_(5)-(12)}
\end{table}
Therefore, using all of this information, we can derive an equivalent version of \cref{eqn:case_5_recursion} for each case up to $(12)$ (including the symmetric ones), accounting for all of their contributions. 

\subsection{Cases \texorpdfstring{$(13)$--$(16)$}{(13)-(16)}}
We now move on to more complicated cases that have four protruding edges. 
The vectorial contribution is more difficult to calculate, as we must account for $3^{4} = 81$ possibilities for how the protruding edges attach to the lower order graph. Furthermore, there is more interaction between the vectorial, combinatorial, and loop terms. This did not occur in the previous sets of cases because the protruding edges were always part of the same path through the black edges attached to the vertices in $\mathbb{C}_{1,2}$.
However, one must now keep track of which protruding edges connect to one another through the vertices in $\mathbb{C}_{1,2}$. 

For example, we look at the possibilities for case $(13)$, shown in \cref{fig:case_13}. 
\begin{figure}[ht!]
    \centering
    \includegraphics[width=\linewidth]{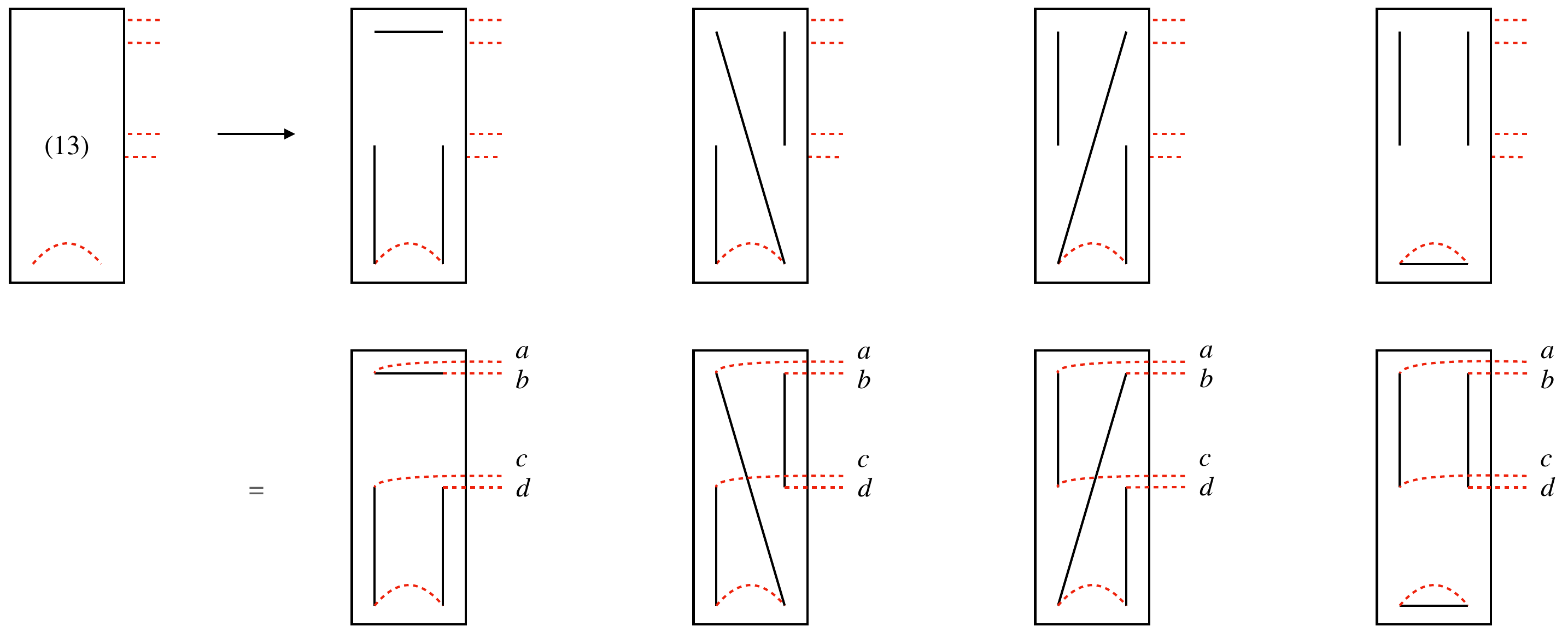}
    \caption{Evaluation of case $(13)$. By convention, we take the top left vertex to row $a$, the top right vertex to row $b$, the middle left vertex to row $c$, and the middle right vertex to row $d$, where $a,b,c,d\in\{1,2,3\}$. The types of edges that are created after integrating out the two leftmost columns are determined by the type of the black edges.}
    \label{fig:case_13}
\end{figure}
By convention, we take the top left vertex to row $a$, the top right vertex to row $b$, the middle left vertex to row $c$, and the middle right vertex to row $d$, where $a,b,c,d\in\{1,2,3\}$. We see that, when the black edges attached to the vertices in $\mathbb{C}_{1,2}$ are type-1, then the red edges that protrude from the top row are connected to one another, which means that one generates an edge of type $ab$ when collapsing this path. However, if the black edges associated with $\mathbb{C}_{1,2}$ are type-2, then it is instead $ac$ and $bd$ that are connected. In total, one of the possible types of black edges connect edges $ab$ and $cd$, and three connect $ac$ and $bd$. In the case where $ab$ and $cd$ are connected, this means that we generate edges of type $ab$ and $cd$ but we lose edges of type $1a, 1b, 2c, 2d$. When $ac$ and $bd$ are connected, we of course gain edges of type $ac$ and $bd$, but we still lose edges of type $1a, 1b, 2c, 2d$. We use these observations to build up the vectorial contribution of the graph by summing over all 81 possibilities of $a,b,c,d\in\{1,2,3\}$. This is tedious to do by hand, but simple numerically. 

We need also account for the loop and combinatorial factors that associate to each of these vectorial contributions. Luckily, we do not need to consider 81 cases parameterized by $a,b,c,d$, but we must consider each of the subcases defined by the four possible sets of black edges in connecting the vertices in $\mathbb{C}_{1,2}$. Loop-wise, we simply need to count how many loops are induced. Working from the left to right in \cref{fig:case_13}, we get $0,0,0,1$ loops, respectively, leading to factors of $1,1,1,k$, respectively. The combinatorial factor is given by
\begin{equation}\label{eqn:combinatorial_13-16}
 2^{\delta_{ab}}2^{\delta_{cd}}\bkt*{(\delta_{ac}\delta_{bd} + \delta_{ad}\delta_{bc}-\delta_{abcd})2\binom{b_{ab}}{2} + (1-(\delta_{ac}\delta_{bd} + \delta_{ad}\delta_{bc}-\delta_{abcd}))b_{ab}b_{cd}}
\end{equation}
in the case where edges $ab$ and $cd$ are connected. If instead $ac$ and $bd$ are connected, we replace each instance of $ab$ and $cd$ with $ac$ and $bd$, respectively. We then again account for all 81 cases and attach each combinatorial factor and loop factor to its associated vectorial term. 

To understand \cref{eqn:combinatorial_13-16}, consider the following, where we assume we are dealing with type-1 black edges so that we are creating edges $ab$ and $cd$. We get a factor of $2$ when $a$ and $b$ are the same because they correspond to protruding edges coming from the same row, meaning there is a choice of which edge to connect where. The same holds for $c$ and $d$. If all four edges connect to the same row, i.e.~$a=b=c=d$, then one might naively think we need to add an extra factor of $6$ (to get to a total of $4!$ possible connections), but this is incorrect, as $ab$ and $cd$ are always paired given their connection through case $(13)$ with black edges of type-1. Now, if $a=c$ and $b=d$ or $a=d$ and $b=c$, then the two edges $ab$ and $cd$ are the same type, meaning we are creating two edges of the same type in the graph of order $n-1$. There are therefore $\binom{b_{ab}}{2}$ choices of which edges these are in the lower order graph, but we also need an extra factor of 2 to decide which one the groups of protruding edges each maps to. If $ab$ and $cd$ correspond to different types of edges, then we just get a factor of $b_{ab}b_{cd}$, as we simply need to account for which of these edges are generated through the integration process. 

Therefore, we see that cases $(13)$--$(16)$ raise substantially more complications in their evaluation. In particular, the type of black edges leads to far more interaction between the loop, vectorial, and combinatorial contributions that must be carefully combined in code to achieve the correct recursion. While we have only described case $(13)$ in detail, cases $(14)$--$(16)$ follow in the exact same manner, though there are more graphs to consider in the cases where two rows have only one protruding edge. 

\subsection{Case \texorpdfstring{$(17)$}{(17)}}
Case $(17)$ raises the same issues, though there are only four graphs to consider. However, we have $243 = 3^{6}$ possible options for how the protruding edges may connect to the graph at lower order (this is true in general, but not all of these are possible when $n$ is small). See \cref{fig:case_17}. 
\begin{figure}[ht!]
    \centering
    \includegraphics[width=\linewidth]{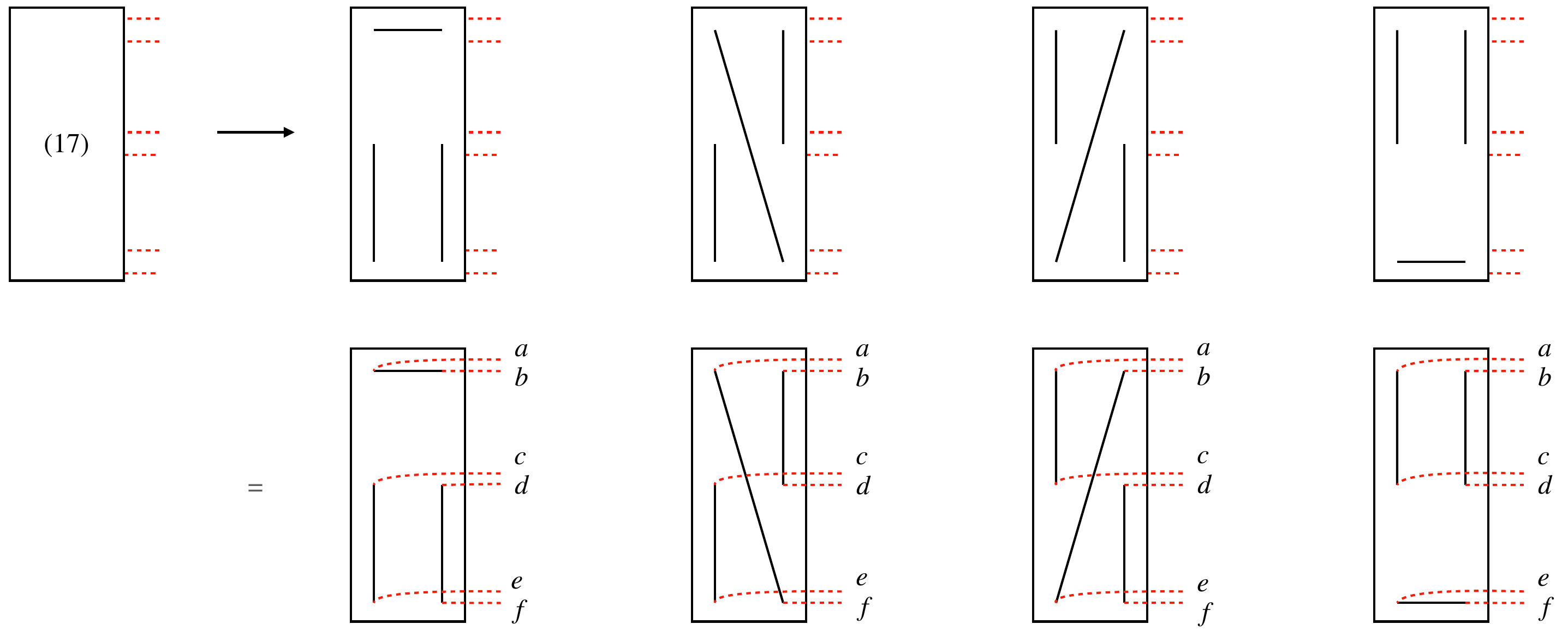}
    \caption{Evaluation of case $(17)$. We repeat the convention for cases $(13)$--$(16)$ by taking the top left vertex to row $a$, the top right vertex to row $b$, the middle left vertex to row $c$, and the middle right vertex to row $d$, but we now also take the bottom left to $e$ and the bottom right to $f$, where $a,b,c,d,e,f\in\{1,2,3\}$. The types of edges that are created are still determined by the type of the black edges.}
    \label{fig:case_17}
\end{figure}
We repeat the convention for cases $(13)$--$(16)$ by taking the top left vertex to row $a$, the top right vertex to row $b$, the middle left vertex to row $c$, and the middle right vertex to row $d$, but we now also take the bottom left to $e$ and the bottom right to $f$, where $a,b,c,d,e,f\in\{1,2,3\}$. Now, for type-1 black edges, we create $ab$, $ce$, and $df$; for type-2, it is $af$, $bd$, and $ce$; for type-3 it is $ac$, $be$, and $df$; and for type-4 it is $ac$, $bd$, and $ef$. We always lose edges of type $1a, 1b, 2c, 2d, 3e, 3f$ regardless of the type of the black edges. Furthermore, the loop contribution is always a factor of 1, as there are no internal loops to case $(17)$. 

The combinatorial factor, however, is quite complicated. Assume for now that we are working with type-1 black edges such that $ab$, $ce$, and $df$ are linked. The combinatorial factor is 
\begin{align}
    &(2^{\delta_{ab}})^{3}3!\binom{b_{ab}}{3} & \times 1\bkt*{\{a,b\} = \{c,e\} = \{d,f\}} \label{eqn:17_comb_1}\\
    +&2^{\delta_{ab}}2^{\delta_{ce}}\times2\binom{b_{ab}}{2}\times2^{\delta_{df}}b_{df} & \times 1\bkt*{\{a,b\} = \{c,e\} \neq \{d,f\}} \label{eqn:17_comb_2}\\
    +&2^{\delta_{ab}}2^{\delta_{df}}\times2\binom{b_{ab}}{2}\times2^{\delta_{ce}}b_{ce} & \times 1\bkt*{\{a,b\} = \{d,f\} \neq \{c,e\}} \label{eqn:17_comb_3}\\
    +&2^{\delta_{ce}}2^{\delta_{df}}\times2\binom{b_{ce}}{2}\times2^{\delta_{ab}}b_{ab} & \times 1\bkt*{\{a,b\} \neq \{c,e\} = \{d,f\}} \label{eqn:17_comb_4}\\
    +&2^{\delta_{ab}}2^{\delta_{ce}}2^{\delta_{df}}b_{ab}b_{ce}b_{df} &\times 1\bkt*{\{a,b\} \neq \{c,e\} \neq \{d,f\}}\label{eqn:17_comb_5}.
\end{align}
Here, $1[\text{A}]$ is an indicator function that is $1$ if statement A is true and $0$ if it is false. For example, $1[\{a,b\} = \{c,e\}, \{d,f\}]$ is $1$ if $\{a,b\}$, $\{c,e\}$, and $\{d,f\}$ are all equal as sets (that is, order does not matter). The middle three lines [\cref{eqn:17_comb_2,eqn:17_comb_3,eqn:17_comb_4}] are just repetitions of the combinatorial factors for cases $(13)-(16)$, but accounting for which sets of four edges may be sent to the same row. The last line [\cref{eqn:17_comb_5}] is simple and accounts for the case where all of the edge types $ab, ce, df$ are different. The first line [\cref{eqn:17_comb_1}] requires a bit of explanation. In the case where $a \neq b$, we simply have to choose three edges of type $ab$ where the order matters (they each could have been created by integrating out different graphs at a higher order). In the case where $a = b$, this is still the case, but now we need a factor of $2$ for each edge, as we can flip which vertices are connected where.

Again, it is hard to account for all of these elements by hand, but it is simple numerically. With this final case sorted out, we simply combine contributions of all of the cases $g(n,\vec{a})_{\textrm{case}(i)}$ to find $g(n,\vec{a})$.

\section{Computing Individual Coefficients}\label{app:individual_coefficients}
In this appendix, we discuss the various methods by which one can compute individual coefficients in the polynomial expansion of the second moment. Recall that, per \cref{thm:second_moment}, the second moment may be expanded as
\begin{equation}
    M_{2}(k,n) = (2n-1)!!\sum_{i=1}^{2n} c_{i}k^{i}.
\end{equation}
Ideally, one would simply be able to find a closed functional form for the right-hand side of this equation (as was possible for the equivalent definition of the first moment). But, unfortunately, such a result currently eludes us. Therefore, the best we can do is find individual coefficients. We now discuss methods of calculating $c_{2n}$ and $c_{2n-1}$.

\subsection{Leading Order Coefficient \texorpdfstring{$c_{2n}$}{c2n}}\label{sec:c2n}
We begin with the leading order coefficient $c_{2n}$. Recall that \cref{lemma:k}(ii) gives that $c_{2n} = (2n)!!$. The proof of this lemma is contained in the companion text Ref.~\cite{ehrenbergTransitionAnticoncentrationGaussian2023}, and we briefly describe that proof. However, we also provide a second technique for understanding the result that is useful to understanding the proof of the first sub-leading order term $c_{2n-1}$. 

Recall that, in order for a graph in $\mathbb{G}_{n}^{2} = \mathbb{G}_{n}^{2}(0,0,0)$ to have $2n$ connected components, it must possess only type-1 and type-4 black edges. The two vertices connected by each horizontal black edge must also be connected by a red edge to form a 2-vertex connected component. The remaining vertical edges from the type-1 sets of black edges are then paired off (i.e., connected via horizontal red edges) into 4-vertex connected components, and the same holds for black vertical edges from type-4 sets. This leads to $2n$ total connected components. The original proof that the total number of graphs satisfying these constraints is $(2n)!!$ proceeds by reducing these graphs to ones in $\mathbb{G}_{n}^{1}$ and then counting them (with a weight given by the number of connected components). This is evaluated by using the equation for the first moment in \cref{thm:first_moment}. See the companion piece Ref.~\cite{ehrenbergTransitionAnticoncentrationGaussian2023} for more explicit details. 

Another way to compute this coefficient is by making a combinatorial argument. As discussed, $c_{2n}$ contains contributions only from graphs that possess solely type-1 and type-4 sets of black edges. Again, in order to create the maximal number of connected components, the horizontal black edges must also be connected by red edges to create a size-2 connected component. The remaining type-1 vertical black edges are paired off, and the type-4 vertical black edges are similarly paired off. So, for a graph of order $n$, say that there are $p$ sets of type-1 black edges and, therefore, $n-p$ sets of type-4 black edges. There are $\binom{n}{p}$ sets of black edges with this type distribution. There are then $(2p-1)!!$ ways to pair off the $2p$ vertical type-1 black edges, and $(2n-2p-1)!!$ ways to pair off the $2n-2p$ vertical type-4 black edges. Therefore, summing over $p \in \{0, 1, \dots, n\}$, we get that
\begin{equation}
    c_{2n} = \sum_{p=0}^{n} \binom{n}{p}(2p-1)!!(2n-2p-1)!!.
\end{equation}
We can massage the right-hand side a bit using the fact that $(2x-1)!! = (2x)!/(2x)!! = (2x)!/(2^{x}x!)$. Expanding out the binomial coefficient and converting all terms to single factorials yields
\begin{equation}
    c_{2n} = \frac{n!}{2^{n}}\sum_{p=0}^{n} \binom{2p}{p}\binom{2n-2p}{n-p}.
\end{equation}
The summation evaluates to $4^{n}$ using the convolution of the Taylor series for $(1-4x)^{-1/2}$ \cite{stackexchange_binomial_summation}. Therefore, 
\begin{equation}
    c_{2n} = 2^{n}n! = (2n)!!,
\end{equation}
which, of course, matches the known result. 

\subsection{First Subleading Coefficient \texorpdfstring{$c_{2n-1}$}{c2n-1}}\label{sec:c2n-1}
We now generalize the above combinatorial version of the $c_{2n}$ calculation to $c_{2n-1}$. It is slightly more complicated, as there is a bit of casework to consider, but the general idea is the same. In particular, the key idea is that because $2n$ is the maximal number of connected components, finding a graph with $2n-1$ connected components comes down to counting the ways that one can create a ``deficit'' of exactly one connected component from the maximal number. There are nine ways to accomplish this. 

First, consider starting with graphs with a maximal number of connected components, meaning, as per \cref{sec:c2n}, they have only type-1 and type-4 black edges. The connected components have either 2 vertices (red and black edge between 2 vertices in the same row) or 4 (two vertical black edges of the same type that are paired off via red edges). We refer to these as type-$x$ 2-vertex and 4-vertex connected components, respectively (where $x$ is either 1 or 4).  One can convert these graphs with maximal connected components into graphs with a deficit of a single connected component in the following ways, all of which involve merging two connected components into a single one:
\begin{enumerate}[(1):]
    \item merge two type-1 2-vertex connected components;
    \item merge two type-4 2-vertex connected components;
    \item merge one type-1 2-vertex connected component with one type-4 4-vertex connected component;
    \item merge one type-4 2-vertex connected component with one type-1 4-vertex connected component;
    \item merge two type-1 4-vertex connected components;
    \item merge two type-4 4-vertex connected components;
    \item merge one type-1 4-vertex connected component with one type-4 4-vertex connected component.
\end{enumerate}
These options are visualized (up to the symmetry of exchanging the roles of type-1 and type-4 edges) in \cref{fig:c2n-1_1-7}.
\begin{figure}[ht!]
    \centering
    \includegraphics[width=\linewidth]{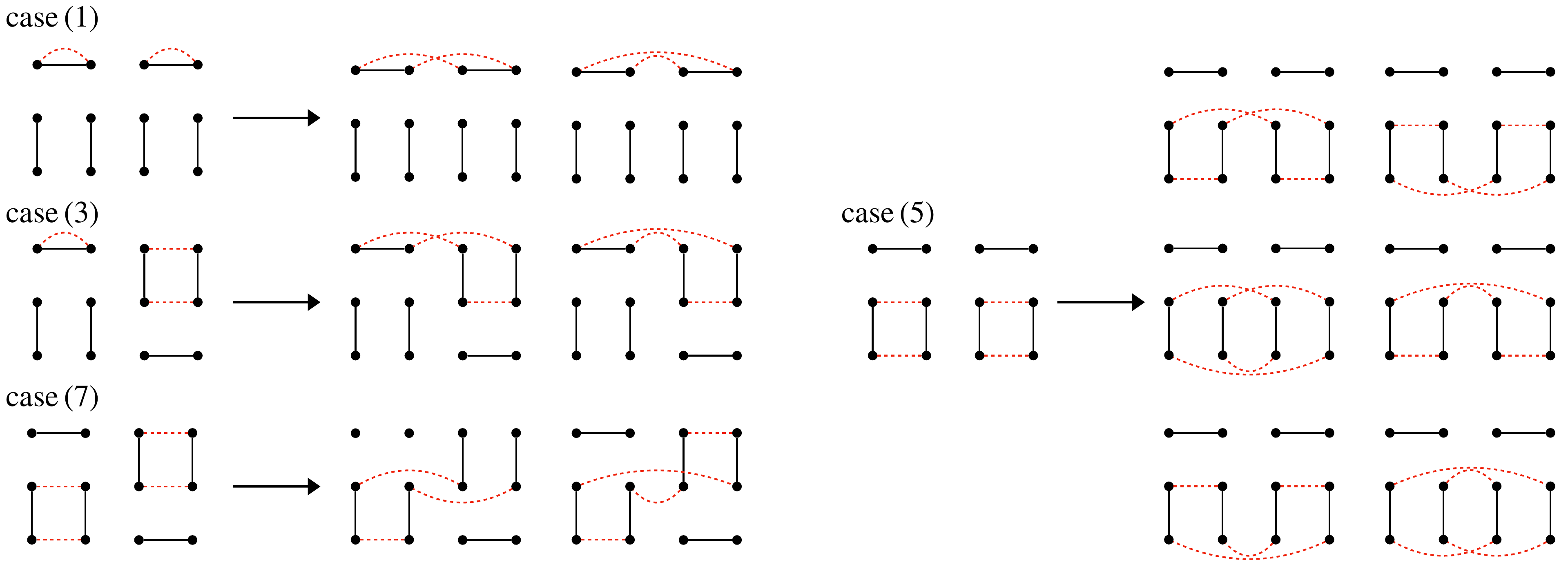}
    \caption{Possible ways of merging type-1 and type-4 vertices to create a deficit of a single connected component. Here, we only show cases $(1)$, $(3)$, $(5)$, and $(7)$, as $(2)$, $(4)$, and $(6)$ are symmetric with $(1)$, $(3)$, and $(5)$ with type-1 and type-4 edges switched.}
    \label{fig:c2n-1_1-7}
\end{figure}

Next, we must also consider cases with type-2 and type-3 black edges. There are two options here: either the graph can have exactly one set of type-2 or type-3 edges, or it can have exactly two sets (it does not matter whether it is two type-2 sets of edges, two type-3 sets of edges, or one of each). The rest of the sets of black edges must all be of type 1 or type 4. Then, creating a deficit can be done in the following ways:
\begin{enumerate}[(1):]
  \setcounter{enumi}{7}
    \item connect one type-2 or type-3 edge (the edge connecting the top row to the bottom row) to one type-1 vertical edge and one type-4 vertical edge to make a 6-vertex loop;
    \item connect two type-2 or type-3 edges (again, the top-to-bottom edges) to form a 4-vertex connected component.
\end{enumerate}
These are visualized in \cref{fig:c2n-1_8-9}.
\begin{figure}[ht!]
    \centering
    \includegraphics[width=0.75\linewidth]{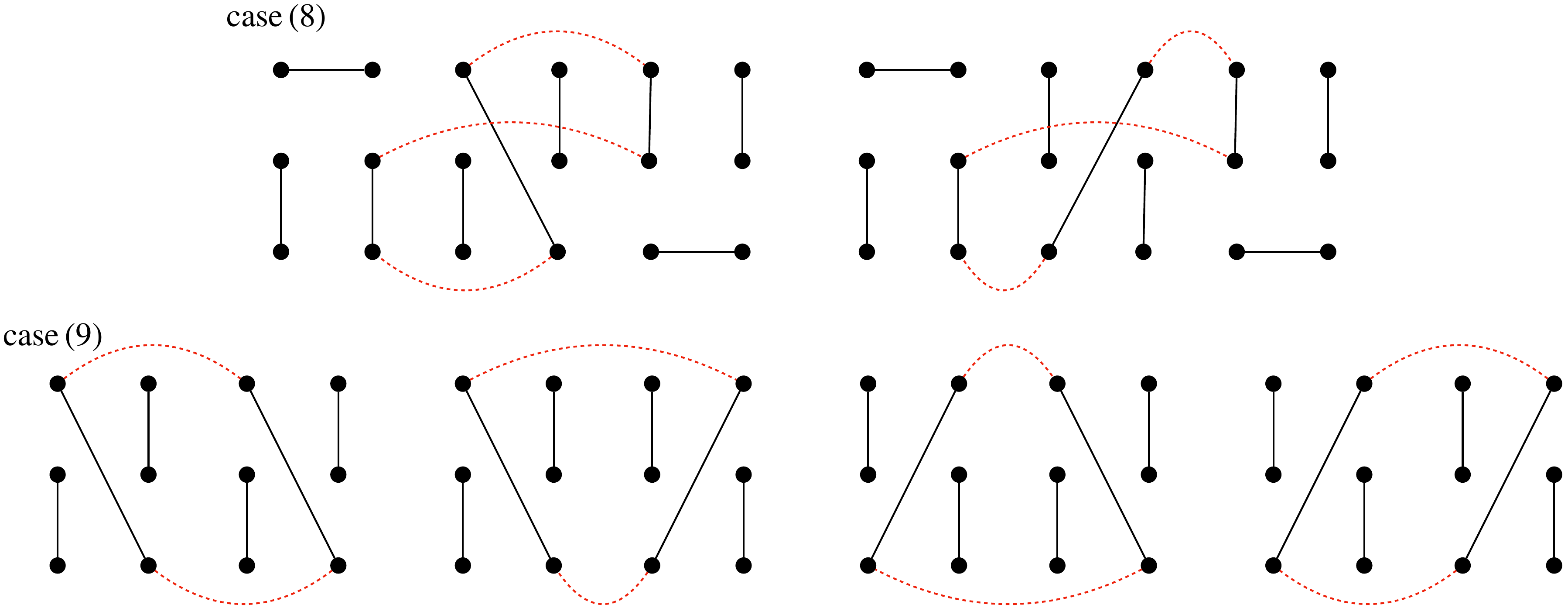}
    \caption{Possible ways of creating a deficit of a single connected component while using type-2 and/or type-3 edges.}
    \label{fig:c2n-1_8-9}
\end{figure}
The rest of horizontal black edges must be connected with red edges to form $2$-vertex connected components, and the remaining vertical edges must be appropriately paired off in order to ensure $2n-2$ other connected components are formed.

The end result of accounting for all of these cases is a (double) sum that computes $c_{2n-1}$:
\begin{equation}\label{eqn:c2n-1}
\begin{split}
    c_{2n-1} &= \sum_{p=0}^{n}\binom{n}{p} \times \bigg[ \underbrace{2\binom{p}{2}(2p-1)!!(2(n-p)-1)!!}_{(1)}+\underbrace{2\binom{n-p}{2}(2p-1)!!(2(n-p)-1)!!}_{(2)} \\
    &\phantom{\sum_{p=0}^{n}}+\underbrace{2\binom{p}{1}\binom{2(n-p)}{2}(2p-1)!!(2(n-p-1)-1)!!}_{(3)} + \underbrace{2\binom{n-p}{1}\binom{2p}{2}(2(p-1)-1)!!(2(n-p)-1)!!}_{(4)} \\
    &\phantom{\sum_{p=0}^{n}}+\underbrace{6\binom{2p}{4}(2(p-2)-1)!!(2(n-p)-1)!!}_{(5)} + \underbrace{6\binom{2(n-p)}{4}(2p-1)!!(2(n-p-2)-1)!!}_{(6)} \\
    &\phantom{\sum_{p=0}^{n}}+\underbrace{2\binom{2p}{2}\binom{2(n-p)}{2}(2(p-1)-1)!!(2(n-p-1)-1)!!}_{(7)} \bigg]\\
    &+ \sum_{p=0}^{n-1}\underbrace{2\binom{n}{1}\binom{n-1}{p}(2p+1)(2(n-p-1)+1)(2p-1)!!(2(n-p-1)-1)!!}_{(8)} \\
    &+ \sum_{p=0}^{n-2}\underbrace{4\binom{n}{2}\binom{n-2}{p}(2p+1)!!(2(n-p-2)+1)!!}_{(9)}
\end{split}
\end{equation}
The last sum should be taken to be $0$ when $n = 1$ and the sum is empty (this is because this case of course requires at least $n = 2$ to have two sets of type-2/3 edges). Each of these terms can be derived through a simple combinatorial argument regarding which types of edges are present and how they must be connected. For each case, say that there are $p$ type-1 sets of black edges. This means there are $n-p$, $n-p-1$, and $n-p-2$ sets of type-4 black edges for cases $(1)$-$(7)$, case $(8)$, and case $(9)$, respectively (in the latter two cases, the remaining set(s) of edges are type-2 and/or type-3). Each case then comes down to deciding how to order the sets of edges, how to choose which edges are connected together, and then pairing off the remaining edges of the same type to build the remaining 2- and 4-vertex connected components. We do not detail how to count every single case, but we discuss two examples, case $(1)$ and case $(8)$. The rest should be straightforward to derive by extending these arguments.

In case $(1)$, we merge two 2-vertex connected components of type 1. First, we have a factor of $\binom{n}{p}$ to account for all ways of having $p$ type-1 sets of edges. We then must select $2$ of the $p$ horizontal black edges to merge into a single connected component, hence the factor of $\binom{p}{2}$; see \cref{fig:c2n-1_1-7}. The additional factor of $2$ comes from the two possible ways of merging these into a single connected component. Finally, the remaining double factorial factors are the number of ways of pairing off the vertical black edges with those of the same type. We then must sum from $p = 0$ to $n$ to account for all possible black edge type distributions. 

Case $(8)$ proceeds similarly. First, we have a factor of $\binom{n}{1}$, or $n$, to choose where the type-2 or type-3 set of edges is. The factor of $2$ out front now actually accounts for whether it is type 2 or type 3. Next, we have $\binom{n-1}{p}$ to account for the placement of the $p$ type-1 sets of edges. Next, there are now $2p+1$ black edges that span the second and third rows (i.e., they are black edges that arise from type-1 sets of black edges). It is $2p+1$ because the type-2 or type-3 set of black edges contributes 1, and the $p$ type-1 sets contribute $2p$. Analogously, there are also $(2(n-p-1)+1)$ black edges spanning the first and second rows. We have to select one of each to connect to the black edge that spans the first and third rows to make a single 6-vertex connected component. The remaining factors are again the number of ways to pair off the remaining vertical black edges with those of the same type (horizontal black edges must form 2-vertex connected components to reach the required number of connected components). 

It is possible, but quite tedious, to simplify this double sum by looking at each individual term and then applying a similar technique as in the evaluation of the sum for $c_{2n}$. That is, for each term in the sum, we use the convolution of various Taylor series and compare the coefficients of $x^{n}$. We start with the first term
\begin{equation}
   (1) \to \sum_{p=0}^{n}\binom{n}{p} 2\binom{p}{2}(2p-1)!!(2(n-p)-1)!! = \frac{n!}{2^{n}} \sum_{p=0}^{n} p(p-1)\binom{2p}{p}\binom{2n-2p}{n-p}.
\end{equation}
One then has through Taylor expansion that
\begin{equation}
    x^{2}\frac{\mathrm{d}^{2}}{\mathrm{d}x^{2}}\frac{1}{\sqrt{1-4x}}=\sum_{n=0}^{\infty}\binom{2n}{n}n(n-1)x^{n},
\end{equation}
which implies that
\begin{equation}
    12x^{2}\frac{1}{(1-4x)^{3}}=\left(x^{2}\frac{\mathrm{d}^{2}}{\mathrm{d}x^{2}}\frac{1}{\sqrt{1-4x}}\right)\frac{1}{\sqrt{1-4x}}=\sum_{n=0}^{\infty}\sum_{p=0}^{n}p(p-1)\binom{2p}{p}\binom{2n-2p}{n-p}x^{n}.
\end{equation}
Using the Online Encycopledia of Integer Sequences (OEIS), we find the three-fold convolution of powers of 4 \href{https://oeis.org/A038845 }{A038845} \cite{oeis} has formula $(n+2)(n+1)2^{2n-1}$, meaning
\begin{equation}
    \sum_{n=0}^{\infty}12(n+2)(n+1)2^{2n-1}x^{n+2}=12x^{2}\frac{1}{(1-4x)^{3}}=\sum_{n=0}^{\infty}\sum_{p=0}^{n}p(p-1)\binom{2p}{p}\binom{2n-2p}{n-p}x^{n}.
\end{equation}
Therefore, comparing powers of $x$, we get that
\begin{equation}
    \sum_{p=0}^{n}p(p-1)\binom{2p}{p}\binom{2n-2p}{n-p} = 12n(n-1)2^{2n-5},
\end{equation}
meaning the first term in the sum is (after some algebra)
\begin{equation}
    \sum_{p=0}^{n}\binom{n}{p} 2\binom{p}{2}(2p-1)!!(2(n-p)-1)!! = (2n)!! \frac{3n(n-1)}{8}.
\end{equation}
Note also by the symmetry between $p$ and $n-p$, the contribution of the second term is the same. 

We can perform similar manipulations for the other terms. In particular, 
\begin{align}
    (3) \to \sum_{p=0}^{n}\binom{n}{p}2p\binom{2n-2p}{2}(2p-1)!!(2(n-p-1)-1)!! = \frac{n!}{2^{n-1}}\sum_{p=0}^{n}(n-p)p\binom{2p}{p}\binom{2n-2p}{n-p}. 
\end{align}
Instead of taking the Taylor expansion for the second derivative of $(1-4x)^{-1/2}$ and convolving it with that for $(1-4x)^{-1/2}$, we convolve the Taylor series for the first derivative with itself. That is,
\begin{equation}
    \frac{4x^{2}}{(1-4x)^{3}}=\left(x\frac{\mathrm{d}}{\mathrm{d}x}\frac{1}{\sqrt{1-4x}}\right)^{2} = \sum_{n=0}^{\infty}\sum_{p=0}^{n} p(n-p)\binom{2p}{p}\binom{2n-2p}{n-p}x^{n},
\end{equation}
which, using the same result as for $(1)$ (just with a difference of a factor of $3$), yields
\begin{equation}
    \sum_{p=0}^{n} p(n-p)\binom{2p}{p}\binom{2n-2p}{n-p} = 4n(n-1)2^{2n-5}.
\end{equation}
This means that the third term yields a contribution of
\begin{equation}
    (3) \to \frac{n!}{2^{n-1}}4n(n-1)2^{2n-5} = (2n)!!\frac{2n(n-1)}{8}.
\end{equation}
Again, by the symmetry between $n$ and $n-p$, the contribution from the fourth term is the same. 

Next:
\begin{align}
    (5) \to \sum_{p=0}^{n}\binom{n}{p}6\binom{2p}{4}(2(p-2)-1)!!(2(n-p)-1)!! = \frac{n!}{2^{n}}\sum_{p=0}^{n}p(p-1)\binom{2p}{p}\binom{2n-2p}{n-p} = (2n)!!\frac{3n(n-1)}{8}
\end{align}
because this is the exact same as $(1)$. Again, by symmetry, $(6)$ has the same contribution. 

We also have that
\begin{align}
    (7) \to \sum_{p=0}^{n}\binom{n}{p}2\binom{2p}{2}\binom{2(n-p)}{2}(2(p-1)-1)!!(2(n-p-1)-1)!! = \frac{n!}{2^{n-1}} \sum_{p=0}^{n}p(n-p)\binom{2p}{p}\binom{2n-2p}{n-p} = (2n)!!\frac{2n(n-1)}{8},
\end{align}
which follows because this term happens to be the same as $(3)$. 

We now move on to the final two cases. Again, similar manipulations yield that
\begin{align}
\begin{split}
    (8) &\to \sum_{p=0}^{n-1}2\binom{n}{1}\binom{n-1}{p}(2p+1)(2(n-p-1)+1)(2p-1)!!(2(n-p-1)-1)!! \\
    &= \frac{n!}{2^{n-1}}\sum_{p=0}^{n}(2p+1)(n-p)\binom{2p}{p}\binom{2n-2p}{n-p} \\
    &= \frac{n!}{2^{n-1}}2\sum_{p=0}^{n}p(n-p)\binom{2p}{p}\binom{2n-2p}{n-p} + \frac{n!}{2^{n-1}}\sum_{p=0}^{n}(n-p)\binom{2p}{p}\binom{2n-2p}{n-p}.
\end{split}
\end{align}
We have expanded the upper limit to $p = n$ because the factor of $n-p$ sets this additional contribution to $0$. The first term in the last equation is simply twice the contribution of (3), which is $(2n)!!4n(n-1)/8$. The second term requires yet another manipulation of Taylor series. By very similar arguments to the above, we have that
\begin{equation}
    x\frac{\mathrm{d}}{\mathrm{d}x}\frac{1}{\sqrt{1-4x}}=\sum_{n=0}^{\infty}\binom{2n}{n}nx^{n},
\end{equation}
which implies that
\begin{equation}
    \frac{2x}{(1-4x)^{2}} = \left(x\frac{\mathrm{d}}{\mathrm{d}x}\frac{1}{\sqrt{1-4x}}\right)\frac{1}{\sqrt{1-4x}}=\sum_{n=0}^{\infty}\sum_{p=0}^{n}p\binom{2p}{p}\binom{2n-2p}{n-p}x^{n},
\end{equation}
which is the same as the sum we are interested in (up to the symmetry of replacing $n-p$ with $p$). Using OEIS sequence \href{https://oeis.org/A002697}{A002697} \cite{oeis}, that is, the convolution of powers of $4$, we find that 
\begin{equation}
    \frac{2x}{(1-4x)^{2}} = \sum_{n=0}^{\infty}2(n+1)4^{n}x^{n+1},
\end{equation}
which means that, comparing powers of $x^{n}$,
\begin{equation}
    \frac{n!}{2^{n-1}}\sum_{p=0}^{n}(n-p)\binom{2p}{p}\binom{2n-2p}{n-p} = \frac{n!}{2^{n-1}} 2 n4^{n-1} = n(2n)!!.
\end{equation}
Finally, then 
\begin{equation}
    (8) \to  (2n)!!\frac{4n(n-1)}{8}  + (2n)!!n = \frac{4n(n+1)}{8}. 
\end{equation}

Last, we get that
\begin{align}
\begin{split}
    (9) \to &\sum_{p=0}^{n-2}4\binom{n}{2}\binom{n-2}{p}(2p+1)!!(2(n-p-2)+1)!! \\
    &= \frac{n!}{2^{n-1}}\sum_{p=0}^{n-2}\binom{2p+2}{p+1}\binom{2n-2p-2}{n-p-1}(p+1)(n-p-1)\\
    &= \frac{n!}{2^{n-1}}\sum_{x=0}^{n}\binom{2x}{x}\binom{2n-2x}{n-x}(x)(n-x),
\end{split}
\end{align}
where we have set $x = p+1$ and then expanded the limits of summation to include $x = 0$ and $x = n$ (because these terms contribute 0). Therefore, this contribution is the same as $(3)$, $(4)$, and $(7)$, which is $(2n)!!2n(n-1)/8$.

Therefore, in total, we have that
\begin{equation}
    \frac{c_{2n-1}}{(2n)!!} = 4 \frac{3n(n-1)}{8} + 4 \frac{2n(n-1)}{8} + \frac{4n(n-1)}{8} + n = (3n-2)n.
\end{equation}
Therefore, 
\begin{equation}
    c_{2n-1} = (2n)!!(3n-2)n.
\end{equation}

Numerically evaluating the sums yields the same value up to $n=40$, and this also matches the value of $c_{2n-1}$ computed via the recursion. We note that $(3n-2)n$ are the so-called octagonal numbers, which are OEIS entry \href{https://oeis.org/A000567}{A000567} \cite{oeis}. However, we are not sure whether there is a deeper connection between these numbers and the graph theoretic problem at the core of this calculation. Additionally, while it is nice that we have been able to find an exact formula for a second coefficient, this calculation does not seem scalable, meaning other methods are likely needed to try to find the full expansion of the second moment.

\section{Alternative method for computing coefficients \texorpdfstring{$c_i$}{ci}}\label{app:alternative}

In this section, we present an alternative method for computing coefficients $c_i$ in
\begin{equation} \label{eq:M2alexey}
    M_{2}(k,n) = (2n-1)!!\sum_{i=1}^{2n} c_{i}k^{i}.
\end{equation}

Using this method, we obtain a useful expression for $c_1$. We also outline how this method can be used to set up an alternative recursive code for computing the coefficients $c_i$ for all $i$. While we have not implemented this code, there is a possibility it is more efficient than the recursive code discussed in the main text. It is also possible that this new method may yield other useful analytical results about  $c_i$, including their asymptotic behavior.

We start by recalling Eq.\ (\ref{eqn:second_moment_graph}):
\begin{equation}
    M_2(k,n) = (2n-1)!!\sum_{G\in\mathbb{G}_{n}^{2}}k^{C(G)},
\end{equation}
where the sum goes over all graphs possessing the allowed assignments of black and red edges. The new method relies on the following key simplifying observation: for a given fixed assignment of black edges, the contribution to $M_2(k,n)$ (summed over all allowed red edge assignments) depends only on $\vec e = (e_{11}, e_{12}, e_{13},e_{23},e_{33})$, where $e_{ij}$ is the number of black edges that connect row $i$ to row $j$. In particular, the answer does not depend on what columns the black edges are connecting. The proof of this key observation is simple: for a fixed set of black edges, the contribution to $M_2(k,n)$ is summed over all possible red perfect matchings in each of the three rows. This means that we can swap any two vertices in a given row (while pulling the ends of the black edges to the new destinations) without changing the answer. This completes the proof.

Let $p_1$ be the number of type-1 sets of black edges, $p_4$ be the number of type-4 sets of black edges, and $p$ be the combined number of type-2 and type-3 sets of black edges (type-2 and type-3 sets are equivalent as far as their contributions to $e_{ij}$). Then $e_{11} = p_1$, $e_{33} = p_4$, $e_{12} = p + 2 p_4$, $e_{23} = p + 2 p_1$, and $e_{13} = p$. We then write
\ba
M_2(k,n) = \sum_{p_1 = 0}^n \sum_{p_4 = 0}^{n-p_1} {\binom{n}{p_1}} {\binom{n-p_1}{p_4}} 2^p g(\vec e),
\ea
where $p = n - p_1 - p_4$. The combinatorial factors come from choosing $p_1$ sets of type-1 black edges out of $n$ possible locations, then $p_4$ sets of type-4 black edges from $n-p_1$ possible locations, and finally multiplying by a factor of $2$ for each choice of whether a given contribution to $p$ is type-2 or type-3. Additionally,
\ba
g(\vec e) = \sum_{i = 1}^{2 n} d_i(\vec e) k^i,
\ea
where $d_i(\vec e)$ is the number of ways (using the allowed red-edge assignments) to make $i$ loops given the black edges specified by $\vec e$.

The coefficients $d_i(\vec e)$ can then be computed with the help of the visualization shown in Fig.\ \ref{fig:newmethod}(a). The three black dots labeled $1$, $2$, and $3$ represent the three rows. The numbers $e_{jk}$ on the five edges (including the two loops) show how many black edges connect row $j$ to row $k$. Roughly speaking, the coefficient $d_i(\vec e)$ is the number of ways to connect all the black edges specified by $\vec e$ into exactly $i$ loops.   
The red edges are used to connect the black edges to each other and are taken into account automatically, which is one of the key advantages of this approach (slightly more specifically, for any two black edges that share a row, it is possible to connect them with a red edge between the vertices in that shared row). 
Each way of joining the edges $\vec e$ into loops also comes with a combinatorial factor that takes into account the fact that all edges are distinguishable and the fact that edges that stay in the same row can each be traversed in one of two directions. 
\begin{figure}[ht!]
    \centering
    \includegraphics[width=0.7\linewidth]{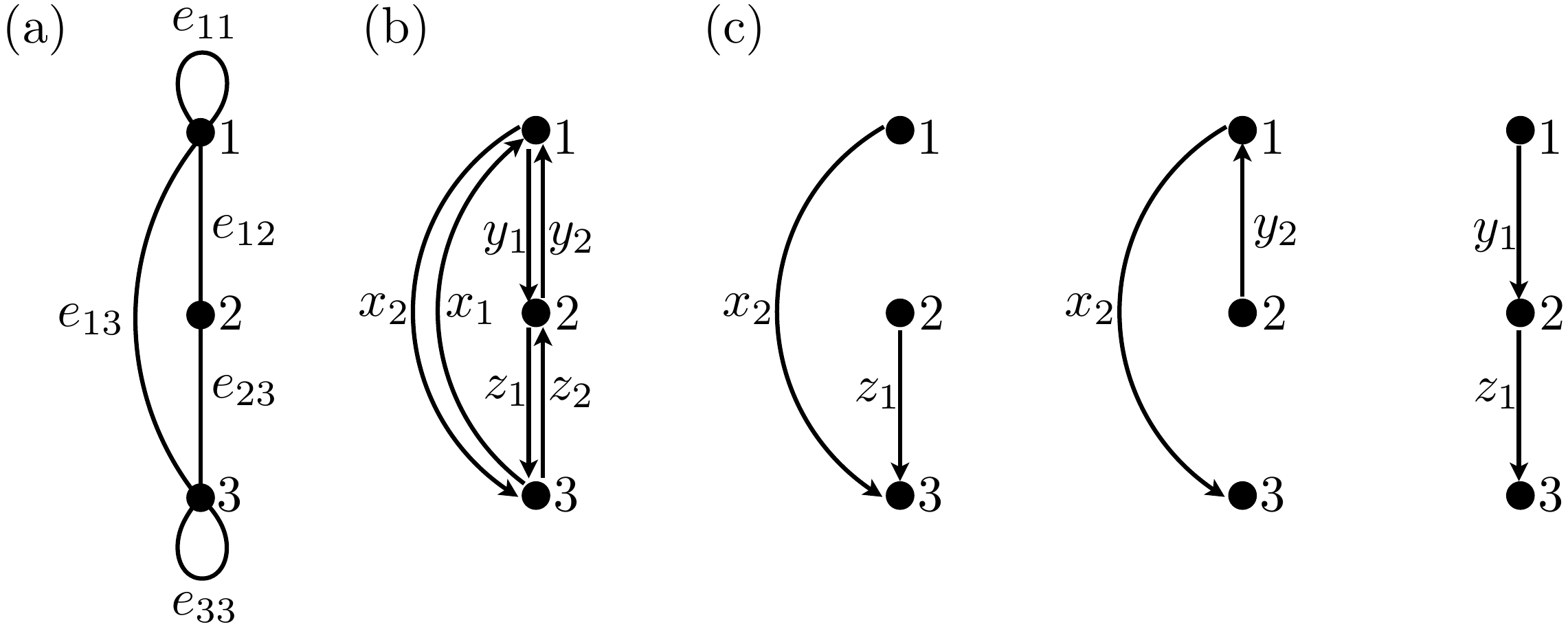}
    \caption{Graphs useful for understanding the new method for calculating coefficients $c_{i}$. (a) Once the types of black edges are assigned, the contribution to $M_2(k,n)$ depends only on the number of black $e_{ij}$ edges connecting row $i$ to row $j$. (b) In order to compute $c_1$, the number of single-loop graphs contributing to $M_2$, we first set $e_{11} = e_{33} = 0$ (later adding in the effect of nonzero values), fix the winding number $w$ of the loop, and break up $1$-$3$ edges into $x_1 = (e_{13} + w)/2$ clockwise edges and $x_2 = (e_{13} - w)/2$ counterclockwise edges. We similarly break up the $1$-$2$ and $2$-$3$ edges. We then use the BEST theorem \cite{vanaardenne-ehrenfest51} to count the number of Eulerian circuits on this directed graph. (c) For $u = 3$, the three types of arborescences contributing to $t_u(G) = x_2 z_1 + y_2 x_2  + y_1 z_1$ for the graph $G$ shown in (b).}
    \label{fig:newmethod}
\end{figure}

The coefficients $g(\vec e)$ can be computed using a recursive procedure. Instead of doing a recursion on $n$ (which is what we do in the main text, with details presented in \cref{app:recursion}), we perform the recursion on the number of black edges $e_{11} + e_{12} + e_{13} + e_{23} + e_{33}$. As in the main text, we need to define a more general function $g(\vec e, \sigma, c, s)$ to make the recursion work. $\sigma$ is a binary variable, so that $\sigma = 1$ means we are in the process of building a loop, while $\sigma = 0$ means that we need to start a new loop. If $\sigma = 1$, we need to also specify $s \in \{1,2,3\}$ (standing for start) indicating the row where the current loop started and $c \in \{1,2,3\}$ (standing for current) indicating the row where we currently are.

As in the main text, the recursive procedure is efficient, i.e.\ takes polynomial time in the number of edges. We first directly compute $g(\vec e, \sigma, c, s)$ for small values of $e_{11} + e_{12} + e_{13} + e_{23} + e_{33}$. Then the recursive step goes as follows. If $\sigma = 0$, we can either (1) close the loop right away by reducing $e_{11}$ or $e_{33}$ by 1, keep $\sigma = 0$, and multiply by $k$, or (2) set $\sigma = 1$, start a new loop at row $i$, set $s = i$, reduce $e_{ij}$ by 1 (for some $j$), and set $c = j$. If $\sigma = 1$, we can either (1) close the loop by reducing $e_{cs}$ by 1, set $\sigma = 0$, and multiply by $k$, or (2) continue building the loop, keep $\sigma = 1$, keep $s$ unchanged, reduce $e_{cj}$ by 1 (for some $j$), and change the value of $c$ to $j$. As we do these calculations, we need to also include appropriate combinatorial factors deciding which black edge to take (e.g.,\ if we pick one of $e_{ij}$ edges, we need to multiply by $e_{ij}$, and if $i = j$, we need to multiply by another factor of 2).

While we have not coded up this procedure, we believe that it offers another complementary way of understanding and analyzing the second moment. 

\subsection{Computing \texorpdfstring{$c_1$}{c1}}

Again, while we have not coded up the above recursive procedure, we will show how to use the new approach to compute $c_1$ in Eq.\ (\ref{eq:M2alexey}), i.e.,\ the number of ways to build a single-loop graph, which we were not able to directly compute using the original method. 

To proceed, we will at first ignore the contributions of the edges $e_{11}$ and $e_{33}$ (effectively pretending that they are equal to zero), but we will address how to deal with them later on. We will also assign a direction to this single loop, and we will later divide the final answer by two because each loop will be counted twice (because there are two possible directions around a loop). While it may seem to make things more difficult to add directionality to a previously undirected graph, it will actually allow us to make use of known results.

To proceed, we sort the contributions to $c_1$ according to the winding number $w$ of the loop around the triangle formed by rows 1, 2, and 3, which can now be well defined because we have added directionality to the edges. Once $w$ is fixed, the total numbers of edges in the triangle of each directionality also become fixed. Specifically, as shown in Fig.\ \ref{fig:newmethod}(b), $x_1 = (e_{13} + w)/2$ is the number of $1$-$3$ edges traversed (i.e.,\ directed) from $3$ to $1$, $x_2 = (x-w)/2$ is the number of $1$-$3$ edges traversed from $1$ to $3$, $y_1 = (e_{12} + w)/2$ is the number $1$-$2$ edges traversed from $1$ to $2$, $y_2 = (e_{12}-w)/2$ is the number of $1$-$2$ edges traversed from $2$ to $1$, $z_1 = (e_{23}+w)/2$ is the number $2$-$3$ edges traversed from $2$ to $3$, and $z_2 = (e_{23} - w)/2$ is the number $2$-$3$ edges traversed from $3$ to $2$. Note that, here, we are treating a positive winding number as going \emph{clockwise} around the graph. There will also be combinatorial factors associated with \emph{which} edges go in which direction, but we will handle that factor later. We are now interested in the number of Eulerian circuits on the resulting directed graph $G$, i.e.,\ the number of directed closed paths that visit each edge exactly once. The BEST theorem \cite{vanaardenne-ehrenfest51} says that the number of such Eulerian circuits is
\ba
ec(G) = t_u(G) \prod_{v \in V} (\text{deg}(v) - 1)!,
\ea
where $V = \{1,2,3\}$ is the set of 3 vertices of our graph, $\text{deg}(v)$ is the indegree of vertex $v$, and $t_u(G)$ is the number of arborescences of $G$ with root $u$, i.e.,\ the number of directed tree subgraphs of $G$ such that, for any vertex $v$, there is exactly one directed path from $v$ to $u$. If graph the $G$ has an Eulerian circuit, it is known that $t_u(G)$ is independent of the choice of $u$. Choosing $u = 3$, the three types of trees (arborescences) contributing to $t_u(G)$ for the graph $G$ in Fig.\ \ref{fig:newmethod}(b) are shown in Fig.\ \ref{fig:newmethod}(c). The result is $t_u(G) = x_2 z_1 + y_2 x_2  + y_1 z_1 $. The term $x_2 z_1$ [corresponding to the first graph in Fig.\ \ref{fig:newmethod}(c)]  counts trees (arborescences) made up of a $1\rightarrow 3$ edge and a $2 \rightarrow 3$ edge; the term $y_2 x_2$ [corresponding to the second graph in Fig.\ \ref{fig:newmethod}(c)] counts trees made up of a $2 \rightarrow 1$ edge and a $1 \rightarrow 3$ edge; and the  term $y_1 z_1$ [corresponding to the third graph in Fig.\ \ref{fig:newmethod}(c)] counts trees made up of a $1 \rightarrow 2$ edge and a $2 \rightarrow 3$ edge. Plugging in the definitions of $x_i$, $y_i$, and $z_i$, we find $t_u(G) = (w^2 + e_{12} e_{13} + e_{13} e_{23} + e_{23} e_{12})/4$. Therefore,
\ba
ec(G) &=& \frac{1}{4}(w^2 + e_{12} e_{13} + e_{13} e_{23} + e_{23} e_{12}) \left(\frac{e_{12} + e_{13}}{2}-1\right)!\left(\frac{e_{12} + e_{23}}{2}-1\right)!\left(\frac{e_{13} + e_{23}}{2}-1\right)! 
\nonumber \\
&=& \frac{1}{4}\left(w^2 + 3 n^2 - (p_1 - p_4)^2 - 2 n  (p_1 + p_4)\right) (n-p_1-1)!(n-1)!(n-p_4-1)!.
\ea
We now include the aformentioned combinatorial factors that account for which edges receive which directionality. When choosing which $x_1$ of the $e_{13}$ edges to make into $3 \rightarrow 1$ edges, we pick up a combinatorial factor of ${\binom{e_{13}}{x_1}} = {\binom{n - p_1 - p_4}{(n-p_1 - p_4 +w)/2}}$. Similarly for $e_{12}$ and $e_{23}$: ${\binom{e_{12}}{y_1}} = {\binom{n - p_1 + p_4}{(n-p_1 + p_4 +w)/2}}$ and ${\binom{e_{23}}{z_1}} = {\binom{n + p_1 - p_4}{(n+p_1 - p_4 +w)/2}}$. 

We can now also account for the fact that $e_{11}$ and $e_{33}$ may actually be nonzero. We keep $G$ defined as before (i.e.\ using only $1$-$2$, $1$-$3$, and $2$-$3$ edges), but we now dress the loops defined on $G$ (and counted above) with additional $1$-$1$ and $3$-$3$ edges. The number of times our loop visits vertex $1$ is given by $\deg(1) = n-p_1$, so we need to sort $e_{11} = p_1$ edges into $n-p_1$ buckets, which gives a factor of ${\binom{e_{11} + n-p_1-1}{e_{11}}} = {\binom{n-1}{p_1}}$ (by the standard ``stars and bars'' argument). Similarly, $e_{33} = p_4$ loops give ${\binom{e_{33} + n-p_4-1}{e_{33}}}={\binom{n-1}{p_4}}$. Because all $e_{11} = p_1$ edges are distinguishable and can be traversed in two different ways, we also get a factor of $p_1! 2^{p_1}$ (that is, after the bucket counts are decided, we still have to order the edges and assign each a direction). We similarly get a factor of $p_4! 2^{p_4}$. Putting all these elements together, we have
\ba \label{eq:c1}
c_1 &=& \sum_{p_1 = 0}^n \sum_{p_4 = 0}^{n-p_1} {\binom{n}{p_1}} {\binom{n-p_1}{p_4}} 2^p d_1(\vec e) \nonumber \\
&=& \sum_{p_1 = 0}^{n-1} \sum_{p_4 = 0}^{n-\text{max}(p_1,1)} {\binom{n}{p_1}} {\binom{n-p_1}{p_4}} 2^p \sum_w  \frac{1}{8}\left(w^2 + 3 n^2 - (p_1 - p_4)^2 - 2 n  (p_1 + p_4)\right) (n-p_1-1)!(n-1)!(n-p_4-1)! \nonumber \\
&& \times {\binom{n - p_1 - p_4}{(n-p_1 - p_4 +w)/2}} {\binom{n - p_1 + p_4}{(n-p_1 + p_4 +w)/2}} {\binom{n + p_1 - p_4}{(n+p_1 - p_4 +w)/2}} {\binom{n-1}{p_1}} {\binom{n-1} {p_4}}  p_1! 2^{p_1} p_4! 2^{p_4} \\
&=&  n! ((n-1)!)^3 2^{n-3}\sum_{p_1 = 0}^n \sum_{p_4 = 0}^{n-p_1} \sum_{w = -n + p_1 + p_4}^{n-p_1 - p_4} \frac{{\binom{n - p_1 + p_4}{(n - p_1 + p_4 + w)/2}} {\binom{n + p_1 - p_4}{(n + p_1 - p_4 + w)/2}}  (w^2 + 3 n^2 - (p_1 - p_4)^2 -2 n (p_1 + p_4))}{p_1 ! p_4 ! ((n - p_1 - p_4 - w)/2)! ((n - p_1 - p_4 + w)/2)!}  \nonumber.
\ea
In the second line, we have introduced an extra factor of $1/2$ because we counted every loop twice because of the two directions in which each loop can be traversed. In the second line, we also excluded the cases where all black edge sets are of type-1 ($p_1 = n$) and where all black edge sets are of type-4 ($p_4 = n$), as there is no single-loop contribution in this case (allowing for $p_1 = n$ would make ${\binom{n-1}{p_1}}$ undefined; similarly for $p_4 = n$). In the last line, to simplify the expression, we allow $p_1 = n$ and $p_4 = n$ because the corresponding contribution is now well-defined and vanishes anyway. In the last line, the sum over $w$ runs in increments of $2$ due to a parity constraint (flipping the directionality of a single edge actually changes the winding number by $2$). While one can evaluate the sum over $w$ in the final expression in Eq.\ (\ref{eq:c1}) in terms of hypergeometric functions, we were not able to then evaluate the remaining sums over $p_4$ and $p_1$ to obtain a closed-form expression for $c_1$.

Numerical evaluation of the final expression in Eq.\ (\ref{eq:c1}) agrees with the evaluation of $c_1$ using the recursive method in the main text up to $n=40$ (which is the largest $n$ we apply the latter method to). The final expression in Eq.\ (\ref{eq:c1}) is, however, so simple that it can easily be evaluated for much larger values of $n$. For example, Mathematica~\cite{Mathematica} on a personal computer evaluates it for $n=200$ in about 15 seconds. One can also use Eq.\ (\ref{eq:c1}) to study in detail the asymptotic dependence of $c_1$ on $n$. We also hope that the method introduced in this section can yield other useful analytical results about $c_i$.

\end{appendix}

\end{document}